\documentclass[reqno, eucal]{amsart}
\usepackage{amsmath}
\usepackage{amssymb}
\usepackage{amsfonts}
\usepackage{amsthm}
\numberwithin{equation}{section}
\allowdisplaybreaks[4]
\usepackage{cite}
\usepackage{graphicx}
\usepackage{hyperref}
\usepackage{enumerate}
\usepackage{here}
\usepackage{multirow}
\usepackage{caption}
\usepackage[mathscr]{eucal}
\usepackage{braket}
\usepackage[dvipdfmx]{color}
\usepackage{tikz-cd}
\usepackage{etoolbox}
\makeatletter
\patchcmd{\thmhead@plain}%
   {\thmnote{ {\the\thm@notefont(#3)}}}
   {\thmnote{  #3}}
   {}{}
\let\thmhead\thmhead@plain
\patchcmd{\swappedhead}%
   {\thmnote{ {\the\thm@notefont(#3)}}}
   {\thmnote{  #3}}
   {}{}
\let\swappedhead@plain=\swappedhead
\makeatother
\newtheorem{theorem}{Theorem}[section]
\newtheorem{lemma}[theorem]{Lemma}
\newtheorem{conjecture}[theorem]{Conjecture}
\newtheorem{proposition}[theorem]{Proposition}
\newtheorem{corollary}[theorem]{Corollary}

\theoremstyle{definition}
\newtheorem{definition}[theorem]{Definition}
\newtheorem{example}[theorem]{Example}
\newtheorem{remark}[theorem]{Remark}
\newcommand{\dd}{\delta}
\newcommand{\ee}{\epsilon}
\newcommand{\ttt}{\theta}
\newcommand{\BB}{\bar{\epsilon}}
\newcommand{\EE}{\mathscr{E}}
\newcommand{\asr}[1][m|n]{\EE(#1)_{\mathbb{R}}}
\newcommand{\asz}[1][m|n]{\EE(#1)_{\mathbb{Z}}}
\newcommand{\aaa}{\alpha}
\newcommand{\pmr}{\Phi^{\pm}}
\newcommand{\pr}{\Phi^{+}}
\newcommand{\rr}{\tilde{\Phi}}
\newcommand{\prr}{\tilde{\Phi}^{+}}
\newcommand{\pmrr}{\tilde{\Phi}^{\pm}}
\newcommand{\re}{\Phi_{\mathrm{even}}}
\newcommand{\ri}{\Phi_{\mathrm{iso}}}
\newcommand{\ra}{\Phi_{\mathrm{aniso}}}
\newcommand{\pmre}{\Phi^{\pm}_{\mathrm{even}}}
\newcommand{\pmri}{\Phi^{\pm}_{\mathrm{iso}}}
\newcommand{\pmra}{\Phi^{\pm}_{\mathrm{aniso}}}
\newcommand{\pre}{\Phi^{+}_{\mathrm{even}}}

\newcommand{\pmrer}{\tilde{\Phi}^{\pm}_{\mathrm{even}}}
\newcommand{\pmrir}{\tilde{\Phi}^{\pm}_{\mathrm{iso}}}
\newcommand{\pmrar}{\tilde{\Phi}^{\pm}_{\mathrm{aniso}}}
\newcommand{\prer}{\tilde{\Phi}^{+}_{\mathrm{even}}}
\newcommand{\prir}{\tilde{\Phi}^{+}_{\mathrm{iso}}}
\newcommand{\prar}{\tilde{\Phi}^{+}_{\mathrm{aniso}}}

\newcommand{\tA}[1][m|n]{\mathfrak{sl}(#1)}
\newcommand{\tB}[1][2m+1|2n]{\mathfrak{osp}(#1)}
\newcommand{\GG}[1][m|n]{\mathfrak{g}(#1)}
\newcommand{\tAn}[1][m]{\mathfrak{sl}(#1)}
\newcommand{\tBn}[1][2m+1]{\mathfrak{so}(#1)}
\newcommand{\GGn}[1][m]{\mathfrak{g}(#1)}
\newcommand{\psign}{\pi}
\newcommand{\UGG}[1][m|n]{U_q(\GG[#1])}
\newcommand{\UGGp}[1][m|n]{U_q^{+}(\GG[#1])}
\newcommand{\UGGn}[1][m]{U_q(\GG[#1])}
\newcommand{\UGGpn}[1][m]{U_q^{+}(\GG[#1])}

\newcommand{\UtAp}[1][m|n]{U_q^{+}(\mathfrak{sl}(#1))}
\newcommand{\UtBp}[1][2m+1|2n]{U_q^{+}(\mathfrak{osp}(#1))}
\newcommand{\UtBpn}[1][2m+1]{U_q^{+}(\mathfrak{so}(#1))}

\newcommand{\phaseIrel}{\phi}
\newcommand{\phaseIIrel}{\varphi}
\newcommand{\WG}[1][m|n]{W(\GG[#1])}
\newcommand{\WGn}[1][m]{W(\GGn[#1])}
\newcommand{\WSG}[1][m|n]{SW(\GG[#1])}
\newcommand{\htf}{\mathrm{ht}}
\newcommand{\tm}{\gamma}
\newcommand{\tmt}{\tilde{\gamma}}
\newcommand{\tdr}{\mathscr{R}}
\newcommand{\tdl}{\mathscr{L}}
\newcommand{\tdlt}{\tilde{\mathscr{L}}}
\newcommand{\tdlq}{\tdl(q)}
\newcommand{\tdm}{\mathscr{M}}
\newcommand{\tdmq}{\tdm(q)}
\newcommand{\tdn}{\mathscr{N}}
\newcommand{\tdk}{\mathscr{K}}
\newcommand{\tdj}{\mathscr{J}}
\newcommand{\tdg}{\mathscr{G}}
\newcommand{\tdx}{\mathscr{X}}
\newcommand{\tdy}{\mathscr{Y}}
\newcommand{\tdz}{\mathscr{Z}}
\newcommand{\tdrC}{\mathcal{R}}
\newcommand{\tdlC}{\mathcal{L}}
\newcommand{\tdmC}{\mathcal{M}}
\newcommand{\tdnC}{\mathcal{N}}
\newcommand{\tdjC}{\mathcal{J}}
\newcommand{\tdxC}{\mathcal{X}}
\newcommand{\tdyC}{\mathcal{Y}}
\newcommand{\tdzC}{\mathcal{Z}}
\newcommand{\ronri}{\theta}
\newcommand{\tdjPhaseI}{\psi_1}
\newcommand{\tdjPhaseII}{\psi_2}
\newcommand{\appGa}{X}
\newcommand{\appGb}{Z}
\newcommand{\boson}{\mathcal{B}_q}
\newcommand{\ii}{1}
\newcommand{\jj}{2}
\newcommand{\kk}{3}
\newcommand{\oop}{e_{\ii\jj}}
\newcommand{\poo}{e_{\jj\kk}}
\newcommand{\pot}{e_{(\jj\kk)\kk}}
\newcommand{\ooo}{e_{\ii\jj\kk}}
\newcommand{\oot}{e_{((\ii\jj)\kk)\kk}}
\newcommand{\ott}{e_{(((\ii\jj)\kk)\kk)\jj}}
\newcommand{\toX}{e_{\jj\ii}}
\newcommand{\rtX}{e_{\kk\jj}}

\newcommand{\ltrlo}{e_{(\jj\kk)\ii}}
\newcommand{\ltolr}{e_{(\jj\ii)\kk}}
\newcommand{\rltol}{e_{\kk\jj\ii}}
\newcommand{\rlrtl}{e_{\kk(\kk\jj)}}
\newcommand{\olltrlrl}{e_{\ii((\jj\kk)\kk)}}
\newcommand{\lltrlrlo}{e_{((\jj\kk)\kk)\ii}}
\newcommand{\lltolrlr}{e_{((\jj\ii)\kk)\kk}}
\newcommand{\rlrltoll}{e_{\kk(\kk(\jj\ii))}}
\newcommand{\tlllotlrlrl}{e_{\jj(((\ii\jj)\kk)\kk)}}
\newcommand{\loltrllltrl}{e_{(\ii(\jj\kk))(\jj\kk)}}
\newcommand{\ltrllltrlol}{e_{(\jj\kk)((\jj\kk)\ii)}}
\newcommand{\ltollltrlrl}{e_{(\jj\ii)((\jj\kk)\kk)}}
\newcommand{\lltrlrlltol}{e_{((\jj\kk)\kk)(\jj\ii)}}
\newcommand{\tlrlrltolll}{e_{\jj(\kk(\kk(\jj\ii)))}}
\newcommand{\rlotl}{e_{\kk(\ii\jj)}}
\newcommand{\lrtlo}{e_{\kk\jj\ii}}
\newcommand{\rlrlotll}{e_{\kk(\kk(\ii\jj))}}
\newcommand{\lrlrtllo}{e_{(\kk(\kk\jj))\ii}}

\newcommand{\lrlrtlllotl}{e_{(\kk(\kk\jj))(\ii\jj)}}
\newcommand{\lrtlllrtlol}{e_{(\kk\jj)((\kk\jj)\ii)}}
\newcommand{\tllrlrtllol}{e_{\jj((\kk(\kk\jj))\ii)}}
\newcommand{\stm}{\Gamma}
\newcommand{\stmB}{\Theta}
\newcommand{\stmC}{\Xi}
\newcommand{\phaseI}{\rho_1}
\newcommand{\phaseII}{\rho_2}
\newcommand{\phaseIII}{\rho_3}
\newcommand{\phaseIb}{\eta_1}
\newcommand{\phaseIIb}{\eta_2}
\newcommand{\phaseIIIb}{\eta_3}
\newcommand{\spaceD}{\hspace{3em}}
\newcommand{\spaceDb}{\hspace{3em}}
\newcommand{\geneBoe}{\tilde{e}}
\newcommand{\cc}{\tikz[baseline=-0.5ex]\draw[radius=1ex] (0,0) circle;}
\newcommand{\ccb}{\tikz[baseline=-0.5ex]\draw[black, fill=black, radius=1ex] (0,0) circle;}
\newcommand{\cct}{%
\begin{tikzpicture}[baseline=-0.5ex]%
\draw[radius=1ex] (0,0) circle;%
\foreach \i in {0,...,3}{%
\coordinate (N\i) at (\i*360/4++45:1ex);%
}%
\draw (N0) -- (N2);%
\draw (N1) -- (N3);%
\end{tikzpicture}%
}
\newcommand{\cctn}{%
\begin{tikzpicture}[baseline=-0.5ex]%
\foreach \i in {0,...,3}{%
\coordinate (N\i) at (\i*360/4++45:1ex);%
}%
\draw (N0) -- (N2);%
\draw (N1) -- (N3);%
\end{tikzpicture}%
}
\newcommand{\ltA}{dash}
\newcommand{\ltB}{Rightarrow}
\newcommand{\ddtwo}[3]{%
\begin{tikzcd}[cramped, sep=14pt, ampersand replacement=\&]%
#2 \arrow[#1, start anchor={[xshift=-2pt]}, end anchor={[xshift=2pt]}]{r} \& #3%
\end{tikzcd}%
}
\newcommand{\ddtwoA}[2]{\ddtwo{\ltA}{#1}{#2}}
\newcommand{\ddtwoB}[2]{\ddtwo{\ltB}{#1}{#2}}
\newcommand{\ddtwoWul}[5]{%
\begin{tikzcd}[cramped, sep=14pt, ampersand replacement=\&, row sep=0ex, baseline=-0.5ex]%
#4 \& #5 \\%
#2 \arrow[#1, start anchor={[xshift=-2pt]}, end anchor={[xshift=2pt]}]{r} \& #3%
\end{tikzcd}%
}
\newcommand{\ddthree}[5]{%
\begin{tikzcd}[cramped, sep=14pt, ampersand replacement=\&, row sep=0ex, baseline=-0.5ex]%
#3 \arrow[#1, start anchor={[xshift=-2pt]}, end anchor={[xshift=2pt]}]{r} \& #4%
\arrow[#2, start anchor={[xshift=-2pt]}, end anchor={[xshift=2pt]}]{r} \& #5
\end{tikzcd}%
}
\newcommand{\ddthreeA}[3]{\ddthree{\ltA}{\ltA}{#1}{#2}{#3}}
\newcommand{\ddthreeB}[3]{\ddthree{\ltA}{\ltB}{#1}{#2}{#3}}
\newcommand{\ddthreeWul}[8]{%
\begin{tikzcd}[cramped, sep=14pt, ampersand replacement=\&, row sep=0ex, baseline=-0.5ex]%
#6 \& #7 \& #8\\%
#3 \arrow[#1, start anchor={[xshift=-2pt]}, end anchor={[xshift=2pt]}]{r} \& #4%
\arrow[#2, start anchor={[xshift=-2pt]}, end anchor={[xshift=2pt]}]{r} \& #5
\end{tikzcd}%
}
\newcommand{\ddthreeAWul}[6]{\ddthreeWul{\ltA}{\ltA}{#1}{#2}{#3}{#4}{#5}{#6}}
\newcommand{\ddthreeBWul}[6]{\ddthreeWul{\ltA}{\ltB}{#1}{#2}{#3}{#4}{#5}{#6}}
\newcommand{\ddgenA}{%
\begin{tikzcd}[sep=small, ampersand replacement=\&, row sep=0ex, baseline=-0.5ex, column sep=1ex]%
\BB_{1}-\BB_{2} \& \BB_{2}-\BB_{3} \& \qquad \& \quad \& \BB_{N-1}-\BB_{N} \\%
\cctn \arrow[\ltA, start anchor={[xshift=-2pt]}, end anchor={[xshift=2pt]}]{r} \& \cctn%
\arrow[\ltA, start anchor={[xshift=-2pt]}, end anchor={[xshift=2pt]}]{r} \& {\ }%
\arrow[\ltA, dotted, start anchor={[xshift=-2pt]}, end anchor={[xshift=2pt]}]{r} \& {\ }%
\arrow[\ltA, start anchor={[xshift=-2pt]}, end anchor={[xshift=2pt]}]{r} \& \cctn%
\end{tikzcd}%
}
\newcommand{\ddgenArightmost}{%
\begin{tikzcd}[sep=small, ampersand replacement=\&, row sep=0ex, baseline=-0.5ex, column sep=1ex]%
\cc \arrow[\ltA, start anchor={[xshift=-2pt]}, end anchor={[xshift=2pt]}]{r} \& {\ }%
\arrow[\ltA, dotted, start anchor={[xshift=-2pt]}, end anchor={[xshift=2pt]}]{r} \& {\ }%
\arrow[\ltA, start anchor={[xshift=-2pt]}, end anchor={[xshift=2pt]}]{r} \& \cc%
\arrow[\ltA, start anchor={[xshift=-2pt]}, end anchor={[xshift=2pt]}]{r} \& \cct%
\end{tikzcd}%
}
\newcommand{\ddgenBe}{%
\begin{tikzcd}[sep=small, ampersand replacement=\&, row sep=0ex, baseline=-0.5ex, column sep=1ex]%
\BB_{1}-\BB_{2} \& \BB_{2}-\BB_{3} \& \qquad \& \quad \& \BB_{N-1}-\BB_{N} \& \BB_{N} \\%
\cctn \arrow[\ltA, start anchor={[xshift=-2pt]}, end anchor={[xshift=2pt]}]{r} \& \cctn%
\arrow[\ltA, start anchor={[xshift=-2pt]}, end anchor={[xshift=2pt]}]{r} \& {\ }%
\arrow[\ltA, dotted, start anchor={[xshift=-2pt]}, end anchor={[xshift=2pt]}]{r} \& {\ }%
\arrow[\ltA, start anchor={[xshift=-2pt]}, end anchor={[xshift=2pt]}]{r} \& \cctn%
\arrow[\ltB, start anchor={[xshift=-2pt]}, end anchor={[xshift=2pt]}]{r} \& \cc%
\end{tikzcd}%
}
\newcommand{\ddgenBa}{%
\begin{tikzcd}[sep=small, ampersand replacement=\&, row sep=0ex, baseline=-0.5ex, column sep=1ex]%
\BB_{1}-\BB_{2} \& \BB_{2}-\BB_{3} \& \qquad \& \quad \& \BB_{N-1}-\BB_{N} \& \BB_{N} \\%
\cctn \arrow[\ltA, start anchor={[xshift=-2pt]}, end anchor={[xshift=2pt]}]{r} \& \cctn%
\arrow[\ltA, start anchor={[xshift=-2pt]}, end anchor={[xshift=2pt]}]{r} \& {\ }%
\arrow[\ltA, dotted, start anchor={[xshift=-2pt]}, end anchor={[xshift=2pt]}]{r} \& {\ }%
\arrow[\ltA, start anchor={[xshift=-2pt]}, end anchor={[xshift=2pt]}]{r} \& \cctn%
\arrow[\ltB, start anchor={[xshift=-2pt]}, end anchor={[xshift=2pt]}]{r} \& \ccb%
\end{tikzcd}%
}
\newcommand{\ddgenAd}{%
\begin{tikzcd}[sep=small, ampersand replacement=\&, row sep=0ex, baseline=-0.5ex, column sep=0ex]%
\ee_{1}-\ee_{2} \& \quad \& \quad \& \ee_{m-1}-\ee_{m} \& \ee_{m}-\dd_{1} \& \dd_{1}-\dd_{2} \& \quad \& \quad \& \dd_{n-1}-\dd_{n} \\%
\cc \arrow[\ltA, start anchor={[xshift=-2pt]}, end anchor={[xshift=2pt]}]{r} \& {\ }%
\arrow[\ltA, dotted, start anchor={[xshift=-2pt]}, end anchor={[xshift=2pt]}]{r} \& {\ }%
\arrow[\ltA, start anchor={[xshift=-2pt]}, end anchor={[xshift=2pt]}]{r} \& \cc%
\arrow[\ltA, start anchor={[xshift=-2pt]}, end anchor={[xshift=2pt]}]{r} \& \cct%
\arrow[\ltA, start anchor={[xshift=-2pt]}, end anchor={[xshift=2pt]}]{r} \& \cc%
\arrow[\ltA, start anchor={[xshift=-2pt]}, end anchor={[xshift=2pt]}]{r} \& {\ }%
\arrow[\ltA, dotted, start anchor={[xshift=-2pt]}, end anchor={[xshift=2pt]}]{r} \& {\ }%
\arrow[\ltA, start anchor={[xshift=-2pt]}, end anchor={[xshift=2pt]}]{r} \& \cc%
\end{tikzcd}%
}
\newcommand{\ddgenBed}{%
\begin{tikzcd}[sep=small, ampersand replacement=\&, row sep=0ex, baseline=-0.5ex, column sep=0ex]%
\dd_{1}-\dd_{2} \& \quad \& \quad \& \dd_{n-1}-\dd_{n} \& \dd_{n}-\ee_{1} \& \ee_{1}-\ee_{2} \& \quad \& \quad \& \ee_{m-1}-\ee_{m} \& \ee_{m} \\%
\cc \arrow[\ltA, start anchor={[xshift=-2pt]}, end anchor={[xshift=2pt]}]{r} \& {\ }%
\arrow[\ltA, dotted, start anchor={[xshift=-2pt]}, end anchor={[xshift=2pt]}]{r} \& {\ }%
\arrow[\ltA, start anchor={[xshift=-2pt]}, end anchor={[xshift=2pt]}]{r} \& \cc%
\arrow[\ltA, start anchor={[xshift=-2pt]}, end anchor={[xshift=2pt]}]{r} \& \cct%
\arrow[\ltA, start anchor={[xshift=-2pt]}, end anchor={[xshift=2pt]}]{r} \& \cc%
\arrow[\ltA, start anchor={[xshift=-2pt]}, end anchor={[xshift=2pt]}]{r} \& {\ }%
\arrow[\ltA, dotted, start anchor={[xshift=-2pt]}, end anchor={[xshift=2pt]}]{r} \& {\ }%
\arrow[\ltA, start anchor={[xshift=-2pt]}, end anchor={[xshift=2pt]}]{r} \& \cc%
\arrow[\ltB, start anchor={[xshift=-2pt]}, end anchor={[xshift=2pt]}]{r} \& \cc%
\end{tikzcd}%
}
\newcommand{\ddgenBad}{%
\begin{tikzcd}[sep=small, ampersand replacement=\&, row sep=0ex, baseline=-0.5ex, column sep=1ex]%
\dd_{1}-\dd_{2} \& \dd_{2}-\dd_{3} \& \qquad \& \quad \& \dd_{n-1}-\dd_{n} \& \dd_{n} \\%
\cc \arrow[\ltA, start anchor={[xshift=-2pt]}, end anchor={[xshift=2pt]}]{r} \& \cc%
\arrow[\ltA, start anchor={[xshift=-2pt]}, end anchor={[xshift=2pt]}]{r} \& {\ }%
\arrow[\ltA, dotted, start anchor={[xshift=-2pt]}, end anchor={[xshift=2pt]}]{r} \& {\ }%
\arrow[\ltA, start anchor={[xshift=-2pt]}, end anchor={[xshift=2pt]}]{r} \& \cc%
\arrow[\ltB, start anchor={[xshift=-2pt]}, end anchor={[xshift=2pt]}]{r} \& \ccb%
\end{tikzcd}%
}%
\begin{document}
\title[Three dimensional integrability and PBW bases of quantum superalgebras]
{Tetrahedron and 3D Reflection Equation from PBW Bases of \\
the Nilpotent Subalgebra of Quantum Superalgebras}
\author{Akihito Yoneyama}
\address{Akihito Yoneyama, Institute of Physics,
University of Tokyo, Komaba, Tokyo 153-8902, Japan}
\email{yoneyama@gokutan.c.u-tokyo.ac.jp}
\maketitle
\vspace{0.5cm}
\begin{center}{\bf Abstract}\end{center}
In this paper, we study transition matrices of PBW bases of the nilpotent subalgebra of quantum superalgebras associated with all possible Dynkin diagrams of type A and B in the case of rank 2 and 3, and examine relationships with three-dimensional (3D) integrability.
We obtain new solutions to the Zamolodchikov tetrahedron equation via type A and the 3D reflection equation via type B, where the latter equation was proposed by Isaev and Kulish as a 3D analog of the reflection equation of Cherednik.
As a by-product of our approach, the Bazhanov-Sergeev solution to the Zamolodchikov tetrahedron equation is characterized as the transition matrix for a particular case of type A, which clarifies an algebraic origin of it.
Our work is inspired by the recent developments connecting transition matrices for quantum non-super algebras with intertwiners of irreducible representations of quantum coordinate rings.
We also discuss the crystal limit of transition matrices, which gives a super analog of transition maps of Lusztig's parametrizations of the canonical basis.
\vspace{0.4cm}
\tableofcontents
\addtocontents{toc}{\setcounter{tocdepth}{1}}
\section{Introduction}\label{sec 1}
\subsection{Background}\label{sec 11}
The Zamolodchikov tetrahedron equation\cite{Zam80} is a three dimensional analog of the Yang-Baxter equation\cite{Bax07}, where the latter equation serves as a cornerstone of integrable systems in two dimensions in terms of its physical applications and deeply understood algebraic aspects.
Along the same line as the Yang-Baxter equation, the tetrahedron equation gives the condition of factorizations for a four-body scattering of strings, and also gives a sufficient condition for the commutativity of the associated layer-to-layer transfer matrix.
Solutions to the Yang-Baxter equation are called $R$ matrices, and it is well known that we can systematically construct $R$ matrices through the Drinfeld-Jimbo quantum algebras\cite{Dri86,Jim86b}, but unlike the Yang-Baxter equation, there is no known way to obtain non-trivial solutions to the tetrahedron equation as such a systematic procedure.
\par
Historically, one important family of solutions to the tetrahedron equation is the $N$-state Zamolodchikov model, which was first proposed by \cite{Zam81} for $N=2$ as the first non-trivial solution and later generalized by \cite{BB92,SMS96} for general $N$.
From an algebraic point of view, it is known that the solutions are related to the $R$ matrices associated with the cyclic representations of the affine quantum algebra $U_q(A_{n-1}^{(1)})$ at roots of unity.
For the history of the solutions, see the introduction of \cite{BS06} and references therein.
\par
In this paper, we focus on infinite-dimensional solutions on the Fock spaces, which are essentially different solutions from the $N$-state Zamolodchikov model.
Our starting point is the known solution $(\tdr,\tdl)$ to the following tetrahedron equations:
\begin{align}
\tdr_{123}\tdr_{145}\tdr_{246}\tdr_{356}=\tdr_{356}\tdr_{246}\tdr_{145}\tdr_{123}
,
\label{KV94 te intro}
\\
\tdl_{123}\tdl_{145}\tdl_{246}\tdr_{356}=\tdr_{356}\tdl_{246}\tdl_{145}\tdl_{123}
,
\label{BS06 te intro}
\end{align}
where indices represent the tensor components on which each matrix acts non-trivially.
The matrix elements of $\tdr\in\mathrm{End}(F\otimes F\otimes F)$ and $\tdl\in\mathrm{End}(V\otimes V\otimes F)$ will be specified in (\ref{3dR mat el}) and (\ref{3dL mat el}), where $F$ and $V$ are the bosonic and Fermionic Fock spaces, respectively.
We call them the 3D R and 3D L.
\par
The 3D R was first derived\cite{KV94} as the intertwiner of the irreducible representations of the quantum coordinate ring $A_q(A_2)$, where the associated tetrahedron equation (\ref{KV94 te intro}) holds as the identity of the intertwiner of the irreducible representations of the quantum coordinate ring $A_q(A_3)$\cite{RTF90}.
The 3D R was also independently discovered by the seminal paper \cite{BS06} as explained later, and they are identified by \cite{KO12}.
As an amazing connection, the 3D R also gives the transition matrix of the PBW bases of the nilpotent subalgebra of the quantum algebra $U_q^{+}(A_2)$.
It is first observed by Sergeev\cite{Ser08}, and later systematically generalized as the Kuniba-Okado-Yamada theorem\cite{KOY13}, which states that intertwiners of irreducible representations of quantum coordinate rings agree with transition matrices of PBW bases of the nilpotent subalgebra of quantum algebras for all finite-dimensional simple Lie algebras.
See also \cite{Sai16,Tan17} which proved and sophisticated this theorem from a different point of view.
\par
On the other hand, the 3D L was obtained by a heuristic \textit{quantization} of the solution to the local Yang-Baxter equation\cite{MN89} by Bazhanov-Sergeev\cite{BS06}.
They made an ansatz that the 3D L gives an operator-valued solution to the local Yang-Baxter equation, which is equivalent to the tetrahedron equation (\ref{BS06 te intro}), and solved (\ref{BS06 te intro}) for $\tdr$.
It also gives an alternative derivation of the 3D R.
As a remarkable result related to the 3D L, it is known that the layer-to-layer transfer matrix of size $m\times n$ associated with the 3D L gives the spectral duality between different row-to-row transfer matrices: $\tAn[m]$ spin chain of system size $n$ and $\tAn[n]$ spin chain of system size $m$\cite{BS06}.
The duality is called the rank-size duality, and later, also appeared in the context of the five-dimensional gauge theory\cite{MMRZZ13}.
\par
Of course, the 3D R and 3D L are essentially three-dimensional objects, but it is known that there is an interesting connection to the $R$ matrix.
More concretely, there is a systematic way to reduce one solution to the tetrahedron equation to an infinite family of $R$ matrices.
By applying this procedure to the 3D R and 3D L, we can obtain explicit formulae of the $R$ matrices associated with some affine quantum algebras\cite{KOS15}.
By $n$-concatenation of the 3D R, we obtain the $R$ matrices associated with the symmetric tensor representations of $U_q(A_{n-1}^{(1)})$, and the Fock representations of $U_q(D_{n+1}^{(2)})$, $U_q(A_{2n}^{(2)})$ and $U_q(C_{n}^{(1)})$.
Similarly, by $n$-concatenation of the 3D L, we obtain the $R$ matrices associated with the fundamental representations of $U_q(A_{n-1}^{(1)})$, and the spin representations of $U_q(D_{n+1}^{(2)})$, $U_q(B_{n}^{(1)})$ and $U_q(D_{n}^{(1)})$.
Moreover, by mixing uses of some 3D R and 3D L, we also obtain the $R$ matrices associated with the generalized quantum groups\cite{KOS15}.
They are called \textit{matrix product solutions} to the Yang-Baxter equation.
For more details, see \cite{KOS15} and references therein.
\subsection{Motivation}\label{sec 12}
One of our motivations for this paper is why the 3D R and 3D L lead to such similar results, although they have totally different origins.
Actually, the 3D L has been derived again in several ways after \cite{BS06}.
First, \cite{KO12} identified the tetrahedron equation (\ref{BS06 te intro}) as the set of intertwining relations of the irreducible representations of $A_q(A_2)$, that is, the tetrahedron equation (\ref{BS06 te intro}) is obtained by arranging the intertwining relations of the 3D R into the matrix form (\ref{BS06 te intro}), simply by introducing the matrix $\tdl$.
See Remark \ref{KO12 obs rem}.{\ }for more details of this observation.
This procedure also works for the intertwining relation for $A_q(C_2)$\cite{KP18} and even for $A_q(G_2)$\cite{Kun18}, and leads to matrix product solutions to the reflection equation of Cherednik\cite{Che84} and the $G_2$ reflection equation.
See Remark \ref{KP18 rem} for more details for type C.
These are interesting connections but quite mysterious.
Also, although this connection for type A gives a derivation of the tetrahedron equation (\ref{BS06 te intro}), algebraic origins of the 3D L has been still unclear.
\par
On the one hand, Sergeev gave a parallel derivation\cite{Ser09} for the 3D R and 3D L by using the methods called \textit{quantum geometry}\cite{BMS10}.
At first glance, it seems that they consider something like a super analog of the irreducible representations of the quantum coordinate ring $A_q(A_2)$.
See for example (56)$\sim$(59) of \cite{Ser09}.
However, to verify it is highly non-trivial because there is no theory about irreducible representations of quantum super coordinate rings like Soibelman's theory for the non-super case\cite{Soi92}.
Then, the result by \cite{Ser09} can not be understood in terms of usual languages of \textit{quantum algebras}, at least straightforwardly.
\par
We also remark that the classical limit of the tetrahedron equation (\ref{BS06 te intro}) is recently derived in relation to non-trivial transformations of a plabic network, which can be interpreted as cluster mutations\cite{GSZ20}.
\subsection{Main achievements}\label{sec 13}
In this paper, we give a derivation for the 3D L in terms of the PBW bases of the nilpotent subalgebra of the quantum superalgebra\cite{Yam94} associated with the Dynkin diagram \ddtwoA{\cc}{\cct}.
We identify the 3D L with the transition matrix of them, which clarifies a completely parallel origin for the 3D L to the 3D R.
This result is just a special case of our investigations: we study transition matrices associated with all Dynkin diagrams of type A in the case of rank 2, which become parts of the tetrahedron equations.
Actually, we obtain a matrix $\tdn\in\mathrm{End}(V\otimes F\otimes V)$ by considering the case of \ddtwoA{\cct}{\cct}, which is new and different from the 3D R and 3D L.
The matrix elements of $\tdn$ will be specified in (\ref{3dN mat el}), and we call $\tdn$ as the 3D N.
\par
By considering the transition matrix for the case of rank 3 and attributing it to a composition of transition matrices of rank 2 in two ways, we obtain several solutions to the tetrahedron equation which the 3D R, L, and N satisfy.
We study the transition matrices associated with all Dynkin diagrams of type A in the case of rank 3, where \ddthreeA{\cct}{\cc}{\cc} and \ddthreeA{\cct}{\cct}{\cc} can be easily attribnuted to \ddthreeA{\cc}{\cc}{\cct} and \ddthreeA{\cc}{\cct}{\cct}, respectively, so we consider 6 Dynkin diagrams in total.
The cases for \ddthreeA{\cc}{\cc}{\cc} and \ddthreeA{\cc}{\cc}{\cct} reproduce the known tetrahedron equations (\ref{KV94 te intro}) and (\ref{BS06 te intro}), respectively.
For the case of \ddthreeA{\cc}{\cct}{\cct}, we obtain the following equation:
\begin{align}
\tdn(q^{-1})_{123}\tdn(q^{-1})_{145}\tdr_{246}\tdl_{356}=\tdl_{356}\tdr_{246}\tdn(q^{-1})_{145}\tdn(q^{-1})_{123}
.
\label{3dn te intro}
\end{align}
This suggests the 3D N gives a new solution to the tetrahedron equation.
The remaining 3 cases also give the tetrahedron like equations, but actually they are the tetrahedron equations \textit{up to sign factors}.
Further investigations should be done as to whether we can attribute them to the usual tetrahedron equations.
See Remark \ref{LLMM remark} related to this issue.
\par
We can generalize these results to the case of type B.
For type B, the associated equation is the 3D reflection equation\cite{IK97}, which is proposed by Isaev and Kulish as a three-dimensional analog of the reflection equation of Cherednik\cite{Che84}.
They also call the equation the \textit{tetrahedron reflection equation}.
Actually, they obtained the equation as the associativity condition for the 3D boundary Zamolodchikov algebra\cite[(9)]{IK97}, just as the tetrahedron equation is obtained as the associativity condition for the 3D Zamolodchikov algebra\cite{Zam80}.
Physically, the 3D reflection equation gives the condition for factorizations for a three-body scattering of strings with boundary reflections, along the same line as the tetrahedron equation.
\par
Essentially, there are only two known non-trivial solutions to the 3D reflection equation\cite{KO12,KO13}.
Here, we present one of the equations:
\begin{align}
\tdr_{456}
\tdr_{489}
\tdj_{3579}
\tdr_{269}
\tdr_{258}
\tdj_{1678}
\tdj_{1234}
=
\tdj_{1234}
\tdj_{1678}
\tdr_{258}
\tdr_{269}
\tdj_{3579}
\tdr_{489}
\tdr_{456}
,
\label{KO13 tre intro}
\end{align}
where $\tdr$ is the 3D R and the matrix elements of $\tdj\in\mathrm{End}(F\otimes F\otimes F\otimes F)$ will be specified in (\ref{3dJ mat el}).
We call $\tdj$ as the 3D J.
The 3D J was first derived as the intertwiner of the irreducible representations of the quantum coordinate ring $A_q(B_2)$, where the associated 3D reflection equation (\ref{KO13 tre intro}) holds as the identity of the intertwiner of the irreducible representations of the quantum coordinate ring $A_q(B_3)$\cite{KO13}.
As an immediate corollary of the Kuniba-Okado-Yamada theorem, we find the 3D J also gives the transition matrix of the PBW bases of the nilpotent subalgebra of the quantum algebra $U_q^{+}(B_2)$.
Our result for type B gives new solutions to the 3D reflection equation, which generalizes the solution (\ref{KO13 tre intro}) to the family of solutions (\ref{tre general}).
Actually, we introduce three analogs of the 3D J; we call them the 3D X, Y and Z.
We emphasize that our result also gives some explicit formula of transition matrices for type B.
\par
Our idea comes from trying to interpret Sergeev's result\cite{Ser09} on the side of PBW bases through the Kuniba-Okado-Yamada theorem although the theorem has not been established for the super case.
Note however our proofs do not need any result for quantum coordinate rings.
The derivation of the tetrahedron and 3D reflection equation is done only using \textit{higher-order} relations for quantum superalgebras.
\par
Finally, we discuss the behavior of transition matrices at $q=0$, which is known as the crystal limit\cite{Kas91}.
In the crystal limit, transition matrices of PBW bases give so-called transition maps of Lusztig's parametrizations of the canonical basis because PBW bases correspond to the canonical basis in that case\cite{Lus90,BZ01}.
Then, if we take the limit for super cases, it is expected we can obtain a super analog of transition maps.
In this paper, we show that we can take normalizations for transition matrices so that such non-trivial limits exist, and obtain explicit formulae for almost all cases.
In contrast to non-super cases, non-trivial elements of transition matrices take not only $0,1$ but also $-1$ in the crystal limit, and they define non-trivial bijections on mixed spaces of $\{0,1\}$ and $\mathbb{Z}_{\geq 0}$.
\subsection{Outline}\label{sec 14}
The ourline of this paper is as follows.
In Section \ref{sec 2}, we briefly review basic facts about finite-dimensional Lie superalgebras of type A and B.
Then, we introduce quantum superalgebras and their PBW theorem by \cite{Yam94}.
In Section \ref{sec 3}, we summarize the 3D operators which give solutions to the tetrahedron and 3D reflection equations.
Section \ref{sec 4} and \ref{sec 5} are main parts of this paper.
They can be read almost independently.
In Section \ref{sec 4}, we consider transition matrices of PBW bases of the nilpotent subalgebra of quantum superalgebras associated with all possible Dynkin diagrams of type A in the case of rank 2 and 3, and we obtain several solutions to the tetrahedron equation.
In Section \ref{tA qroot rels subsec}, we introduce some notations to briefly describe the PBW bases of rank 3 and higher-order relations for them, which are used in Section \ref{sec 43}.
In Section \ref{sec 42} and \ref{sec 43}, we study transition matrices of rank 2 and 3, respectively.
Section \ref{sec 5} is type B version of Section \ref{sec 4}, where the associated equation is the 3D reflection equaion.
Finally, in section \ref{sec 6}, we discuss the crystal limit of transition matrices.
Appendix \ref{app 3dX} is devoted to the proof of Theorem \ref{my 3dX result}.
In Appendix \ref{app 3dZ}, we derive recurrence equations for the 3D Z, which is the transition matrix associated with \ddtwoB{\cc}{\ccb}.
\subsection*{Acknowledgements}
The author thanks Atsuo Kuniba and Masato Okado for helpful comments and discussions.
Special thanks are due to Atsuo Kuniba for the encouragement to complete this project.
\addtocontents{toc}{\setcounter{tocdepth}{2}}
\section{Quantum superalgebras of type A and B}\label{sec 2}
\subsection{Root data of finite-dimensional Lie superalgebras}\label{sec 21}
In this paper, we consider quantum superalgebras associated with finite-dimensional Lie superalgebras $\tA$ and $\tB$\cite{CW12,FSS89,Kac77,Zha14}.
Here, $m,n$ are non-negative integers and we assume $m+n\geq 2$.
We call $\tA$ as type A and $\tB$ as type B.
Let $\GG$ denote $\tA$ or $\tB$.
If we set $n=0$, $\GGn=\GG[m|0]$ is reduced to finite-dimensional simple Lie algebras.
To avoid confusion, we also call the finite-dimensional simple Lie algebras as the finite-dimensional simple non-super Lie algebras.
In this case, we simply write $\tA[m|0]$ and $\tB[2m+1|0]$ by $\tAn$ and $\tBn$, respectively.
In this section, we describe root data of $\GG$.
Here, we use a similar setup to \cite{XZ16}.
\par
We set $N=m+n$.
Let $\asr$ be the $N$-dimensional real vector space with a non-degenerate symmetric bilinear form $(\cdot,\cdot):\asr\times\asr\to\mathbb{R}$.
We use $\ee_i{\ }(i=1,\cdots,m)$ and $\dd_i{\ }(i=1,\cdots,n)$ as a basis of $\asr$ with a non-degenerate symmetric bilinear form given by
\begin{align}
(\ee_i,\ee_j)=(-1)^{\theta}\delta_{i,j}
,\quad
(\dd_i,\dd_j)=-(-1)^{\theta}\delta_{i,j}
,\quad
(\ee_i,\dd_j)=0
,
\label{ambi sp}
\end{align}
where $\ttt=0,1$ which is specified above Example \ref{theta example}, and $\delta_{i,j}$ is the Kronecker delta.
Let $\BB(m|n)=(\BB_1,\cdots,\BB_{N})$ denote an ordered basis of $\asr$ which is a permutation of $\ee_i{\ }(i=1,\cdots,m)$ and $\dd_i{\ }(i=1,\cdots,n)$.
Without loss of generality, we only consider cases when $\ee_{i}$ appears before $\ee_{i+1}$ and $\dd_{i}$ appears before $\dd_{i+1}$ in $\BB(m|n)$ for all $i$, which is called admissible.
\par
Let $\Phi$ be the set of roots of $\GG$ and $\Pi=\{\aaa_1,\cdots,\aaa_{r}\}$ be the set of simple roots of $\GG$, where $r$ is the rank of $\GG$.
Here, $r=N-1$ for $\tA$ and $r=N$ for $\tB$.
We write the set of labels by $I=\{1,\cdots,r\}$.
We call $\Phi$ and $\Pi$ as the root system and fundamental system of $\GG$.
When $\BB(m|n)$ is given, the fundamental system $\Pi$ is realized as Table \ref{simple roots}.
\begin{table}[htb]
\centering
\caption{}
\label{simple roots}
\begin{tabular}{c|c}\hline
$\GG$ & simple roots \\\hline
$\tA$ & $\aaa_i=\BB_{i}-\BB_{i+1}{\ }(i=1,\cdots,r)$ \\\hline
$\tB$ & $\aaa_i=\BB_{i}-\BB_{i+1}{\ }(i=1,\cdots,r-1)$, $\aaa_{r}=\BB_{r}$ \\\hline
\end{tabular}
\end{table}
We write the positve and negative part of the root lattice of $\GG$ by $Q^{\pm}=\pm\sum_{i=1}^{r}\mathbb{Z}_{\geq 0}\aaa_{i}\backslash\{0\}$ and the positive and negative part of the root system by $\pmr=\Phi\cap Q^{\pm}$, which will be identified in Table \ref{positive roots} for each case.
We also set the weight lattice of $\GG$ by $\asz=\sum_{i=1}^{m}\mathbb{Z}\ee_{i}\oplus \sum_{i=1}^{n}\mathbb{Z}\dd_{i}$.
\par
For $\lambda=\sum_{i=1}^{m}a_{i}\ee_{i}+\sum_{i=1}^{n}b_{i}\dd_{i}\in \asz$, we define the parity $p:\asz\to \{0,1\}$ of $\lambda$ as $p(\lambda)=\sum_{i=1}^{n}b_i{\ }(\mathrm{mod}{\ }2)$, and this induces the parity of elements of $\Phi$ via its realization.
We call $\lambda\in\asz$ is even if $p(\lambda)=0$ and odd if $p(\lambda)=1$.
We set the set of indices of odd simple roots by $\tau\subset I$.
Then, the positive part of the root system $\pr$ is given in Table \ref{positive roots}.
\begin{table}[htb]
\centering
\caption{}
\label{positive roots}
\begin{tabular}{c|c}\hline
$\GG$ & positive part of root system \\\hline
$\tA$ & $\pr=\{\BB_{i}-\BB_{j}{\ }(1\leq i<j\leq N)\}$ \\\hline
$\tB$ & $\pr=\{\BB_{i}\pm\BB_{j}{\ }(1\leq i<j\leq N),{\ }\BB_{i}{\ }(1\leq i\leq N),{\ }2\BB_{i}{\ }(1\leq i\leq N,{\ }i\in\tau)\}$ \\\hline
\end{tabular}
\end{table}
We set the set of reduced roots by $\rr=\{\aaa\in\Phi\mid \aaa/2\notin\Phi\}$ and the positive and negative part of it by $\pmrr=\rr\cap Q^{\pm}$.
\par
For $\alpha\in\Phi$, if $p(\alpha)=1$ and $(\alpha,\alpha)=0$, we call $\alpha$ as the isotropic odd root.
On the one hand, if $p(\alpha)=1$ and $(\alpha,\alpha)\neq 0$, we call $\alpha$ as the anisotropic odd root.
The set of even roots, isotropic odd roots and anisotropic odd roots are denoted by $\re$, $\ri$ and $\ra\subset \Phi$, respectively.
Also, the set of the positive and negative part of even roots, isotropic odd roots and anisotropic odd roots are denoted by $\pmre$, $\pmri$ and $\pmra\subset \pmr$, respectively.
We also set the reduced version of them by $\pmrer=\pmre\cap\rr$, $\pmrir=\pmri$ and $\pmrar=\pmra$.
\subsection{Cartan matrices, Dynkin diagrams and Weyl groups}\label{sec 22}
Let $(a_{ij})_{i,j\in I}$, $(d_{i})_{i\in I}$ be the Cartan matrix and the symmetrizing matrix of $\GG$.
Here, $d_i$ is given by $(\alpha_i,\alpha_i)/2$ for $\alpha_{i}\in\prer\cup\prar$ and $1$ for $\alpha_{i}\in\prir$.
Also, $a_{ij}$ is given by $a_{ij}=(\alpha_i,\alpha_j)/d_i$.
We often write $A=(a_{ij})$ and $D=\mathrm{diag}(d_1,\cdots,d_r)$, and the symmetrized Catran matrix by $DA=(d_{i}a_{ij})$.
We call the pair $(A,p)$ as the Cartan data of $\GG$.
For later use, we also define $d_{\alpha}$ by $(\alpha,\alpha)/2$ for $\alpha\in\re\cup\ra$ and $1$ for $\alpha\in\ri$.
Let $\mathfrak{h}=\sum_{i=1}^{r}\mathbb{C}h_i$ be the Cartan subalgebra of $\GG$, where $\{h_i\}_{i\in I}$ is chosen as $\aaa_j(h_i)=a_{ij}$.
We call $\{h_i\}_{i\in I}$ as the set of simple coroots of $\GG$.
\par
The Cartan data can be diagrammatically represented by the Dynkin diagram.
The Dynkin diagram associated with $(A,p)$ is defined as follows.
First, we set $r$ dots and decorate the $i$-th dot by $\cc$ for $\aaa_i\in\prer$, $\cct$ for $\aaa_i\in\prir$ and $\ccb$ for $\aaa_i\in\prar$, respectively.
We also use $\cctn$ representing $\cc$ or $\cct$.
Then, for every pair of different numbers $(i,j)$, we connect them with $|a_{ij}|$ lines if $a_{ij}\neq 0$.
Also, if $|a_{ij}|\geq 2$, these lines are equipped with an arrow pointing from the $j$-th dot to the $i$-th dot.
All possible Dynkin diagrams of $\GG$ are given in Table \ref{dd table gen}, where $\BB_{N}=\ee_{m}$ for the first Dynkin diagram of $\tB$ and $\BB_{N}=\dd_{n}$ for the second Dynkin diagram of $\tB$.
\begin{table}[htb]
\centering
\caption{}
\label{dd table gen}
\begin{tabular}{c|c}\hline
$\GG$ & Dynkin diagram \\\hline
$\tA$ & \ddgenA \\\hline
\multirow{2}{*}[-2.4ex]{$\tB$} & \ddgenBe \\
& \ddgenBa\\\hline
\end{tabular}
\end{table}
\par
Here, we specify the value of $\theta=0,1$ in (\ref{ambi sp}).
If $\min\{(\alpha_{i},\alpha_{j})\mid i\neq j\}<0$ for both values of $\theta$, we choose $\theta$ so that $(\BB_1,\BB_1)=1$ holds.
If not, we choose $\theta$ so that $\min\{(\alpha_{i},\alpha_{j})\mid i\neq j\}<0$ is satisfied.
\begin{example}\label{theta example}
For the case \ddthreeWul{\ltA}{\ltB}{\cct}{\cc}{\cc}{\dd_1-\ee_2}{\ee_2-\ee_3}{\ee_3}{\ }, we have $DA=\begin{pmatrix}0&-1&0\\ -1&2&-1\\ 0&-1&1\end{pmatrix}$ for $\theta=0$ and $DA=\begin{pmatrix}0&1&0\\ 1&-2&1\\ 0&1&-1\end{pmatrix}$ for $\theta=1$.
We then choose $\theta=0$ for this case.
\par
For more examples, see Section \ref{sec 4} and \ref{sec 5}.
\end{example}
\par
Let $\WG$ be the Weyl group of $\GG$ which is generated by reflections $s_{\alpha}{\ }(\alpha\in\re\cup\ra)$ which are associated with even and anisotropic roots.
The action of them is given by
\begin{align}
s_{\alpha}(\beta)
=\beta-\frac{2(\alpha,\beta)}{(\alpha,\alpha)}\alpha
\quad
(\beta\in\Phi)
.
\end{align}
Under actions of $\WG$, the root system is invariant.
The image of the fundamental system is a different one, but it gives the same Cartan data.
For the finite-dimensional simple non-super Lie algebras, it is known that all possible choice of the fundamental system is conjugate via the Weyl group actions\cite[\S 10.3. Theorem]{Hum72}.
Then, the Dynkin diagrams one-to-one correspond to the finite-dimensional simple non-super Lie algebras.
For non-super case, relations of $\WGn$ are given by $s_i^2=1$, $(s_is_j)^{m_{ij}}=1{\ }(i\neq j)$ where $m_{ij}=2,3,4$ for $a_{ij}a_{ji}=0,1,2$, respectively.
Here, we write $s_{i}=s_{\alpha_i}$.
\par
For general finite-dimensional Lie superalgebras, however, the fundamental systems are not always conjugate via the Weyl group actions.
It is known that by adding some elements and extending the Weyl group $\WG$, all fundamental systems become conjugate\cite[Appendix II. Theorem]{LSS85}.
The elements are called odd reflections, and we call the extended Weyl group the Weyl supergroup denoted by $\WSG$.
Formally, odd reflections are reflections associated with odd roots, and the action of the elements of the Weyl supergroup is given by
\begin{align}
s_{\alpha}(\beta)
=
\begin{cases}
\beta-\frac{2(\alpha,\beta)}{(\alpha,\alpha)}\alpha
\quad &(\alpha\in\re\cup\ra), \\
\beta+\alpha
\quad &(\alpha\in\ri,{\ }(\alpha,\beta)\neq 0,{\ }\beta\neq\alpha), \\
\beta
\quad &(\alpha\in\ri,{\ }(\alpha,\beta)=0,{\ }\beta\neq\alpha), \\
-\alpha
\quad &(\beta=\alpha),
\end{cases}
\end{align}
where $\alpha,\beta\in\Phi$.
Similar to usual reflections, the root system is invariant under actions of odd reflections, but the image of the fundamental system gives different Cartan data.
Therefore, the Dynkin diagrams do \textit{not} correspond to the finite-dimensional Lie superalgebras but rather their fundamental systems.
\par
The standard choice of the fundamental system of $\GG$ is called distinguished, where the associated Dynkin diagrams have only one odd root.
The realizations and the corresponding Dynkin diagrams are given in Table \ref{dis roots din}.
\begin{table}[htb]
\centering
\caption{}
\label{dis roots din}
\begin{tabular}{c|c}\hline
$\GG$ & distinguished Dynkin diagram \\\hline
$\tA$ & \ddgenAd \\\hline
$\tB{\ }(m>0)$ & \ddgenBed \\\hline
$\tB[1|2n]$ & \ddgenBad\\\hline
\end{tabular}
\end{table}
In this paper, we focus on the nilpotent subalgebra of quantum superalgebras, rather than the whole algebras.
Since the nilpotent subalgebra depends on the choice of the fundamental system of $\GG$, in addition to the distinguished Dynkin diagrams, we also consider non-distinguished ones as given in Table \ref{dd table gen}.
\subsection{Quantum superalgebras}
Throughout this paper, we assume $q$ is generic.
We set $q_i=q^{d_i}$ and $v_{i}=q^{(\BB_i,\BB_i)}{\ }(i\in I)$.
We use a variant of $q$-number and its factorial defined by
\begin{align}
[k]_{q,\psign}
=
\frac{(\psign q)^{k}-q^{-k}}{\psign q-q^{-1}}
,\quad
[m]_{q,\psign}!
=
\prod_{k=1}^{m}
[k]_{q,\psign}
,
\end{align}
where $k,m\in\mathbb{Z}_{\geq 0}$ and $\psign=\pm 1$\cite{CHW16}.
We promise $[0]_{q,\psign}!=1$.
For simplicity, we write $[k]_{q,1}=[k]_{q}$ and $[m]_{q,1}!=[m]_{q}!$ for $\psign=1$.
The quantum superalgebra $\UGG$ associated with the Cartan data $(A,p)$ is an associative algebra over $\mathbb{C}$ generated by $\{e_i,f_i,k_i=q_{i}^{h_i}\mid i\in I\}$ satisfying the following relations\cite{CK90,Yam94}:
\begin{align}
&k_i^{\pm 1}k_i^{\mp 1}=1
,\quad k_ik_j=k_jk_i
,\quad k_ie_j=q_i^{a_{ij}}e_jk_i
,\quad k_if_j=q_i^{-a_{ij}}f_jk_i
\label{qs rel1}
,\\
&e_if_j-(-1)^{p(\alpha_i)p(\alpha_j)}f_je_i=\delta_{i,j}\frac{k_i-k_i^{-1}}{q_i-q_i^{-1}}
\label{qs rel2}
,\\
&\sum_{\nu=0}^{1+|a_{ij}|}
(-1)^{\nu+p(\alpha_i)\nu(\nu-1)/2+\nu p(\alpha_i)p(\alpha_j)}
e_{i}^{(1+|a_{ij}|-\nu)}
e_{j}
e_{i}^{(\nu)}
=0
\quad
(a_{ij}\neq 0,{\ }i\neq j,{\ }\alpha_{i}\in\prer\cup\prar)
\label{qs rel3}
,\\
&\sum_{\nu=0}^{1+|a_{ij}|}
(-1)^{\nu+p(\alpha_i)\nu(\nu-1)/2+\nu p(\alpha_i)p(\alpha_j)}
f_{i}^{(1+|a_{ij}|-\nu)}
f_{j}
f_{i}^{(\nu)}
=0
\quad
(a_{ij}\neq 0,{\ }i\neq j,{\ }\alpha_{i}\in\prer\cup\prar)
\label{qs rel4}
,\\
&[e_{i},e_{j}]=0
,\quad
[f_{i},f_{j}]=0
\quad
(a_{ij}=0)
\label{qs rel5}
,
\end{align}
where we set $e_{i}^{(\nu)}=e_{i}^{\nu}/[\nu]_{q_{i},(-1)^{p(\alpha_{i})}}!$, $f_{i}^{(\nu)}=f_{i}^{\nu}/[\nu]_{q_{i},(-1)^{p(\alpha_{i})}}!$, and so-called additional relations for the case when the associated Dynkin diagram has the subdiagram {\ }\ddthree{\ltA}{\ltA}{\cctn}{\cct}{\cctn}{\ }or {\ }\ddthree{\ltA}{\ltB}{\cctn}{\cct}{\cc}{\ }or {\ }\ddthree{\ltA}{\ltB}{\cctn}{\cct}{\ccb}:
\begin{align}
\begin{split}
&e_{i-1}e_{i}e_{i+1}e_{i}
+(-1)^{\phaseIrel_{i}}e_{i}e_{i-1}e_{i}e_{i+1}
+(-1)^{\phaseIIrel_{i}}e_{i}e_{i+1}e_{i}e_{i-1}
+(-1)^{\phaseIrel_{i}+\phaseIIrel_{i}}e_{i+1}e_{i}e_{i-1}e_{i}
\\
&-(-1)^{p(\alpha_{i-1})}
\left(q+q^{-1}\right)
e_{i}e_{i-1}e_{i+1}e_{i}
=0
\quad
(\alpha_{i}\in\prir)
,
\end{split}
\label{qs rel6}
\\
\begin{split}
&f_{i-1}f_{i}f_{i+1}f_{i}
+(-1)^{\phaseIrel_{i}}f_{i}f_{i-1}f_{i}f_{i+1}
+(-1)^{\phaseIIrel_{i}}f_{i}f_{i+1}f_{i}f_{i-1}
+(-1)^{\phaseIrel_{i}+\phaseIIrel_{i}}f_{i+1}f_{i}f_{i-1}f_{i}
\\
&-(-1)^{p(\alpha_{i-1})}
\left(q+q^{-1}\right)
f_{i}f_{i-1}f_{i+1}f_{i}
=0
\quad
(\alpha_{i}\in\prir)
,
\end{split}
\label{qs rel7}
\end{align}
where $\phaseIrel_i,\phaseIIrel_i$ are given by
\begin{align}
\phaseIrel_{i}=p(\alpha_{i-1})+p(\alpha_{i+1})
,\quad
\phaseIIrel_{i}=p(\alpha_{i-1})p(\alpha_{i+1})
.
\end{align}
For reader's convenience, we list (\ref{qs rel3}) for each case which appears in this paper:
\begin{enumerate}[(1)]
\item
$|a_{ij}|=1$ and $\alpha_i\in\prer$:
\begin{align}
e_{i}^{2}e_{j}-(q+q^{-1})e_{i}e_{j}e_{i}+e_{j}e_{i}^2=0
,
\end{align}
\item
$|a_{ij}|=2$ and $\alpha_i\in\prer$:
\begin{align}
e_{i}^3e_{j}-(q+1+q^{-1})e_{i}^2e_{j}e_{i}+(q+1+q^{-1})e_{i}e_{j}e_{i}^2-e_{j}e_{i}^3=0
,
\end{align}
\item
$|a_{ij}|=2$ and $\alpha_i\in\prar$, $\alpha_j\in\prer$:
\begin{align}
e_{i}^3e_{j}+(1-q-q^{-1})e_{i}^2e_{j}e_{i}+(1-q-q^{-1})e_{i}e_{j}e_{i}^2+e_{j}e_{i}^3=0
,
\end{align}
\item
$|a_{ij}|=2$ and $\alpha_i\in\prar$, $\alpha_j\in\prir$:
\begin{align}
e_{i}^3e_{j}-(1-q-q^{-1})e_{i}^2e_{j}e_{i}+(1-q-q^{-1})e_{i}e_{j}e_{i}^2-e_{j}e_{i}^3=0
,
\end{align}
\end{enumerate}
where we always assume $i\neq j$.
\par
In this paper, we focus on $\UGGp$ which is the nilpotent subalgebra of $\UGG$ generated by $\{e_{i}\}_{i\in I}$.
We represent $\UGGp$ by the Dynkin diagram associated with the Cartan data $(A,p)$.
We have the root space decomposition of $\UGGp=\bigoplus_{\alpha\in Q^{+}}\UGGp_{\alpha}$ where each root space are given by $\UGGp_{\alpha}=\{g\mid k_ig=q_{i}^{\alpha(h_{i})}gk_{i}{\ }(i\in I)\}$.
For $x\in\UGGp_{\alpha},{\ }y\in\UGGp_{\beta}$, we define the $q$-commutator $[\cdot,\cdot]_{q}$ by
\begin{align}
[x,y]_{q}
=xy-(-1)^{p(\alpha)p(\beta)}q^{-(\alpha,\beta)}yx
,
\end{align}
and for simplicity we write $[\cdot,\cdot]_{1}=[\cdot,\cdot]$ for $q=1$.
By using the $q$-commutator, the Serre relation (\ref{qs rel3}) and the additional relation (\ref{qs rel6}) are simply written as follows:
\begin{align}
\begin{alignedat}{2}
&[[e_{j},e_{i}]_{q},e_{i}]_{q}
=[e_{i},[e_{i},e_{j}]_{q}]_{q}
=0
&&\quad (|a_{ij}|=1)
,\\
&[[[e_{j},e_{i}]_{q},e_{i}]_{q},e_{i}]_{q}
=[e_{i},[e_{i},[e_{i},e_{j}]_{q}]_{q}]_{q}
=0
&&\quad (|a_{ij}|=2)
,
\end{alignedat}
\label{serre qcom}
\end{align}
\vspace{-1em}
\begin{align}
[[[e_{i-1},e_{i}]_{q},e_{i+1}]_{q},e_{i}]=0
.
\end{align}
\par
For later use, let $\chi:\UGG\to\UGG$ be the anti-algebra automorphism given by
\begin{align}
\chi(e_i)=e_i
,\quad
\chi(f_i)=f_i
,\quad
\chi(k_i)=(-1)^{p(\alpha_{i})}k_i^{-1}
\label{anti chi}
.
\end{align}
Then, $\chi$ also gives the anti-algebra automorphism on $\UGGp$.
\subsection{PBW bases of the nilpotent subalgebra of quantum superalgebras}
We begin with non-super cases.
In that case, we have $\pr=\pre$.
Let $w_0$ be the longest element of $\WGn$.
When a reduced expression of $w_0=s_{i_1}\cdots s_{i_{l}}$ is given, we set $\beta_{t}{\ }(t=1,\cdots,l)$ by
\begin{align}
\beta_{t}=s_{i_1}\cdots s_{i_{t-1}}(\aaa_{i_{t}})
,
\label{pr re}
\end{align}
where we set $l=|\pr|$.
Then, it is known that we have $\beta_{t}\in\pr{\ }(t=1,\cdots,l)$, $\beta_{i}\neq\beta_{j}{\ }(i\neq j)$ and $\pr=\{\beta_{t}\mid 1\leq t\leq l\}$\cite[P.25]{Hum90}.
\par
It is also known that there exists a quantum analog of this procedure.
Let $T_i:\UGGn\to\UGGn{\ }(i\in I)$ be the algebra automorphism given by
\begin{align}
T_i(e_j)
&=
\begin{cases}
&-k_if_i
\quad
(i=j)
,\\
&\sum_{r=0}^{-a_{ij}}
(-1)^{r}q_{i}^r
e_i^{(r)}
e_j
e_i^{(-a_{ij}-r)}
\quad
(i\neq j)
,
\end{cases}
\\
T_i(f_j)
&=
\begin{cases}
&-e_ik_i^{-1}
\quad
(i=j)
,\\
&\sum_{r=0}^{-a_{ij}}
(-1)^{r}q_{i}^{-r}
e_i^{(-a_{ij}-r)}
e_j
e_i^{(r)}
\quad
(i\neq j)
,
\end{cases}
\\
T_i(k_j)
&=
k_{i}^{-a_{ij}}k_{j}
.
\end{align}
Here, $T_i$ is known as the so-called Lusztig's braid group action on $\UGGn$\cite{Lus90}.
Actually, it is known that $\{T_i\}_{i\in I}$ satisfy the braid group relations.
We set $e_{\beta_{t}}{\ }(t=1,\cdots,l)$ by
\begin{align}
e_{\beta_t}=T_{i_1}T_{i_2}\cdots T_{i_{t-1}}(e_{i_t})
,
\label{qroot Lus}
\end{align}
where $\beta_{t}$ is given by (\ref{pr re}).
Then, it is known that we have $e_{\beta_{t}}\in\UGGpn_{\beta_{t}}{\ }(t=1,\cdots,l)$ and $e_{\beta_{i}}\neq e_{\beta_{j}}{\ }(i\neq j)$.
Also, it gives a PBW basis of $\UGGpn$, which depends on the choice of reduced expressions of $w_0$\cite{Lus90}:
\begin{theorem}\label{Lus PBW thm}
For $A=(a_1,\cdots,a_{l})\in(\mathbb{Z}_{\geq 0})^{l}$, we set
\begin{align}
E^{A}=e_{\beta_{1}}^{(a_1)}e_{\beta_{2}}^{(a_2)}\cdots e_{\beta_{l}}^{(a_l)}
,
\end{align}
where we normalize $e_{\beta_{t}}^{(a_{t})}=e_{\beta_{t}}^{a_{t}}/[a_t]_{p_t}!$, $p_t=q^{d_{\beta_{t}}}$.
Then, $\{E^{A}\mid A\in (\mathbb{Z}_{\geq 0})^{l}\}$ is a basis of $\UGGpn$.
\end{theorem}
\par
For super cases, it is known that there is a naive construction of a PBW basis without using some maps like the Lusztig's braid group action\cite{KT91,Yam94}.
We note that a super analog of Lusztig's braid group action was introduced in the context of the so-called Weyl groupoid\cite{Hec10,HY08}.
\par
Let us explain the construction by \cite{Yam94}.
We define two partial orders $<$ on $\prr$ as follows.
For $\gamma=\sum_{i=1}^{r}c_i\alpha_i\in\prr$, we define the integers $\htf(\gamma),g(\gamma),c_{\gamma}\in\mathbb{N}$ by $\htf(\gamma)=\sum_{i=1}^{r}c_i$, $g(\gamma)=\min\{i\mid c_i\neq 0\}$ and $c_{\gamma}=c_{g(\gamma)}$.
Then, for $\alpha,\beta\in\prr$, we set two partial orders $O_1,O_2$ by
\begin{align}
&O_1:
\quad
\alpha<\beta
\quad
\Longleftrightarrow
\quad
g(\alpha)<g(\beta)
{\ }
\mathrm{or}
{\ }
(g(\alpha)=g(\beta){\ }\mathrm{and}{\ } \htf(\alpha)<\htf(\beta))
\label{no1}
,\\
&O_2:
\quad
\alpha<\beta
\quad
\Longleftrightarrow
\quad
g(\alpha)>g(\beta)
{\ }
\mathrm{or}
{\ }
(g(\alpha)=g(\beta){\ }\mathrm{and}{\ } \htf(\alpha)>\htf(\beta))
\label{no2}
.
\end{align}
Note that $O_1$ is the same order as \cite{Yam94}.
By using them, we define quantum root vectors as follows:
\begin{definition}\label{qrv}
For every $\beta\in\prr$, we define the elements $e_{\beta}\in\UGGp_{\beta}$ as follows:
\begin{enumerate}[(i)]
\item
If $\beta=\alpha_{i}$, we set $e_{\beta}=e_{i}$.
\item
If $\beta=\alpha+\alpha_{i}$ where $\alpha\in\prr$ and $g(\alpha)<i$, we define $e_{\beta}$ depending on the partial order $O_i$.
We set $e'_{\beta}=[e_{i},e_{\alpha}]_{q}$ for $O_1$, and $e'_{\beta}=[e_{\alpha},e_{i}]_{q}$ for $O_2$.
Then, we set $e_{\beta}=e'_{\beta}/(q^{1/2}+q^{-1/2})$ for the case $\GG=\tB$, $i=r$ and $\alpha=\BB_{j}{\ }(1\leq j\leq r-1)$. We set $e_{\beta}=e'_{\beta}$ otherwise.
\end{enumerate}
\end{definition}
\noindent
We note that the above normalization factor $q^{1/2}+q^{-1/2}$ naturally appears from the Lusztig's braid group action for non-super cases.
\par
Then, the quantum root vectors give PBW bases of $\UGGp$:
\begin{theorem}\label{Yam PBW thm}
Let $\beta_t{\ }(t=1,\cdots,l)$ denote the reduced roots, which satisfy $\beta_1<\cdots<\beta_{l}$ under the order $O_i$.
Here, $l=|\prr|$.
For $A=(a_1,\cdots,a_{l})$ where $a_{t}\in\mathbb{Z}_{\geq 0}$ for $\beta_{t}\in\prer\cup\prar$ and $a_{t}\in\{0,1\}$ for $\beta_{t}\in\prir$, we set
\begin{align}
E^{A}_{i}=e_{\beta_{1}}^{(a_1)}e_{\beta_{2}}^{(a_2)}\cdots e_{\beta_{l}}^{(a_l)}
,
\end{align}
where we normalize $e_{\beta_{t}}^{(a_{t})}=e_{\beta_{t}}^{a_{t}}/[a_t]_{p_t,(-1)^{p(\beta_{t})}}!$, $p_t=q^{d_{\beta_{t}}}$.
Then
\begin{align}
B_i
=\{E^{A}_{i}\mid a_{t}\in\mathbb{Z}_{\geq 0}{\ }(\beta_{t}\in\prer\cup\prar),{\ }a_{t}\in\{0,1\}{\ }(\beta_{t}\in\prir)\}
,
\label{Yam PBW basis}
\end{align}
is a basis of $\UGGp$.
\end{theorem}
\begin{proof}
We attribute the statement to \cite{Yam94}.
First, we consider the case when the order is given by $O_1$.
In \cite{Yam94}, the order among the elements of $\prr$ is the same as $O_1$, but the quantum root vectors are defined by $e^{\mathrm{Yam}}_{\beta}=[e_{\alpha},e_{i}]_q$ instead of $e_{\beta}=[e_{i},e_{\alpha}]_q$ as Definition \ref{qrv}.
However, $e_{\beta}$ and $e^{\mathrm{Yam}}_{\beta}$ satisfy the following simple relation:
\begin{align}
e_{\beta}
=[e_{i},e_{\alpha}]_q
=(-1)^{p(\alpha_i)p(\alpha)+1}q^{-(\alpha_i,\alpha)}[e_{\alpha},e_{i}]_{q^{-1}}
=(-1)^{p(\alpha_i)p(\alpha)+1}q^{-(\alpha_i,\alpha)}
\left(e^{\mathrm{Yam}}_{\beta}|_{q\to q^{-1}}\right)
.
\end{align}
Then, the only differences between our construction and \cite{Yam94} are overall factors and its $q$-dependence.
Since the relations of $\UGG$ are invariant under $q\to q^{-1}$, by the Proposition 10.4.1{\ }of \cite{Yam94}, we find that (\ref{Yam PBW basis}) gives a PBS basis of $U_{q^{-1}}^{+}(\GG)$.
Then, under the order $O_1$, (\ref{Yam PBW basis}) gives a PBW basis of $\UGGp$.
\par
The case when the order is given by $O_2$ is attributed to the case of $O_1$.
Actually, $E_{2}^{A^{\mathrm{op}}}=\chi(E_{1}^{A})$ holds for every $A$, where $A^{\mathrm{op}}$ is the reverse order of $A$.
This shows (\ref{Yam PBW basis}) under the order $O_2$ gives a PBW basis of $\UGGp$ because $\chi$ is an automorphism on $\UGGp$.
\end{proof}
\begin{remark}
The construction by \cite{Yam94} can be considered as a natural analog of one of Theorem \ref{Lus PBW thm} as follows.
For non-super cases, we call an order $<$ among the elements of $\pr$ is \textit{normal} (or \textit{convex}) if, for $\alpha\in\pr$ which is written by $\alpha=\beta+\gamma{\ }(\beta,\gamma\in\pr)$, the order among $\alpha,\beta,\gamma$ satisfies $\beta<\alpha<\gamma$ or $\gamma<\alpha<\beta$.
Then, it is known that there exists a one-to-one correspondence between orders induced by reduced expressions of $w_0$ like (\ref{pr re}) and normal orders\cite[\S 3 Proposition 2]{Zhe87}.
The normal order can be defined in a similar way for super cases, and the orders (\ref{no1}) and (\ref{no2}) actually satisfy the condition of the normal order.
\end{remark}
\par
Let $\tm^{A}_{B}$ and $\tmt^{A}_{B}$ be the transition matrices given by
\begin{align}
E_{2}^{A}
&=
\sum_{B}
\tm_{B}^{A}E_{1}^{B^{\mathrm{op}}}
,
\label{tm1}
\\
E_{1}^{A}
&=
\sum_{B}
\tmt_{B}^{A}E_{2}^{B^{\mathrm{op}}}
.
\label{tm2}
\end{align}
where $X^{\mathrm{op}}=(x_l,\cdots,x_1)$ is the reverse order of $X=(x_1,\cdots,x_l)$.
They are one of the main objects of this paper.
By using $E_2^{X^{\mathrm{op}}}=\chi(E_1^{X})$, we obtain the following relation:
\begin{align}
\tmt_{B}^{A}=\tm_{B^{\mathrm{op}}}^{A^{\mathrm{op}}}
.
\label{tm rel}
\end{align}
We then only consider $\tm_{B}^{A}$ below.
\subsection{Technical lemmas for higher-order relations}
In this section, we introduce some technical lemma used to prove higher-order relations in the later sections.
First, the $q$-commutator enjoy the following Jacobi like identity\cite[(4.4.2)]{Yam94}.
\begin{lemma}\label{qcom jacobi}
For $x\in\UGGp_{\alpha},{\ }y\in\UGGp_{\beta},{\ }z\in\UGGp_{\gamma}$, we have
\begin{align}
[[x,y]_{q},z]_{q}-[x,[y,z]_{q}]_{q}
=
(-1)^{p(\beta)p(\gamma)}q^{-(\beta,\gamma)}[x,z]_{q}y
-(-1)^{p(\alpha)p(\beta)}q^{-(\alpha,\beta)}y[x,z]_{q}
.
\end{align}
\end{lemma}
\begin{proof}
By writing down the definitions, we get
\begin{align}
\begin{split}
[[x,y]_{q},z]_{q}
=
&xyz
-(-1)^{p(\alpha)p(\beta)}
q^{-(\alpha,\beta)}yxz
-(-1)^{(p(\alpha)+p(\beta))p(\gamma)}
q^{-(\alpha+\beta,\gamma)}zxy
\\
&+(-1)^{p(\alpha)p(\beta)+p(\beta)p(\gamma)+p(\gamma)p(\alpha)}
q^{-(\alpha,\beta)-(\beta,\gamma)-(\gamma,\alpha)}zyx
,
\end{split}
\\
\begin{split}
[x,[y,z]_{q}]_{q}
=
&xyz
-(-1)^{p(\beta)p(\gamma)}
q^{-(\beta,\gamma)}xzy
-(-1)^{p(\alpha)(p(\beta)+p(\gamma))}
q^{-(\alpha,\beta+\gamma)}yzx
\\
&+(-1)^{p(\alpha)p(\beta)+p(\beta)p(\gamma)+p(\gamma)p(\alpha)}
q^{-(\alpha,\beta)-(\beta,\gamma)-(\gamma,\alpha)}zyx
.
\end{split}
\end{align}
We then obtain the desired results.
\end{proof}
\begin{corollary}\label{qcom jacobi coro}
We set $x\in\UGGp_{\alpha},{\ }y\in\UGGp_{\beta},{\ }z\in\UGGp_{\gamma}$.
\begin{enumerate}[(1)]
\item
If $[x,z]_q=0$, we have $[[x,y]_{q},z]_{q}=[x,[y,z]_{q}]_{q}$.
\item
If $[y,z]=0$ and $(\beta,\gamma)=0$, we have $[[x,y]_{q},z]_{q}=(-1)^{p(\beta)p(\gamma)}[[x,z]_{q},y]_{q}$.
\end{enumerate}
\end{corollary}
\par
By using Corollary \ref{qcom jacobi coro}(2) for $y=z=e_{i}{\ }(\alpha_i\in\prir)$, we obtain $[[x,e_{i}]_{q},e_{i}]_{q}=0$.
This suggests the Serre relation (\ref{serre qcom}) actually holds even when $\alpha_{i}\in\prir$:
\begin{corollary}\label{generalized serre}
We set $e_{i},e_{j}$ satisfying $a_{ij}\neq 0$ and $i\neq j$.
Then, we have
\begin{align}
\begin{alignedat}{2}
&[[e_{j},e_{i}]_{q},e_{i}]_{q}
=[e_{i},[e_{i},e_{j}]_{q}]_{q}
=0
&&\quad (|a_{ij}|=1)
,\\
&[[[e_{j},e_{i}]_{q},e_{i}]_{q},e_{i}]_{q}
=[e_{i},[e_{i},[e_{i},e_{j}]_{q}]_{q}]_{q}
=0
&&\quad (|a_{ij}|=2)
.
\end{alignedat}
\end{align}
\end{corollary}
We also use the following relations for quantum root vectors.
Lemma \ref{qroot general lemma} is given in Lemma 5.2.1.(iii) and Remark 5.2.2.(i) of \cite{Yam94}.
\begin{lemma}\label{qroot general lemma}
We consider the quantum root vectors $e_{\alpha}$ under the order $O_1$.
\begin{enumerate}[(1)]
\item
For $\alpha\in\prir$, we have $e_{\alpha}^{2}=0$.
\item
Let $\alpha\in\prr$ satisfy $c_{\alpha}=1$, where $c_{\alpha}$ is given by the above of (\ref{no1}).
Take $\alpha_{i}$ satisfying $g(\alpha)<i$ and $\alpha+\alpha_{i}\notin\prr$.
We then have
\begin{align}
[e_{\alpha},e_{i}]_{q}=0
.
\end{align}
\end{enumerate}
\end{lemma}
\section{Tetrahedron equation and 3D reflection equation}\label{sec 3}
\subsection{Tetrahedron equation}\label{te sec}
\begin{figure}[ht]
\centering
\caption{}
\includegraphics{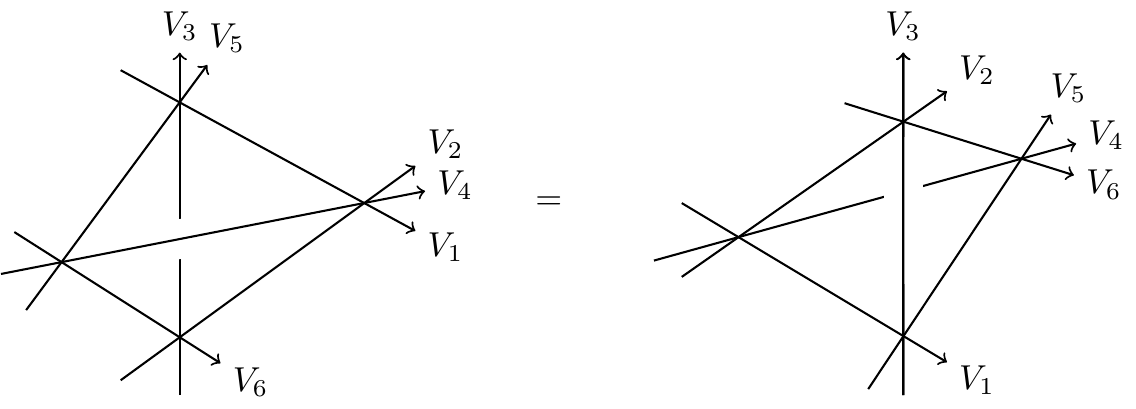}
\label{te fig}
\end{figure}
In this section, we summarize solutions to the tetrahedron and 3D reflection equation, which are related to transition matrices of PBW bases of the nilpotent subalgebra of quantum non-super algebras.
For the history of them, see Section \ref{sec 1}.
Here, we consider the tetrahedron equation\cite{Zam80}, which is a three dimensional analog of the Yang-Baxter equation\cite{Bax07}.
The equation is defined on the tensor product of six linear spaces, and pictorially represented as Figure \ref{te fig}, where $V_i$ are some linear spaces, specified below.
\par
In this paper, we focus on infinite-dimensional solutions on the Fock spaces.
Let $F=\bigoplus_{m=0,1,2,\cdots}\mathbb{C}\ket{m}$ be the bosonic Fock space.
We set $\tdr(q)\in\mathrm{End}(F\otimes F\otimes F)$\cite{KV94}\footnote{The formula given in \cite{KV94} involves misprints unfortunately.} by
\begin{align}
&\tdr(q)(\ket{i}\otimes \ket{j}\otimes \ket{k})
=\sum_{a,b,c\in\mathbb{Z}_{\geq 0}}
\tdr(q)_{i,j,k}^{a,b,c}\ket{a}\otimes \ket{b}\otimes \ket{c}
,
\label{3dR}
\\
&\tdr(q)_{i,j,k}^{a,b,c}
=
\delta_{i+j}^{a+b}
\delta_{j+k}^{b+c}
\sum_{\underset{\scriptstyle\lambda+\mu=b}{\lambda,\mu\in\mathbb{Z}_{\geq 0}}}
(-1)^\lambda q^{i(c-j)+(k+1)\lambda+\mu(\mu-k)}
\frac{(q^2)_{c+\mu}}{(q^2)_c}
\begin{pmatrix}
i \\
\mu
\end{pmatrix}_{\!\! q^2}
\begin{pmatrix}
j \\
\lambda
\end{pmatrix}_{\!\! q^2}
,
\label{3dR mat el}
\end{align}
where we use $\delta_{x}^{y}=\delta_{x,y}$ and the $q$-factorial and $q$-binomial:
\begin{align}
(q)_m=\prod_{k=1}^m(1-q^{k})
,\quad
\begin{pmatrix}l \\m\end{pmatrix}_{\!\! q}
=
\begin{cases}
\displaystyle \frac{(q)_l}{(q)_{l-m}(q)_{m}} \quad &(0\leq m\leq l), \\
{\ \ \ \ \ \ \ } 0 &(\mathrm{otherwise}).
\end{cases}
\end{align}
Summations in (\ref{3dR}) are actually finite due to $\delta_{i+j}^{a+b}\delta_{j+k}^{b+c}$ in (\ref{3dR mat el}).
This is also the same for other 3D operators we will introduce later.
For simplicity, we also use the abbreviated notation $\tdr=\tdr(q)$.
Then, the matrix $\tdr$ satisfies the following tetrahedron equation without a spectral parameter:
\begin{align}
\tdr_{123}\tdr_{145}\tdr_{246}\tdr_{356}=\tdr_{356}\tdr_{246}\tdr_{145}\tdr_{123}
,
\label{KV94 te}
\end{align}
where indices represent the tensor components on which each matrix acts non-trivially.
We simply call $\tdr$ as the 3D R.
The 3D R gives the transition matrix of the PBW bases of the nilpotent subalgebra of the quantum algebra $\UtAp[3]$ associated with the Dynkin diagram \ddtwoA{\cc}{\cc}.
See Theorem \ref{KOY13 thm 3dR} for a more detailed statement.
\par
On the one hand, it is known that there is another solution on the Fock spaces\cite{BS06}.
Let $V=\bigoplus_{m=0,1}\mathbb{C}u_m$ be the fermionic Fock space.
We set $\tdlq\in\mathrm{End}(V\otimes V\otimes F)$ by
\begin{align}
&\tdlq(u_i\otimes u_j\otimes \ket{k})
=\sum_{a,b\in\{0,1\},c\in\mathbb{Z}_{\geq 0}}
\tdlq_{i,j,k}^{a,b,c}
u_a\otimes u_b\otimes \ket{c}
,
\label{3dL}
\\
\begin{split}
&\tdlq_{0,0,k}^{0,0,c}=\tdlq_{1,1,k}^{1,1,c}=\delta_{k,c}
, \quad \tdlq_{0,1,k}^{0,1,c}=-\delta_{k,c}q^{k+1}
, \quad \tdlq_{1,0,k}^{1,0,c}=\delta_{k,c}q^{k}
, \\
&\tdlq_{1,0,k}^{0,1,c}=\delta_{k-1,c}(1-q^{2k})
, \quad \tdlq_{0,1,k}^{1,0,c}=\delta_{k+1,c},
\end{split}
\label{3dL mat el}
\end{align}
where $\tdlq_{i,j,k}^{a,b,c}=0$ other than (\ref{3dL mat el}).
For simplicity, we also use the abbreviated notation $\tdl=\tdlq$.
The matrix $\tdl$ together with the 3D R satisfies the following tetrahedron equation\cite{BS06}:
\begin{align}
\tdl_{123}\tdl_{145}\tdl_{246}\tdr_{356}=\tdr_{356}\tdl_{246}\tdl_{145}\tdl_{123}
.
\label{BS06 te}
\end{align}
We simply call $\tdl$ as the 3D L.
Although the original tetrahedron equation obtained in \cite{BS06} involves some parameters, the above equation (\ref{BS06 te}) is equivalent to it because they are actually cancelled out as remarked in \cite{BS06}.
Actually, our 3D L concides with the original one up to a gauge transformation by some diagonal matrix.
\par
Unlike the 3D R, the 3D L has lacked an algebraic origin in terms of established \textit{quantum algebras} although they exhibit quite parallel results for their reduction to matrix product solutions to the Yang-Baxter equation as we explained in Section \ref{sec 11}.
In Theorem \ref{my 3dL result}, we will derive the 3D L as the transition matrix of the PBW bases of the nilpotent subalgebra of the quantum superalgebra $\UtAp[2|1]$ associated with the Dynkin diagram \ddtwoA{\cc}{\cct}, which clarifies a parallel origin for the 3D L to the 3D R.
\par
As a relation for them, the following lemma is used for the proof of Theorem \ref{my 3dL result}:
\begin{lemma}\label{RL similar}
\begin{align}
\begin{split}
\tdr_{0,0,k}^{0,0,c}=\tdl_{0,0,k}^{0,0,c}
, \quad \tdr_{0,1,k}^{0,1,c}=\tdl_{0,1,k}^{0,1,c}
, \quad \tdr_{1,0,k}^{1,0,c}=\tdl_{1,0,k}^{1,0,c}
, \quad \tdr_{1,0,k}^{0,1,c}=\tdl_{1,0,k}^{0,1,c}
, \quad \tdr_{0,1,k}^{1,0,c}=\tdl_{0,1,k}^{1,0,c}
.
\end{split}
\end{align}
\end{lemma}
For later use, in addition to the 3D R and 3D L, we define $\tdmq\in\mathrm{End}(F\otimes V\otimes V)$ by
\begin{align}
&\tdmq(\ket{i}\otimes u_j\otimes u_k)
=\sum_{a\in\mathbb{Z}_{\geq 0},b,c\in\{0,1\}}
\tdmq_{i,j,k}^{a,b,c}
\ket{a}\otimes u_b\otimes u_c
,
\label{3dM}
\\
&\tdmq_{i,j,k}^{a,b,c}
=\tdlq_{k,j,i}^{c,b,a}
.
\label{3dM mat el}
\end{align}
We simply call $\tdmq$ as the 3D M.
For simplicity, we also use the abbreviated notation $\tdm=\tdmq$.
\begin{remark}\label{KO12 obs rem}
As we explained in Section \ref{sec 12}, the tetraehdron equation (\ref{BS06 te}) was derived again in several ways after \cite{BS06}.
Here, we explain the details of the derivation by \cite{KO12}.
The $q$-boson algebra $\boson$ is an associative algebra over $\mathbb{C}$ generated by $\{\mathbf{a}^{\pm},\mathbf{k}^{\pm 1}\}$ satisfying the following relations:
\begin{align}
\mathbf{k}\mathbf{a}^\pm=q^{\pm 1}\mathbf{a}^\pm\mathbf{k}
,\quad \mathbf{a}^-\mathbf{a}^+=\mathbf{1}-q^2\mathbf{k}^2
,\quad \mathbf{a}^+\mathbf{a}^-=\mathbf{1}-\mathbf{k}^2
.
\end{align}
It has a representation on $F$ as follows:
\begin{align}
\mathbf{k}\ket{m}=q^{m}\ket{m}
,\quad \mathbf{a}^+\ket{m}=\ket{m+1}
,\quad \mathbf{a}^-\ket{m}=(1-q^{2m})\ket{m-1}
.
\end{align}
Then, the intertwining relations of the 3D R are given by
\begin{align}
\begin{split}
&\tdr(\mathbf{a}^{\pm}\otimes \mathbf{k}\otimes 1)
=(\mathbf{a}^{\pm}\otimes 1\otimes \mathbf{k}+\mathbf{k}\otimes \mathbf{a}^{\pm}\otimes \mathbf{a}^{\mp})\tdr, \\
&\tdr(1\otimes \mathbf{k}\otimes \mathbf{a}^{\pm})
=(\mathbf{k}\otimes 1\otimes \mathbf{a}^{\pm}+\mathbf{a}^{\mp}\otimes \mathbf{a}^{\pm}\otimes \mathbf{k})\tdr, \\
&\tdr(1\otimes \mathbf{a}^{\pm}\otimes 1)
=(\mathbf{a}^{\pm}\otimes 1\otimes \mathbf{a}^{\pm}-q\mathbf{k}\otimes \mathbf{a}^{\pm}\otimes \mathbf{k})\tdr, \\
&\tdr(\mathbf{a}^{+}\otimes \mathbf{a}^{-}\otimes \mathbf{a}^{+}-q\mathbf{k}\otimes 1\otimes \mathbf{k})
=(\mathbf{a}^{-}\otimes \mathbf{a}^{+}\otimes \mathbf{a}^{-}-q\mathbf{k}\otimes 1\otimes \mathbf{k})\tdr, \\
&[\tdr,\mathbf{k}\otimes \mathbf{k}\otimes 1]=[\tdr,1\otimes \mathbf{k}\otimes \mathbf{k}]=0
.
\end{split}
\label{3dR defining rel}
\end{align}
As we explained in Section \ref{sec 11}, the 3D R is uniquely characterized by (\ref{3dR defining rel}) up to the normalization.
(\ref{3dR mat el}) is obtained by choosing $\tdr(\ket{0}\otimes\ket{0}\otimes\ket{0})=\ket{0}\otimes\ket{0}\otimes\ket{0}$.
On the other hand, matrix elements of the 3D L can be expressed by using $\tdl_{i,j}^{a,b}\in\boson{\ }(i,j,a,b\in\{0,1\})$ defined by
\begin{align}
&\tdl(u_i\otimes u_j\otimes \ket{k})
=\sum_{a,b\in\{0,1\}}
u_a\otimes u_b\otimes \tdl_{i,j}^{a,b}\ket{k}
,
\label{3dL op}
\\
&\tdl_{0,0}^{0,0}=\tdl_{1,1}^{1,1}=1
,\quad
\tdl_{0,1}^{0,1}=-q\mathbf{k}
,\quad
\tdl_{1,0}^{1,0}=\mathbf{k}
,\quad
\tdl_{1,0}^{0,1}=\mathbf{a}^{-}
,\quad
\tdl_{0,1}^{1,0}=\mathbf{a}^{+}
,
\label{3dL op mat el}
\end{align}
where $\tdl_{i,j}^{a,b}=0$ other than (\ref{3dL op mat el}).
Therefore, we can consider the 3D L as an operator-valued $4\times 4$ matrix, and the tetrahedron equation (\ref{BS06 te}) as an operator-valued $8\times 8$ matrix.
The key observation of \cite{KO12} is each matrix element of the operator-valued equation (\ref{BS06 te}) exactly corresponds to an intertwing relation of (\ref{3dR defining rel}).
That is, the tetrahedron equation (\ref{BS06 te}) is equivalent to the set of intertwining relations of the irreducible representations of $A_q(A_2)$.
This is an interesting connection but quite mysterious.
Also, this connection gives a derivation of the tetrahedron equation (\ref{BS06 te}) but the \textit{algebraic origin} of 3D L has been unclear.
\end{remark}
\begin{remark}\label{LLMM remark}
There is another known solution to the tetrahedron equation which the 3D L satisfies.
We set $\tdlt\in\mathrm{End}(F\otimes V\otimes V)$
\begin{align}
&\tdlt(\ket{i}\otimes u_j\otimes u_k)
=\sum_{a\in\mathbb{Z}_{\geq 0},b,c\in\{0,1\}}
\tdlt_{i,j,k}^{a,b,c}\ket{a}\otimes u_b\otimes u_c
\label{3dLt}
,\\
&\tdlt_{i,j,k}^{a,b,c}
=\tdl(-q)_{j,k,i}^{b,c,a}
.
\label{3dLt mat el}
\end{align}
Then, the matrix $\tdlt$ together with the 3D L satisfies the following tetrahedron equation:
\begin{align}
\tdlt_{135}\tdlt_{124}\tdl_{456}\tdl_{236}
=\tdl_{236}\tdl_{456}\tdlt_{124}\tdlt_{135}
.
\label{LLMM te}
\end{align}

The equation (\ref{LLMM te}) was first presented by \cite{BS06,Ser06} and obtained again by quantum geometry settings\cite{Ser09,BMS10}.
It plays an important role to show the commutativity of the layer-to-layer transfer matrix associated with the 3D L\cite{Ser06}.
Later, we derive an equation (\ref{my te c4}) which involves only ``the 3D L like objects'' as (\ref{LLMM te}).
Actually, it resembles equation (\ref{LLMM te}), but involves nonlocal sign factors, so we can not write it as a matrix equation like (\ref{LLMM te}).
We do not deal with this issue in this paper, but it is an interesting question whether we can attribute (\ref{LLMM te}) to Corollary \ref{tA r3 mythm4} or not.
\end{remark}
\subsection{3D reflection equation}
We then proceed to explanations of the 3D reflection equation\cite{IK97}, which is a boundary analog of the tetrahedron equation.
The equation is defined on the tensor product of nine linear spaces.
The diagram of the 3D reflection equation is obtained in \cite[Figure 1]{KO13}.
\par
Essentially, there are only two known non-trivial solutions to the 3D reflection equation\cite{KO12,KO13}.
We use the following notation:
\begin{align}
\left\{
\begin{array}{l}
i_1,\cdots,i_r \\
i_1,\cdots,i_s
\end{array}
\right\}
=
\begin{cases}
{\ }\frac{\prod_{k=1}^r(q)_{i_k}}{\prod_{k=1}^s(q)_{j_k}}
\quad &(^\forall i_k,j_k\in\mathbb{Z}_{\geq 0}), \\
{\ }0
\quad &(\mathrm{otherwise}).
\end{cases}
\end{align}
We set $\tdj(q)\in\mathrm{End}(F\otimes F\otimes F\otimes F)$ by
\begin{align}
\tdj(q)
\ket{i}\otimes\ket{j}\otimes\ket{k}\otimes\ket{l}
=\sum_{a,b,c,d\in\mathbb{Z}_{\geq 0}}
\tdj(q)_{i,j,k,l}^{a,b,c,d}\ket{a}\otimes\ket{b}\otimes\ket{c}\otimes\ket{d}
,
\label{3dJ}
\end{align}
\begin{align}
\begin{split}
 \tdj(q)_{i,j,k,l}^{a,b,c,d}
 =&\delta_{i+2j+k}^{a+2b+c}
 \delta_{j+k+l}^{b+c+d}
 \frac{(q^2)_l}{(q^2)_d}
 \sum_{\alpha,\beta,\gamma\in\mathbb{Z}_{\geq 0}}
 \frac{(-1)^{\alpha+\gamma}}{(q^2)_{b-\beta}}
 q^{\tdjPhaseI/2} \\
 &\times
 \tdj(q)_{a+b-\alpha-\beta-\gamma,0,b+c-\alpha-\beta-\gamma,d}^{i+j-\alpha-\beta-\gamma,0,j+k-\alpha-\beta-\gamma,l}
 \left\{
 \begin{array}{c}
 j,b-\beta,j+k-\alpha-\beta,i+j-\alpha-\beta \\
 \alpha,\beta,\gamma,c-\alpha,a-\alpha,j-\alpha-\beta,b-\beta-\gamma
 \end{array}
\right\}
,
\end{split}
\label{3dJ mat el}
\end{align}
where $\tdj(q)_{i,0,k,l}^{a,0,c,d}$ is given by
\begin{align}
\tdj(q)_{i,0,k,l}^{a,0,c,d}
=
\delta_{i+k}^{a+c}
\delta_{k+l}^{c+d}
\sum_{\lambda\in\mathbb{Z}_{\geq 0}}
(-1)^{c+\lambda}
\frac{(q^2)_{d+\lambda}}{(q^2)_d}
q^{\tdjPhaseII/2}
\left\{
\begin{array}{c}
i,k \\
\lambda,i-\lambda,c-\lambda,k-c+\lambda
\end{array}
\right\}
,
\end{align}
and $\tdjPhaseI,\tdjPhaseII$ are given by
\begin{align}
\tdjPhaseI
&=\alpha(\alpha+2b-2\beta-1)+(2\beta-b)(a+b+c)+\gamma(\gamma-1)-j(i+j+k), \\
\tdjPhaseII
&=(l+d+1)(i+c-2\lambda)+c-i
.
\end{align}
For simplicity, we also use the abbreviated notation $\tdj=\tdj(q)$.
We also set $\tdk(q)\in\mathrm{End}(F\otimes F\otimes F\otimes F)$ by
\begin{align}
&\tdk(q)
\ket{i}\otimes\ket{j}\otimes\ket{k}\otimes\ket{l}
=\sum_{a,b,c,d\in\mathbb{Z}_{\geq 0}}
\tdk(q)_{i,j,k,l}^{a,b,c,d}\ket{a}\otimes\ket{b}\otimes\ket{c}\otimes\ket{d}
,
\label{3dK}
\\
&\tdk(q)_{i,j,k,l}^{a,b,c,d}
=\tdj(q^2)_{l,k,j,i}^{d,c,b,a}
.
\label{3dK mat el}
\end{align}
For simplicity, we also use the abbreviated notation $\tdk=\tdk(q)$.
Then, the matrix $\tdj,\tdk$ together with 3D R satisfies the following 3D reflection equations:
\begin{align}
&\tdr_{456}
\tdr_{489}
\tdj_{3579}
\tdr_{269}
\tdr_{258}
\tdj_{1678}
\tdj_{1234}
=
\tdj_{1234}
\tdj_{1678}
\tdr_{258}
\tdr_{269}
\tdj_{3579}
\tdr_{489}
\tdr_{456}
,
\label{KO13 tre}
\\
&\tdr_{456}
\tdr_{489}
\tdk_{3579}
\tdr_{269}
\tdr_{258}
\tdk_{1678}
\tdk_{1234}
=
\tdk_{1234}
\tdk_{1678}
\tdr_{258}
\tdr_{269}
\tdk_{3579}
\tdr_{489}
\tdr_{456}
.
\label{KO12 tre}
\end{align}
We simply call $\tdj,{\ }\tdk$ as the 3D J and 3D K, respectively.
\par
The origin of the 3D K is quite similar to the 3D J as we explained in Section \ref{sec 13}.
That is, the 3D K gives the intertwiner of the irreducible representations of the quantum coordinate ring $A_q(C_2)$, where the associated 3D reflection equation (\ref{KO12 tre}) holds as the identity of the intertwiner of the irreducible representations of the quantum coordinate ring $A_q(C_3)$\cite{KO12}.
As an immediate corollary of the Kuniba-Okado-Yamada theorem, we can see the 3D J also gives the transition matrix of the PBW bases of the nilpotent subalgebra of the quantum algebra $U_q^{+}(C_2)$.
\par
Although the 3D K itself also appears by considering $U_q^{+}(B_2)$ up to $q$-dependence (see Theorem \ref{KOY13 thm 3dJ} and (\ref{tm rel})), it is worth to emphasize that (\ref{KO12 tre}) does not follow from discussions only using type B, that is, it is essentially type C object different from (\ref{KO13 tre}).
In this paper, we focus on the PBW basis for type B, and will give new solutions to the 3D reflection equation, which generalize the solution (\ref{KO13 tre}) to the family of solutions (\ref{tre general}).
\begin{remark}\label{KP18 rem}
Although the 3D R, J and K have similar origins as we mentioned above, unlike the 3D R, the 3D J and 3D K themselves do not give matrix product solutions to the reflection equation because the 3D boundary Zamolodchikov algebra and its associativity condition, i.e. the 3D reflection equation, take different forms.
Nevertheless, it is known that we can obtain matrix product solutions to the reflection equation by arranging the intertwining relations of 3D K into a matrix equation\cite{KP18}:
\begin{align}
\tdl_{123}
\tdg_{24}
\tdl_{215}
\tdg_{16}
\tdk_{3456}
=
\tdk_{3456}
\tdg_{16}
\tdl_{125}
\tdg_{24}
\tdl_{213}
,
\label{qre}
\end{align}
where $\tdl$ is the 3D L with $q\to q^2$ and we introduce a matrix $\tdg$, which gives $K$-matrices in the reflection equation.
Interestingly, this procedure is exactly in the same way as we explained in Remark \ref{KO12 obs rem}.
The equation (\ref{qre}) is called the \textit{quantized} reflection equation\cite{KP18}.
By reducing the equation (\ref{qre}), we get the solutions to the reflection equation associated with the fundamental representations of $U_q(A_{n-1}^{(1)})$, and the spin representations of $U_q(D_{n+1}^{(2)})$, $U_q(B_{n}^{(1)})$ and $U_q(D_{n}^{(1)})$\cite{KP18}.
See \cite{KP18} for more details.
Later, the $K$-matrices are characterized as the interwiners of some coideal subalgebras of the quantum algebras\cite{KOY19b}.
\end{remark}
\section{PBW bases of type A and tetrahedron equation}\label{sec 4}
\subsection{PBW bases of type A}\label{tA qroot rels subsec}
In this section, we focus on quantum superalgebras of type A in the case of rank 2 and 3.
Here, we introduce some notations to briefly describe the PBW bases of the nilpotent subalgebra of them, and show higher-order relations for them.
For the case of type A, there are no anisotropic odd roots.
We then simply write $\prer\cup\prar$ by $\prer$.
We set $e_{ij},e_{(ij)k},e_{i(jk)}\in\UtAp$ by
\begin{align}
e_{ij}
=
[e_{i},e_{j}]_{q}
,\quad
e_{(ij)k}
=
[e_{ij},e_{k}]_{q}
,\quad
e_{i(jk)}
=
[e_{i},e_{jk}]_{q}
,
\label{tA notations}
\end{align}
where $i,j,k\in I$.
By considering Corollary \ref{qcom jacobi coro}(1), we simply write $e_{ijk}=e_{(ij)k}$ for the case $(\alpha_i,\alpha_k)=0$.
We have the following higher-order relations for them:
\begin{proposition}\label{tA qroot rels}
\begin{align}
&e_{(i,i-1),i+1}=(-1)^{p(\alpha_{i-1})p(\alpha_{i+1})}e_{(i,i+1),i-1}
\label{tA qroot rel1}
,\\
&[e_{i-1,i},e_{i+1,i}]=0
\label{tA qroot rel2}
,\\
&[e_{i},e_{i-1,i,i+1}]=0
\label{tA qroot rel3}
,\\
&e_{i}^2e_{i+1,i+2}-(q+q^{-1})e_{i}e_{i+1,i+2}e_{i}+e_{i+1,i+2}e_{i}^2=0
\quad (\alpha_{i}\in\prer)
\label{tA qroot serre 1}
,\\
&e_{i+1,i+2}^2e_{i}-(q+q^{-1})e_{i+1,i+2}e_{i}e_{i+1,i+2}+e_{i}e_{i+1,i+2}^2=0
\quad (\alpha_{i+1}+\alpha_{i+2}\in\prer)
\label{tA qroot serre 2}
,\\
&e_{i+2}^2e_{i+1,i}-(q+q^{-1})e_{i+2}e_{i+1,i}e_{i+2}+e_{i+1,i}e_{i+2}^2=0
\quad (\alpha_{i+2}\in\prer)
\label{tA qroot serre 3}
,\\
&e_{i+1,i}^2e_{i+2}-(q+q^{-1})e_{i+1,i}e_{i+2}e_{i+1,i}+e_{i+2}e_{i+1,i}^2=0
\quad (\alpha_{i}+\alpha_{i+1}\in\prer)
\label{tA qroot serre 4}
,\\
&e_{i,i+1}^2=0
\quad (\alpha_{i}+\alpha_{i+1}\in\prir)
\label{tA qroot serre 5}
.
\end{align}
\end{proposition}
\begin{proof}
(\ref{tA qroot rel1}) is obtained from Corollary \ref{qcom jacobi coro}(2) because $[e_{i-1},e_{i+1}]=0$ and $(\alpha_{i-1},\alpha_{i+1})=0$.
(\ref{tA qroot rel2}) and (\ref{tA qroot rel3}) are obtained by $[e_{i-1,i},e_{i+1,i}]=[e_{i-1,i,i+1},e_{i}]=0$ where we used Corollary \ref{qcom jacobi coro}(1) and Lemma \ref{qroot general lemma}(2).
(\ref{tA qroot serre 5}) is a cororally of Lemma \ref{qroot general lemma}(1).
\par
For (\ref{tA qroot serre 1}) $\sim$ (\ref{tA qroot serre 4}), we only consider (\ref{tA qroot serre 1}) and (\ref{tA qroot serre 2}).
The remaining relations (\ref{tA qroot serre 3}) and (\ref{tA qroot serre 4}) can be proved in the same way.
By using the $q$-commutator, the left hand side of (\ref{tA qroot serre 1}) can be written as $[[e_{i+1,i+2},e_{i}]_{q},e_{i}]_{q}$.
Then we have
\begin{align}
[[e_{i+1,i+2},e_{i}]_{q},e_{i}]_{q}
=(-1)^{p(\alpha_{i})p(\alpha_{i+2})}[[e_{i+1,i},e_{i+2}]_{q},e_{i}]_{q}
=[[e_{i+1,i},e_{i}]_{q},e_{i+2}]_{q}
=0
,
\end{align}
where we used Corollary \ref{qcom jacobi coro}(2) and the Serre relation (\ref{serre qcom}).
Similarly, the left hand side of (\ref{tA qroot serre 2}) can be written as $[e_{i+1,i+2},[e_{i+1,i+2},e_{i}]_{q}]_{q}$.
Then we have
\begin{align}
[e_{i+1,i+2},[e_{i+1,i+2},e_{i}]_{q}]_{q}
=
(-1)^{p(i)p(i+2)}[e_{i+1,i+2},[e_{i+1,i},e_{i+2}]_{q}]_{q}
=0
,
\end{align}
where we first used (\ref{tA qroot rel1}) and then Corollary \ref{qcom jacobi coro}(1), Corollary \ref{generalized serre} and (\ref{tA qroot rel2 2}).
\end{proof}
\begin{proposition}\label{tA qroot rels2}
\begin{align}
&e_{i+1,(i-1,i)}=(-1)^{p(\alpha_{i-1})p(\alpha_{i+1})}e_{i-1,(i+1,i)}
\label{tA qroot rel1 2}
,\\
&[e_{i,i-1},e_{i,i+1}]=0
\label{tA qroot rel2 2}
,\\
&[e_{i},e_{i+1,i,i-1}]=0
\label{tA qroot rel3 2}
,\\
&e_{i}^2e_{i+2,i+1}-(q+q^{-1})e_{i}e_{i+2,i+1}e_{i}+e_{i+2,i+1}e_{i}^2=0
\quad (\alpha_{i}\in\prer)
\label{tA qroot serre 1 2}
,\\
&e_{i+2,i+1}^2e_{i}-(q+q^{-1})e_{i+2,i+1}e_{i}e_{i+2,i+1}+e_{i}e_{i+2,i+1}^2=0
\quad (\alpha_{i+1}+\alpha_{i+2}\in\prer)
\label{tA qroot serre 2 2}
,\\
&e_{i+2}^2e_{i,i+1}-(q+q^{-1})e_{i+2}e_{i,i+1}e_{i+2}+e_{i,i+1}e_{i+2}^2=0
\quad (\alpha_{i+2}\in\prer)
\label{tA qroot serre 3 2}
,\\
&e_{i,i+1}^2e_{i+2}-(q+q^{-1})e_{i,i+1}e_{i+2}e_{i,i+1}+e_{i+2}e_{i,i+1}^2=0
\quad (\alpha_{i}+\alpha_{i+1}\in\prer)
\label{tA qroot serre 4 2}
,\\
&e_{i+1,i}^2=0
\quad (\alpha_{i}+\alpha_{i+1}\in\prir)
\label{tA qroot serre 5 2}
.
\end{align}
\end{proposition}
\begin{proof}
By applying the anti-algebra automorphism $\chi$ given by (\ref{anti chi}) on (\ref{tA qroot rel1}) $\sim$ (\ref{tA qroot serre 5}), we obtain the desired results.
\end{proof}
\par
By writing down quantum root vectors given by Definition \ref{qrv} for the case of rank 2, we find they are given by
\begin{align}
&B_{1}:
\quad e_{\beta_1}=e_{1}
,\quad e_{\beta_2}=e_{21}
,\quad e_{\beta_3}=e_{2}
,
\label{tA r2 qroot vec1}
\\
&B_{2}:
\quad e_{\beta_1}=e_{2}
,\quad e_{\beta_2}=e_{12}
,\quad e_{\beta_3}=e_{1}
,
\label{tA r2 qroot vec2}
\end{align}
where $\beta_{t}{\ }(t=1,\cdots,3)$ are the same as Theorem \ref{Yam PBW thm}.
For non-super case, (\ref{tA r2 qroot vec1}) and (\ref{tA r2 qroot vec2}) concide with quantum root vectors given by (\ref{qroot Lus}) with the reduced expressions $w_0=s_{1}s_{2}s_{1},s_{2}s_{1}s_{2}$ of the longest element of the Weyl group, respectively.
\par
Similarly, by writing down quantum root vectors given by Definition \ref{qrv} for the case of rank 3, we find they are given by
\begin{align}
&B_{1}:
\quad e_{\beta_1}=e_{1}
,\quad e_{\beta_2}=e_{21}
,\quad e_{\beta_3}=e_{321}
,\quad e_{\beta_4}=e_{2}
,\quad e_{\beta_5}=e_{32}
,\quad e_{\beta_6}=e_{3}
,
\label{tA r3 qroot vec1}
\\
&B_{2}:
\quad e_{\beta_1}=e_{3}
,\quad e_{\beta_2}=e_{23}
,\quad e_{\beta_3}=e_{2}
,\quad e_{\beta_4}=e_{123}
,\quad e_{\beta_5}=e_{12}
,\quad e_{\beta_6}=e_{1}
,
\label{tA r3 qroot vec2}
\end{align}
where $\beta_{t}{\ }(t=1,\cdots,6)$ are the same as Theorem \ref{Yam PBW thm}.
For non-super case, (\ref{tA r3 qroot vec1}) and (\ref{tA r3 qroot vec2}) concide with quantum root vectors given by (\ref{qroot Lus}) with the reduced expressions $w_0=s_{1}s_{2}s_{3}s_{1}s_{2}s_{1},s_{3}s_{2}s_{3}s_{1}s_{2}s_{3}$ of the longest element of the Weyl group, respectively.
\subsection{Transition matrices of PBW bases of type A of rank 2}\label{sec 42}
In this section, we consider transition matrices of the PBW bases of $\UtAp$ of rank 2, so $m+n=3$.
All possible Dynkin diagrams associated with admissible realizations are given in Table \ref{tA all dynkin}.
\begin{table}[htb]
\centering
\caption{}
\label{tA all dynkin}
\begin{tabular}{c|c}\hline
$\GG$ & Dynkin diagram \\\hline
$\tA[3|0]$ & \ddtwoWul{\ltA}{\cc}{\cc}{\epsilon_1-\epsilon_2}{\epsilon_2-\epsilon_3} \\\hline
$\tA[2|1]$ & \ddtwoWul{\ltA}{\cc}{\cct}{\epsilon_1-\epsilon_2}{\epsilon_2-\delta_3}
\quad
\ddtwoWul{\ltA}{\cct}{\cct}{\epsilon_1-\delta_2}{\delta_2-\epsilon_3}
\quad
\ddtwoWul{\ltA}{\cct}{\cc}{\delta_1-\epsilon_2}{\epsilon_2-\epsilon_3} \\\hline
$\tA[1|2]$ & \ddtwoWul{\ltA}{\cc}{\cct}{\delta_1-\delta_2}{\delta_2-\epsilon_3}
\quad
\ddtwoWul{\ltA}{\cct}{\cct}{\delta_1-\epsilon_2}{\epsilon_2-\delta_3}
\quad
\ddtwoWul{\ltA}{\cct}{\cc}{\epsilon_1-\delta_2}{\delta_2-\delta_3} \\\hline
$\tA[0|3]$ & \ddtwoWul{\ltA}{\cc}{\cc}{\delta_1-\delta_2}{\delta_2-\delta_3} \\\hline
\end{tabular}
\end{table}
In Table \ref{tA all dynkin}, $(\Pi,p)$ associated with same Dynkin diagrams are exactly same.
We then only consider quantum superalgebras associated with the following Dynkin diagrams given by (\ref{tA our dynkin})
\begin{equation}
\begin{alignedat}{4}
&\mathrm{(I)}&{\ }&\ddtwoWul{\ltA}{\cc}{\cc}{\epsilon_1-\epsilon_2}{\epsilon_2-\epsilon_3}
&\qquad
&\mathrm{(II)}&{\ }&\ddtwoWul{\ltA}{\cc}{\cct}{\epsilon_1-\epsilon_2}{\epsilon_2-\delta_3}
\\
&\mathrm{(III)}&{\ }&\ddtwoWul{\ltA}{\cct}{\cc}{\epsilon_1-\delta_2}{\delta_2-\delta_3}
&\qquad
&\mathrm{(IV)}&{\ }&\ddtwoWul{\ltA}{\cct}{\cct}{\epsilon_1-\delta_2}{\delta_2-\epsilon_3}
\label{tA our dynkin}
\end{alignedat}
\end{equation}
where they are distinguished except (IV), in the sense defined in Section \ref{sec 22}.
For the case of rank 2, quantum root vecotrs are given by (\ref{tA r2 qroot vec1}) and (\ref{tA r2 qroot vec2}), so the transition matrix in (\ref{tm1}) is given as follows:
\begin{align}
e_{2}^{(a)}e_{12}^{(b)}e_{1}^{(c)}
&=\sum_{i,j,k}
\tm_{i,j,k}^{a,b,c}
e_{1}^{(k)}e_{21}^{(j)}e_{2}^{(i)}
,
\label{tA trans mat def1}
\end{align}
where the domain of indices is specified below.
Hereafter, we consider each case.
Sometimes, we abbreviate simple roots for Dynkin diagrams, but we always assume that they are given as (\ref{tA our dynkin}).
\subsubsection{The case (I) $\vcenter{\hbox{\protect\includegraphics{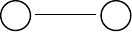}}}$}
In this case, the corresponding symmetrized Cartan matrix is given by
\begin{align}
DA
=
\begin{pmatrix}
2 & -1 \\
-1 & 2
\end{pmatrix}
,
\end{align}
and the corresponding positive roots are given by
\begin{align}
\prer&=\{\alpha_1,\alpha_2,\alpha_1+\alpha_2\}
,
\\
\prir&=\{\}
.
\end{align}
Then, indices are specified as $i,j,k,a,b,c\in\mathbb{Z}_{\geq 0}$ for (\ref{tA trans mat def1}).
The transition matrix in (\ref{tA trans mat def1}) is explicitly given as the consequence of the Kuniba-Okado-Yamada theorem\cite{KOY13}:
\begin{theorem}[\cite{Ser08,KOY13}]\label{KOY13 thm 3dR}
For the quantum superalgebra associated with \ddtwoA{\cc}{\cc}, the transition matrix in (\ref{tA trans mat def1}) is given by
\begin{align}
\tm_{i,j,k}^{a,b,c}
=\tdr_{i,j,k}^{a,b,c}
\label{KOY13 3dR result}
,
\end{align}
where $\tdr$ is the 3D R given by (\ref{3dR mat el}).
\end{theorem}
\subsubsection{The case (II) $\vcenter{\hbox{\protect\includegraphics{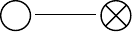}}}$}
In this case, the corresponding symmetrized Cartan matrix is given by
\begin{align}
DA
=
\begin{pmatrix}
2 & -1 \\
-1 & 0
\end{pmatrix}
,
\end{align}
and the corresponding positive roots are given by
\begin{align}
\prer&=\{\alpha_1\}
,
\\
\prir&=\{\alpha_2,\alpha_1+\alpha_2\}
.
\end{align}
Then, indices are specified as $i,j,a,b\in\{0,1\},{\ }k,c\in\mathbb{Z}_{\geq 0}$ for (\ref{tA trans mat def1}).
The transition matrix in (\ref{tA trans mat def1}) is explicitly given as follows:
\begin{theorem}\label{my 3dL result}
For the quantum superalgebra associated with \ddtwoA{\cc}{\cct}, the transition matrix in (\ref{tA trans mat def1}) is given by
\begin{align}
\tm_{i,j,k}^{a,b,c}
=\tdl_{i,j,k}^{a,b,c}
,
\end{align}
where $\tdl$ is the 3D L given by (\ref{3dL mat el}).
\end{theorem}
\begin{proof}
Multiplying both sides of (\ref{tA trans mat def1}) by $k_2$ from left and $k_2^{-1}$ from right, also using the relation (\ref{qs rel1}), we obtain
\begin{align}
e_{2}^{(a)}e_{12}^{(b)}e_{1}^{(c)}
&=\sum_{i,j,k}
q^{b+c-j-k}\tm_{i,j,k}^{a,b,c}
e_{1}^{(k)}e_{21}^{(j)}e_{2}^{(i)}
.
\label{my 3dL result proof1}
\end{align}
On the other hand, $\tm_{i,j,k}^{a,b,c}=q^{b+c-j-k}\tm_{i,j,k}^{a,b,c}$ holds becasuse $\{e_{1}^{(k)}e_{21}^{(j)}e_{2}^{(i)}\}$ are linearly independent by Theorem \ref{Yam PBW thm}.
This means, if $\tm_{i,j,k}^{a,b,c}\neq 0$, $b+c=j+k$ holds.
Similarly, multiplying both sides of (\ref{tA trans mat def1}) by $k_1$ from left and $k_1^{-1}$ from right, also using the relation (\ref{qs rel1}), we obtain $-a+b+2c=-i+j+2k$ if $\tm_{i,j,k}^{a,b,c}\neq 0$.
Combining them, we eventually obtain the following weight conservation:
\begin{align}
\tm_{i,j,k}^{a,b,c}= 0
\quad
(i+j\neq a+b\quad \mathrm{or}\quad j+k\neq b+c)
.
\label{my 3dL result proof2}
\end{align}
\par
Next, we consider (\ref{tA trans mat def1}) for the cases $(a,b)=(0,0),(0,1),(1,0)$.
For these case, the degreee of $e_2$ is at most 1 in both sides of (\ref{tA trans mat def1}) thanks to the weight conservation (\ref{my 3dL result proof2}).
Now, the relations $e_1,e_2$ satisfy are
\begin{align}
e_{1}^2e_{2}
-(q+q^{-1})e_{1}e_{2}e_{1}
+e_{2}e_{1}^{2}
=0
,\quad
e_{2}^2=0
.
\label{my 3dL result proof3}
\end{align}
Therefore, the only relation one can apply on both sides of (\ref{tA trans mat def1}) is the first relation of (\ref{my 3dL result proof3}) for the cases $(a,b)=(0,0),(0,1),(1,0)$.
The first relation of (\ref{my 3dL result proof3}) is the same as one of the case (I) \ddtwoA{\cc}{\cc}, so by Lemma \ref{RL similar}, we obtain $\tm_{i,j,k}^{a,b,c}=\tdl_{i,j,k}^{a,b,c}$ for the cases $(a,b)=(0,0),(0,1),(1,0)$ and $i+j=a+b,{\ }j+k=b+c$ are satisfied.
\par
Eventually, it is sufficient to show that (\ref{tA trans mat def1}) for the cases $(a,b)=(1,1)$
\begin{align}
e_2e_1e_2e_1^{c}
=\tm_{1,1,c}^{1,1,c}
e_1^{c}e_2e_1e_2
\quad (c\in\mathbb{Z}_{\geq 0})
,
\label{my 3dL result proof4}
\end{align}
holds for $\tm_{1,1,c}^{1,1,c}=\tdl_{1,1,c}^{1,1,c}=1$, where we used the weight conservation (\ref{my 3dL result proof2}) and $e_{2}^2=0$, and multiplied both sides by $[c]_{q}!$.
Actually, we can prove (\ref{my 3dL result proof4}) by induction as follows.
When $c=0$, (\ref{my 3dL result proof4}) trivially holds for $\tm_{1,1,0}^{1,1,0}=1$.
Let us suppose (\ref{my 3dL result proof4}) is true for $c=n$ with $\tm_{1,1,n}^{1,1,n}=1$.
Then, we obtain
\begin{align}
e_{2}e_{1}e_{2}e_{1}^{n+1}
=e_{1}^{n}e_{2}e_{1}e_{2}e_{1}
=\frac{1}{q+q^{-1}}
e_{1}^{n}e_{2}e_{1}^2e_{2}
=e_{1}^{n+1}e_{2}e_{1}e_{2}
,
\label{my 3dL result proof5}
\end{align}
where we used (\ref{my 3dL result proof3}).
Thus, (\ref{my 3dL result proof4}) holds for $c=n+1$ with $\tm_{1,1,n+1}^{1,1,n+1}=1$.
To sum up the above discussion, we then obtain $\tm_{i,j,k}^{a,b,c}=\tdl_{i,j,k}^{a,b,c}$.
\end{proof}
\begin{corollary}\label{my inv cor}
\begin{align}
\tdl^{-1}=\tdl
.
\end{align}
\end{corollary}
\begin{proof}
By using (\ref{tA trans mat def1}) and (\ref{tm rel}), we obtain
\begin{align}
 e_{2}^{(a)}e_{12}^{(b)}e_{1}^{(c)}
 =\sum_{i,j,k}
 \tdl_{i,j,k}^{a,b,c}
 e_{1}^{(k)}e_{21}^{(j)}e_{2}^{(i)}
 &=\sum_{i,j,k}\sum_{x,y,z}
 \tdl_{i,j,k}^{a,b,c}
 \tdm_{x,y,z}^{k,j,i}
 e_{2}^{(z)}e_{12}^{(y)}e_{1}^{(x)}
 \\
 &=\sum_{i,j,k}\sum_{x,y,z}
 \tdl_{i,j,k}^{a,b,c}
 \tdl_{z,y,x}^{i,j,k}
 e_{2}^{(z)}e_{12}^{(y)}e_{1}^{(x)}
 ,
\label{inv cor proof1}
\end{align}
Here, we omit the domain of indices but it is easily specified.
Since $\{e_{2}^{(a)}e_{12}^{(b)}e_{1}^{(c)}\}$ are linearly independent by Theorem \ref{Yam PBW thm}, we obtain
\begin{align}
\sum_{i,j,k}\sum_{x,y,z}
 \tdl_{i,j,k}^{a,b,c}
 \tdl_{z,y,x}^{i,j,k}
=\delta_{a,z}\delta_{b,y}\delta_{c,x}
.
\end{align}
This finishes the proof.
\end{proof}
\subsubsection{The case (III) $\vcenter{\hbox{\protect\includegraphics{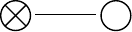}}}$}
In this case, the corresponding symmetrized Cartan matrix is given by
\begin{align}
DA
=
\begin{pmatrix}
0 & -1 \\
-1 & 2
\end{pmatrix}
,
\end{align}
and the corresponding positive roots are given by
\begin{align}
\prer&=\{\alpha_2\}
,
\\
\prir&=\{\alpha_1,\alpha_1+\alpha_2\}
.
\end{align}
Then, indices are specified as $j,k,b,c\in\{0,1\},{\ }i,a\in\mathbb{Z}_{\geq 0}$ for (\ref{tA trans mat def1}).
The transition matrix in (\ref{tA trans mat def1}) is explicitly given as follows:
\begin{corollary}\label{my 3dL result 2}
For the quantum superalgebra associated with \ddtwoA{\cct}{\cc}, the transition matrix in (\ref{tA trans mat def1}) is given by
\begin{align}
\tm_{i,j,k}^{a,b,c}
=\tdm_{i,j,k}^{a,b,c}
,
\end{align}
where $\tdm$ are the 3D M given by (\ref{3dM mat el}).
\end{corollary}
\begin{proof}
$h:\UtAp[2|1]\to\UtAp[1|2]$ defined by $e_{1}\mapsto e_{2}$, $e_{2}\mapsto e_{1}$ gives an algebra homomorphism, where the former algebra is associated with \ddtwoA{\cc}{\cct} and the latter is associated with \ddtwoA{\cct}{\cc}.
Then, by Theorem \ref{my 3dL result} and (\ref{tm rel}), it is easy to see that the statement holds.
\end{proof}
\subsubsection{The case (IV) $\vcenter{\hbox{\protect\includegraphics{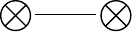}}}$}
In this case, the corresponding symmetrized Cartan matrix is given by
\begin{align}
DA
=
\begin{pmatrix}
0 & -1 \\
-1 & 0
\end{pmatrix}
,
\end{align}
and the corresponding positive roots are given by
\begin{align}
\prer&=\{\alpha_1+\alpha_2\}
,
\\
\prir&=\{\alpha_1,\alpha_2\}
.
\end{align}
Then, indices are specified as $i,k,a,c\in\{0,1\},{\ }j,b\in\mathbb{Z}_{\geq 0}$ for (\ref{tA trans mat def1}).
We set $\tdn(q)\in\mathrm{End}(V\otimes F\otimes V)$ by
\begin{align}
&\tdn(q)(u_i\otimes \ket{j}\otimes u_k)
=\sum_{a,c\in\{0,1\},b\in\mathbb{Z}_{\geq 0}}
\tdn(q)_{i,j,k}^{a,b,c}
u_a\otimes \ket{b}\otimes u_c
,
\label{3dN}
\\
\begin{split}
&\tdn(q)_{0,j,0}^{0,b,0}=\delta_{j,b}q^{j}
,\quad
\tdn(q)_{1,j,1}^{1,b,1}=-\delta_{j,b}q^{j+1}
,\quad
\tdn(q)_{0,j,1}^{0,b,1}=\tdn(q)_{1,j,0}^{1,b,0}=\delta_{j,b}
,\\
&\tdn(q)_{1,j,1}^{0,b,0}=\delta_{j+1,b}q^{j}(1-q^{2})
,\quad
\tdn(q)_{0,j,0}^{1,b,1}=\delta_{j-1,b}[j]_q
,
\end{split}
\label{3dN mat el}
\end{align}
where $\tdn_{i,j,k}^{a,b,c}=0$ other than (\ref{3dN mat el}).
For simplicity, we also use the abbreviated notation $\tdn=\tdn(q)$.
We simply call $\tdn$ as the 3D N.
Then, the transition matrix in (\ref{tA trans mat def1}) is explicitly given as follows:
\begin{theorem}\label{my 3dN result}
For the quantum superalgebra associated with \ddtwoA{\cct}{\cct}, the transition matrix in (\ref{tA trans mat def1}) is given by
\begin{align}
\tm_{i,j,k}^{a,b,c}
=\tdn_{i,j,k}^{a,b,c}
,
\end{align}
where $\tdn$ is the 3D N given by (\ref{3dN mat el}).
\end{theorem}
\begin{proof}
By the same discussion as (\ref{my 3dL result proof2}), we obtain the following weight conservation:
\begin{align}
\tm_{i,j,k}^{a,b,c}=0
\quad
(i+j\neq a+b \quad \mathrm{or}\quad j+k\neq b+c)
.
\label{my 3dN result proof1}
\end{align}
Now, the relations $e_1,e_2$ satisfy are $e_{1}^2=e_{2}^2=0$ and $e_{12},e_{21}$ are given by $e_{12}=e_{1}e_{2}+qe_{2}e_{1}$, $e_{21}=e_{2}e_{1}+qe_{1}e_{2}$.
We introduce the following notation:
\begin{align}
\mu_1(n)=\overbrace{e_{1}e_{2}\cdots e_{1}e_{2}}^{2n}
,\quad
\mu_2(n)=\overbrace{e_{2}e_{1}\cdots e_{2}e_{1}}^{2n}
.
\label{my 3dN result proof2}
\end{align}
We then explicitly write down $e_{12}^b,e_{21}^{j}$ as
\begin{align}
e_{12}^{b}=\mu_1(b)+q^{b}\mu_2(b)
,\quad
e_{21}^{j}=\mu_2(j)+q^{j}\mu_1(j)
,
\label{my 3dN result proof3}
\end{align}
Actually, they are easily shown by induction.
Hereafter, we consider each case for values of $(a,c)$ in (\ref{tA trans mat def1}).
\begin{enumerate}[(i)]
\item\label{my 3dN result proof fist case}
For the case $(a,c)=(0,0)$, by the weight conservation (\ref{my 3dN result proof1}), (\ref{tA trans mat def1}) is explicitly written down as
\begin{align}
\frac{\mu_1(b)+q^{b}\mu_2(b)}{[b]_{q}!}
=
\tm_{0,b,0}^{0,b,0}
\frac{\mu_2(b)+q^{b}\mu_1(b)}{[b]_{q}!}
+\tm_{1,b-1,1}^{0,b,0}
\frac{\mu_1(b)}{[b-1]_{q}!}
,
\label{my 3dN result proof4}
\end{align}
where we use (\ref{my 3dN result proof3}) and $e_{1}^2=e_{2}^2=0$.
By comparing coefficients of each monomial $\mu_{1}(b),\mu_{2}(b)$, we obtain
\begin{align}
\frac{1}{[b]_{q}!}
=\tm_{0,b,0}^{0,b,0}\frac{q^{b}}{[b]_{q}!}+\tm_{1,b-1,1}^{0,b,0}\frac{1}{[b-1]_{q}!}
,\quad
\frac{q^{b}}{[b]_{q}!}
=\tm_{0,b,0}^{0,b,0}\frac{1}{[b]_{q}!}
.
\label{my 3dN result proof5}
\end{align}
It is easy to see that $\gamma_{0,b,0}^{0,b,0}=\tdn_{0,b,0}^{0,b,0}$ and $\gamma_{1,b-1,1}^{0,b,0}=\tdn_{1,b-1,1}^{0,b,0}$ give the solution of (\ref{my 3dN result proof5}).
\item
For the case $(a,c)=(0,1)$, by the weight conservation (\ref{my 3dN result proof1}), (\ref{tA trans mat def1}) is explicitly written down as
\begin{align}
\mu_{1}(b)e_{1}
=\tm_{0,b,1}^{0,b,1}
e_{1}\mu_{2}(b)
,
\label{my 3dN result proof6}
\end{align}
where we used (\ref{my 3dN result proof3}) and $e_{1}^2=e_{2}^2=0$.
We then obtain $\tm_{0,b,1}^{0,b,1}=1=\tdn_{0,b,1}^{0,b,1}$ because $\mu_{1}(b)e_{1}=e_{1}\mu_{2}(b)$.
\item
For the case $(a,c)=(1,0)$, by the weight conservation (\ref{my 3dN result proof1}), (\ref{tA trans mat def1}) is explicitly written down as
\begin{align}
e_{2}\mu_{1}(b)
=\tm_{1,b,0}^{1,b,0}
\mu_{2}(b)e_{2}
,
\label{my 3dN result proof7}
\end{align}
where we used (\ref{my 3dN result proof3}) and $e_{1}^2=e_{2}^2=0$.
We then obtain $\tm_{1,b,0}^{1,b,0}=1=\tdn_{1,b,0}^{1,b,0}$ because $e_{2}\mu_{1}(b)=\mu_{2}(b)e_{2}$.
\item
For the case $(a,c)=(1,1)$, by the weight conservation (\ref{my 3dN result proof1}), (\ref{tA trans mat def1}) is explicitly written down as
\begin{align}
\frac{\mu_2(b+1)}{[b]_{q}!}
=
\tm_{1,b,1}^{1,b,1}
\frac{\mu_1(b+1)}{[b]_{q}!}
+\tm_{0,b+1,0}^{1,b,1}
\frac{\mu_2(b+1)+q^{b+1}\mu_1(b+1)}{[b+1]_{q}!}
,
\label{my 3dN result proof8}
\end{align}
where we used (\ref{my 3dN result proof3}) and $e_{1}^2=e_{2}^2=0$.
Similarly to the case (i), it is easy to see that $\gamma_{1,b,1}^{1,b,1}=\tdn_{1,b,1}^{1,b,1}$ and $\gamma_{0,b+1,0}^{1,b,1}=\tdn_{0,b+1,0}^{1,b,1}$ give the solution of (\ref{my 3dN result proof8}).
\end{enumerate}
\end{proof}
\begin{corollary}\label{my inv cor 3dN}
\begin{align}
\tdn^{-1}=\tdn
.
\end{align}
\end{corollary}
\begin{proof}
This is shown exactly in the same way as Corollary \ref{my inv cor}.
\end{proof}
\begin{remark}\label{3dl 3dn rem}
We find that the matrix elements of the 3D L and 3D N satisfy the following relation:
\begin{align}
\tdn_{i,j,k}^{a,b,c}
=\frac{[j]_{q}!}{[b]_{q}!}\tdl_{1-a,c,b}^{1-i,k,j}
.
\label{3dl 3dn rel}
\end{align}
It is naturally thought that (\ref{3dl 3dn rel}) originates from the fact that both \ddtwoA{\cc}{\cct} and \ddtwoA{\cct}{\cct} give the Dynkin diagrams of $\tA[2|1]$ as Table \ref{tA all dynkin}.
However, the origin of the relation (\ref{3dl 3dn rel}) in terms of the PBW basis is unknown to us.
We do not deal with this issue in this paper, but it is interesting whether, in general, transition matrices associated with a pair of Cartan data mapped to each other via odd reflections are attributed to each other or not.
For example, we will also establish a relation between transition matrices associated with such pair of Cartan data for type B.
See (\ref{my3dY proof9}).
\end{remark}
\subsection{Transition matrices of PBW bases of type A of rank 3 and tetrahedron equation}\label{sec 43}
In this section, we consider the transition matrix of the PBW bases of $\UtAp$ of rank 3, so $m+n=4$.
All possible Dynkin diagrams associated with admissible realizations are given in Table \ref{tA all dynkin r3}.
\begin{table}[htb]
\centering
\caption{}
\label{tA all dynkin r3}
\begin{tabular}{c|c}\hline
$\GG$ & Dynkin diagram \\\hline
$\tA[4|0]$ &  \ddthreeAWul{\cc}{\cc}{\cc}{\ee_1-\ee_2}{\ee_2-\ee_3}{\ee_3-\ee_4}\\\hline
\multirow{3}{*}[-0.8em]{$\tA[3|1]$} &  \ddthreeAWul{\cc}{\cc}{\cct}{\ee_1-\ee_2}{\ee_2-\ee_3}{\ee_3-\dd_4}
\quad \ddthreeAWul{\cc}{\cct}{\cct}{\ee_1-\ee_2}{\ee_2-\dd_3}{\dd_3-\ee_4}
\\
&\\
&  \ddthreeAWul{\cct}{\cct}{\cc}{\ee_1-\dd_2}{\dd_2-\ee_3}{\ee_3-\ee_4}
\quad \ddthreeAWul{\cct}{\cc}{\cc}{\dd_1-\ee_2}{\ee_2-\ee_3}{\ee_3-\ee_4}
\\\hline
\multirow{3}{*}[-2.4em]{$\tA[2|2]$} &  \ddthreeAWul{\cc}{\cct}{\cc}{\ee_1-\ee_2}{\ee_2-\dd_3}{\dd_3-\dd_4}
\quad \ddthreeAWul{\cct}{\cct}{\cct}{\ee_1-\dd_2}{\dd_2-\ee_3}{\ee_3-\dd_4}
\\
&\\
&  \ddthreeAWul{\cct}{\cc}{\cct}{\dd_1-\ee_2}{\ee_2-\ee_3}{\ee_3-\dd_4}
\quad \ddthreeAWul{\cct}{\cc}{\cct}{\ee_1-\dd_2}{\dd_2-\dd_3}{\dd_3-\ee_4}
\\
&\\
&  \ddthreeAWul{\cct}{\cct}{\cct}{\dd_1-\ee_2}{\ee_2-\dd_3}{\dd_3-\ee_4}
\quad \ddthreeAWul{\cc}{\cct}{\cc}{\dd_1-\dd_2}{\dd_2-\ee_3}{\ee_3-\ee_4}
\\\hline
\multirow{3}{*}[-0.8em]{$\tA[1|3]$} &  \ddthreeAWul{\cct}{\cc}{\cc}{\ee_1-\dd_2}{\dd_2-\dd_3}{\dd_3-\dd_4}
\quad \ddthreeAWul{\cct}{\cct}{\cc}{\dd_1-\ee_2}{\ee_2-\dd_3}{\dd_3-\dd_4}
\\
&\\
&  \ddthreeAWul{\cc}{\cct}{\cct}{\dd_1-\dd_2}{\dd_2-\ee_3}{\ee_3-\dd_4}
\quad \ddthreeAWul{\cc}{\cc}{\cct}{\dd_1-\dd_2}{\dd_2-\dd_3}{\dd_3-\ee_4}
\\\hline
$\tA[0|4]$ &  \ddthreeAWul{\cc}{\cc}{\cc}{\dd_1-\dd_2}{\dd_2-\dd_3}{\dd_3-\dd_4}\\\hline
\end{tabular}
\end{table}
In Table \ref{tA all dynkin r3}, $(\Pi,p)$ associated with same Dynkin diagrams are exactly same.
We then only consider the quantum superalgebras associated with the following Dynkin diagrams given by (\ref{tA our dynkin r3}):
\begin{equation}
\begin{alignedat}{4}
&\mathrm{(I)}&{\ }&\ddthreeAWul{\cc}{\cc}{\cc}{\ee_1-\ee_2}{\ee_2-\ee_3}{\ee_3-\ee_4}
&\qquad
&\mathrm{(II)}&{\ }&\ddthreeAWul{\cc}{\cc}{\cct}{\ee_1-\ee_2}{\ee_2-\ee_3}{\ee_3-\dd_4}
\\
&\mathrm{(III)}&{\ }&\ddthreeAWul{\cc}{\cct}{\cct}{\ee_1-\ee_2}{\ee_2-\dd_3}{\dd_3-\ee_4}
&\qquad
&\mathrm{(IV)}&{\ }&\ddthreeAWul{\cc}{\cct}{\cc}{\ee_1-\ee_2}{\ee_2-\dd_3}{\dd_3-\dd_4}
\\
&\mathrm{(V)}&{\ }&\ddthreeAWul{\cct}{\cc}{\cct}{\ee_1-\dd_2}{\dd_2-\dd_3}{\dd_3-\ee_4}
&\qquad
&\mathrm{(VI)}&{\ }&\ddthreeAWul{\cct}{\cct}{\cct}{\ee_1-\dd_2}{\dd_2-\ee_3}{\ee_3-\dd_4}
\label{tA our dynkin r3}
\end{alignedat}
\end{equation}
where (I), (II) and (IV) are distinguished, in the sense defined in Section \ref{sec 22}.
Here, we omit the following Dynkin diagrams given by (\ref{tA our dynkin r3 omit}), because the cases of (VII) and (VIII) are easily attributed to ones of (II) and (III), respectively.
\begin{equation}
\begin{alignedat}{4}
&\mathrm{(VII)}&{\ }&\ddthreeAWul{\cct}{\cc}{\cc}{\ee_1-\dd_2}{\dd_2-\dd_3}{\dd_3-\dd_4}
&\qquad
&\mathrm{(VIII)}&{\ }&\ddthreeAWul{\cct}{\cct}{\cc}{\ee_1-\dd_2}{\dd_2-\ee_3}{\ee_3-\ee_4}
\label{tA our dynkin r3 omit}
\end{alignedat}
\end{equation}
For the case of rank 3, quantum root vectors are given by (\ref{tA r3 qroot vec1}) and (\ref{tA r3 qroot vec2}), so the transition matrix in (\ref{tm1}) is given as follows:
\begin{align}
e_3^{(o_1)}e_{23}^{(o_2)}e_{2}^{(o_3)}e_{123}^{(o_4)}e_{12}^{(o_5)}e_{1}^{(o_6)}
=
\sum_{i_1,i_2,i_3,i_4,i_5,i_6}
\gamma_{i_1,i_2,i_3,i_4,i_5,i_6}^{o_1,o_2,o_3,o_4,o_5,o_6}
e_1^{(i_6)}e_{21}^{(i_5)}e_{321}^{(i_4)}e_{2}^{(i_3)}e_{32}^{(i_2)}e_{3}^{(i_1)}
,
\label{tA trans mat def r3}
\end{align}
where the domain of indices is specified below.
In order to attribute the transition matrix in (\ref{tA trans mat def r3}) to a composition of transition matrices of rank 2, we exploit the following transition matrices $\stm^{(x)}$:
\begin{align}
e_{2}^{(a)}e_{12}^{(b)}e_{1}^{(c)}
&=\sum_{i,j,k}
\stm^{(2|1)}{}_{i,j,k}^{a,b,c}
e_{1}^{(k)}e_{21}^{(j)}e_{2}^{(i)}
,
\label{tA trans mat r3 eq1}
\\
e_{3}^{(a)}e_{23}^{(b)}e_{2}^{(c)}
&=\sum_{i,j,k}
\stm^{(3|2)}{}_{i,j,k}^{a,b,c}
e_{2}^{(k)}e_{32}^{(j)}e_{3}^{(i)}
,
\label{tA trans mat r3 eq2}
\\
e_{23}^{(a)}e_{123}^{(b)}e_{1}^{(c)}
&=\sum_{i,j,k}
\stm^{(23|1)}{}_{i,j,k}^{a,b,c}
e_{1}^{(k)}e_{(23)1}^{(j)}e_{23}^{(i)}
,
\label{tA trans mat r3 eq3}
\\
e_{32}^{(a)}e_{1(32)}^{(b)}e_{1}^{(c)}
&=\sum_{i,j,k}
\stm^{(32|1)}{}_{i,j,k}^{a,b,c}
e_{1}^{(k)}e_{321}^{(j)}e_{32}^{(i)}
,
\label{tA trans mat r3 eq4}
\\
e_{3}^{(a)}e_{123}^{(b)}e_{12}^{(c)}
&=\sum_{i,j,k}
\stm^{(3|12)}{}_{i,j,k}^{a,b,c}
e_{12}^{(k)}e_{3(12)}^{(j)}e_{3}^{(i)}
,
\label{tA trans mat r3 eq5}
\\
e_{3}^{(a)}e_{(21)3}^{(b)}e_{21}^{(c)}
&=\sum_{i,j,k}
\stm^{(3|21)}{}_{i,j,k}^{a,b,c}
e_{21}^{(k)}e_{321}^{(j)}e_{3}^{(i)}
,
\label{tA trans mat r3 eq6}
\end{align}
where the domain of indices will be specified and explicit formulae of $\stm^{(x)}$ are given for each case in (\ref{tA our dynkin r3}).
\par
Then, by using $\stm^{(x)}$, (\ref{tA qroot rel1}) $\sim$ (\ref{tA qroot rel3}) and (\ref{tA qroot rel1 2}) $\sim$ (\ref{tA qroot rel3 2}), we can construct the transition matrix in (\ref{tA trans mat def r3}) in two ways.
The first way is given by
\begin{align}
&\underline{e_3^{(o_1)}e_{23}^{(o_2)}e_{2}^{(o_3)}}e_{123}^{(o_4)}e_{12}^{(o_5)}e_{1}^{(o_6)}
\label{tA lhs1 eq1}
\\
&=
\sum
\stm^{(3|2)}{}_{x_1,x_2,x_3}^{o_1,o_2,o_3}
e_2^{(x_3)}e_{32}^{(x_2)}\underline{e_{3}^{(x_1)}e_{123}^{(o_4)}e_{12}^{(o_5)}}e_{1}^{(o_6)}
\label{tA lhs1 eq2}
\\
&=
\sum
\stm^{(3|2)}{}_{x_1,x_2,x_3}^{o_1,o_2,o_3}
\stm^{(3|12)}{}_{i_1,x_4,x_5}^{x_1,o_4,o_5}
e_2^{(x_3)}
\underline{e_{32}^{(x_2)}e_{12}^{(x_5)}}
{\ }
\underline{e_{3(12)}^{(x_4)}}
{\ }
\underline{e_{3}^{(i_1)}e_{1}^{(o_6)}}
\label{tA lhs1 eq3}
\\
\begin{split}
&=
\sum
(-1)^{\phaseI (i_1o_6+x_4)+\phaseII x_2x_5}
\stm^{(3|2)}{}_{x_1,x_2,x_3}^{o_1,o_2,o_3}
\stm^{(3|12)}{}_{i_1,x_4,x_5}^{x_1,o_4,o_5}
\\
&\spaceD
\times
e_2^{(x_3)}e_{12}^{(x_5)}\underline{e_{32}^{(x_2)}e_{1(32)}^{(x_4)}e_{1}^{(o_6)}}e_{3}^{(i_1)}
\end{split}
\label{tA lhs1 eq4}
\\
\begin{split}
&=
\sum
(-1)^{\phaseI (i_1o_6+x_4)+\phaseII x_2x_5}
\stm^{(3|2)}{}_{x_1,x_2,x_3}^{o_1,o_2,o_3}
\stm^{(3|12)}{}_{i_1,x_4,x_5}^{x_1,o_4,o_5}
\stm^{(32|1)}{}_{i_2,i_4,x_6}^{x_2,x_4,o_6}
\\
&\spaceD
\times
\underline{e_2^{(x_3)}e_{12}^{(x_5)}e_{1}^{(x_6)}}e_{321}^{(i_4)}e_{32}^{(i_2)}e_{3}^{(i_1)}
\end{split}
\label{tA lhs1 eq5}
\\
\begin{split}
&=
\sum
(-1)^{\phaseI (i_1o_6+x_4)+\phaseII x_2x_5}
\stm^{(3|2)}{}_{x_1,x_2,x_3}^{o_1,o_2,o_3}
\stm^{(3|12)}{}_{i_1,x_4,x_5}^{x_1,o_4,o_5}
\stm^{(32|1)}{}_{i_2,i_4,x_6}^{x_2,x_4,o_6}
\stm^{(2|1)}{}_{i_3,i_5,i_6}^{x_3,x_5,x_6}
\\
&\spaceD
\times
e_1^{(i_6)}e_{21}^{(i_5)}\underline{e_{2}^{(i_3)}e_{321}^{(i_4)}}e_{32}^{(i_2)}e_{3}^{(i_1)}
\end{split}
\label{tA lhs1 eq6}
\\
\begin{split}
&=
\sum
(-1)^{\phaseI (i_1o_6+x_4)+\phaseII x_2x_5+\phaseIII i_3i_4}
\stm^{(3|2)}{}_{x_1,x_2,x_3}^{o_1,o_2,o_3}
\stm^{(3|12)}{}_{i_1,x_4,x_5}^{x_1,o_4,o_5}
\stm^{(32|1)}{}_{i_2,i_4,x_6}^{x_2,x_4,o_6}
\stm^{(2|1)}{}_{i_3,i_5,i_6}^{x_3,x_5,x_6}
\\
&\spaceD
\times
e_1^{(i_6)}e_{21}^{(i_5)}e_{321}^{(i_4)}e_{2}^{(i_3)}e_{32}^{(i_2)}e_{3}^{(i_1)}
,
\end{split}
\label{tA lhs1 eq7}
\end{align}
where and summations are taken on $i_k,x_k{\ }(k=1,\cdots,6)$ and we set
\begin{align}
\phaseI
=p(\alpha_1)p(\alpha_3)
,\quad
\phaseII
=p(\alpha_1+\alpha_2)p(\alpha_2+\alpha_3)
,\quad
\phaseIII
=p(\alpha_2)p(\alpha_1+\alpha_2+\alpha_3)
.
\label{te phase}
\end{align}
We have put the underlines to the parts to be rewritten.
The details of the above procedure are as follows.
For (\ref{tA lhs1 eq1}), we used (\ref{tA trans mat r3 eq2}).
For (\ref{tA lhs1 eq2}), we used (\ref{tA trans mat r3 eq5}).
For (\ref{tA lhs1 eq3}), we used (\ref{tA qroot rel1}), (\ref{tA qroot rel2}) and $[e_1,e_3]=0$.
For (\ref{tA lhs1 eq4}), we used (\ref{tA trans mat r3 eq4}).
For (\ref{tA lhs1 eq5}), we used (\ref{tA trans mat r3 eq1}).
For (\ref{tA lhs1 eq6}), we used (\ref{tA qroot rel3 2}).
\par
Similarly, the second way is given by
\begin{align}
&e_3^{(o_1)}e_{23}^{(o_2)}\underline{e_{2}^{(o_3)}e_{123}^{(o_4)}}e_{12}^{(o_5)}e_{1}^{(o_6)}
\label{tA rhs1 eq1}
\\
&=(-1)^{\phaseIII o_3o_4}
e_3^{(o_1)}e_{23}^{(o_2)}e_{123}^{(o_4)}\underline{e_{2}^{(o_3)}e_{12}^{(o_5)}e_{1}^{(o_6)}}
\label{tA rhs1 eq2}
\\
&=
\sum
(-1)^{\phaseIII o_3o_4}
\stm^{(2|1)}{}_{x_3,x_5,x_6}^{o_3,o_5,o_6}
e_3^{(o_1)}\underline{e_{23}^{(o_2)}e_{123}^{(o_4)}e_{1}^{(x_6)}}e_{21}^{(x_5)}e_{2}^{(x_3)}
\label{tA rhs1 eq3}
\\
&=
\sum
(-1)^{\phaseIII o_3o_4}
\stm^{(2|1)}{}_{x_3,x_5,x_6}^{o_3,o_5,o_6}
\stm^{(23|1)}{}_{x_2,x_4,i_6}^{o_2,o_4,x_6}
\underline{e_3^{(o_1)}e_{1}^{(i_6)}}
{\ }
\underline{e_{(23)1}^{(x_4)}}
{\ }
\underline{e_{23}^{(x_2)}e_{21}^{(x_5)}}e_{2}^{(x_3)}
\label{tA rhs1 eq4}
\\
\begin{split}
&=
\sum
(-1)^{\phaseI (o_1i_6+x_4)+\phaseII x_2x_5+\phaseIII o_3o_4}
\stm^{(2|1)}{}_{x_3,x_5,x_6}^{o_3,o_5,o_6}
\stm^{(23|1)}{}_{x_2,x_4,i_6}^{o_2,o_4,x_6}
\\
&\spaceD
\times
e_{1}^{(i_6)}\underline{e_3^{(o_1)}e_{(21)3}^{(x_4)}e_{21}^{(x_5)}}e_{23}^{(x_2)}e_{2}^{(x_3)}
\end{split}
\label{tA rhs1 eq5}
\\
\begin{split}
&=
\sum
(-1)^{\phaseI (o_1i_6+x_4)+\phaseII x_2x_5+\phaseIII o_3o_4}
\stm^{(2|1)}{}_{x_3,x_5,x_6}^{o_3,o_5,o_6}
\stm^{(23|1)}{}_{x_2,x_4,i_6}^{o_2,o_4,x_6}
\stm^{(3|21)}{}_{x_1,i_4,i_5}^{o_1,x_4,x_5}
\\
&\spaceD
\times
e_{1}^{(i_6)}e_{21}^{(i_5)}e_{321}^{(i_4)}\underline{e_3^{(x_1)}e_{23}^{(x_2)}e_{2}^{(x_3)}}
\end{split}
\label{tA rhs1 eq6}
\\
\begin{split}
&=
\sum
(-1)^{\phaseI (o_1i_6+x_4)+\phaseII x_2x_5+\phaseIII o_3o_4}
\stm^{(2|1)}{}_{x_3,x_5,x_6}^{o_3,o_5,o_6}
\stm^{(23|1)}{}_{x_2,x_4,i_6}^{o_2,o_4,x_6}
\stm^{(3|21)}{}_{x_1,i_4,i_5}^{o_1,x_4,x_5}
\stm^{(3|2)}{}_{i_1,i_2,i_3}^{x_1,x_2,x_3}
\\
&\spaceD
\times
e_{1}^{(i_6)}e_{21}^{(i_5)}e_{321}^{(i_4)}e_{2}^{(i_3)}e_{23}^{(i_2)}e_3^{(i_1)}
,
\end{split}
\label{tA rhs1 eq7}
\end{align}
where summations are taken on $i_k,x_k{\ }(k=1,\cdots,6)$.
Again, we have put the underlines to the parts to be rewritten.
The details of the above procedure are as follows.
For (\ref{tA rhs1 eq1}), we used (\ref{tA qroot rel3}).
For (\ref{tA rhs1 eq2}), we used (\ref{tA trans mat r3 eq1}).
For (\ref{tA rhs1 eq3}), we used (\ref{tA trans mat r3 eq3}).
For (\ref{tA rhs1 eq4}), we used (\ref{tA qroot rel1 2}), (\ref{tA qroot rel2 2}) and $[e_1,e_3]=0$.
For (\ref{tA rhs1 eq5}), we used (\ref{tA trans mat r3 eq6}).
For (\ref{tA rhs1 eq6}), we used (\ref{tA trans mat r3 eq2}).
\par
Now, $\{e_{1}^{(i_6)}e_{21}^{(i_5)}e_{321}^{(i_4)}e_{2}^{(i_3)}e_{23}^{(i_2)}e_3^{(i_1)}\}$ are linearly independent by Theorem \ref{Yam PBW thm}.
Then, by comparing (\ref{tA lhs1 eq7}) and (\ref{tA rhs1 eq7}), we obtain the following result:
\begin{theorem}\label{te general thm}
As the identity of transition matrices of quantum superalgebras of type A, we have
\begin{align}
\begin{split}
&\sum
(-1)^{\phaseI (i_1o_6+x_4)+\phaseII x_2x_5+\phaseIII i_3i_4}
\stm^{(3|2)}{}_{x_1,x_2,x_3}^{o_1,o_2,o_3}
\stm^{(3|12)}{}_{i_1,x_4,x_5}^{x_1,o_4,o_5}
\stm^{(32|1)}{}_{i_2,i_4,x_6}^{x_2,x_4,o_6}
\stm^{(2|1)}{}_{i_3,i_5,i_6}^{x_3,x_5,x_6}
\\
&=
\sum
(-1)^{\phaseI (o_1i_6+x_4)+\phaseII x_2x_5+\phaseIII o_3o_4}
\stm^{(2|1)}{}_{x_3,x_5,x_6}^{o_3,o_5,o_6}
\stm^{(23|1)}{}_{x_2,x_4,i_6}^{o_2,o_4,x_6}
\stm^{(3|21)}{}_{x_1,i_4,i_5}^{o_1,x_4,x_5}
\stm^{(3|2)}{}_{i_1,i_2,i_3}^{x_1,x_2,x_3}
.
\end{split}
\label{te general}
\end{align}
where summations are taken on $x_k{\ }(k=1,\cdots,6)$.
\end{theorem}
The above equation (\ref{te general}) generally involve nonlocal sign factors.
In (\ref{tA our dynkin r3}), we have $\phaseI=\phaseII=\phaseIII=0$ for (I), (II) and (III).
In that case, (\ref{te general}) exactly gives the tetrahedron equation.
Hereafter, we specialize Theorem \ref{te general thm} for each case given in (\ref{tA our dynkin r3}).
\subsubsection{The case (I) $\vcenter{\hbox{\protect\includegraphics{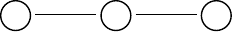}}}$}
In this case, the corresponding symmetrized Cartan matrix is given by
\begin{align}
DA
=
\begin{pmatrix}
2 & -1 & 0 \\
-1 & 2 & -1 \\
0 & -1 & 2
\end{pmatrix}
,
\end{align}
and the corresponding positive roots are given by
\begin{align}
\prer&=\{\alpha_1,\alpha_2,\alpha_3,\alpha_1+\alpha_2,\alpha_2+\alpha_3,\alpha_1+\alpha_2+\alpha_3\}
,
\\
\prir&=\{\}
.
\end{align}
\par
Now, $\stm^{(x)}$ defined by (\ref{tA trans mat r3 eq1}) $\sim$ (\ref{tA trans mat r3 eq6}) are specified as follows:
\begin{lemma}\label{tA trans mat r3 c1 lemma}
For the quantum superalgebra associated with \ddthreeA{\cc}{\cc}{\cc}, we have (\ref{tA trans mat r3 eq1}) $\sim$ (\ref{tA trans mat r3 eq6}) where $\stm^{(x)}$ are given by
\begin{align}
\stm^{(2|1)}
=\stm^{(3|2)}
=\stm^{(23|1)}
=\stm^{(32|1)}
=\stm^{(3|12)}
=\stm^{(3|21)}
=\tdr
.
\label{tA trans mat r3 c1}
\end{align}
\end{lemma}
\begin{proof}
$\stm^{(2|1)}$, $\stm^{(3|2)}$ are direct consquences of Theorem \ref{KOY13 thm 3dR}.
$\stm^{(23|1)}$ is obtained by (\ref{tA qroot serre 1}), (\ref{tA qroot serre 2}) and Theorem \ref{KOY13 thm 3dR}.
Actually, $e_1,e_{23}$ satisfy the exactly same relations of $e_1,e_2$ of $\UtAp[3]$ associated with \ddtwoA{\cc}{\cc}, so $h:\UtAp[3]\to\UtAp[4]$ defined by $e_{1}\mapsto e_{1}$, $e_{2}\mapsto e_{23}$ gives an algebra homomorphism.
Also, $d_{\alpha_2+\alpha_3}=d_{\alpha_2}$ and $d_{\alpha_1+\alpha_2+\alpha_3}=d_{\alpha_1+\alpha_2}$ are satisfied where the left hand sides are for $\UtAp[4]$ and the right hand sides are for $\UtAp[3]$, so $[m]_{q^{d_{\alpha_2+\alpha_3}}}!=[m]_{q^{d_{\alpha_2}}}!$ and $[m]_{q^{d_{\alpha_1+\alpha_2+\alpha_3}}}!=[m]_{q^{d_{\alpha_1+\alpha_2}}}!$ hold.
Therefore, by applying $h$ on (\ref{tA trans mat def1}) for the case \ddtwoA{\cc}{\cc}, we obtain
\begin{align}
e_{23}^{(a)}e_{123}^{(b)}e_{1}^{(c)}
&=\sum_{i,j,k}
\tdr_{i,j,k}^{a,b,c}
e_{1}^{(k)}e_{(23)1}^{(j)}e_{23}^{(i)}
\quad
(i,j,k,a,b,c\in\mathbb{Z}_{\geq 0})
.
\label{tA trans mat r3 c1 lemma proof1}
\end{align}
This is exactly (\ref{tA trans mat r3 eq3}) for $\stm^{(23|1)}=\tdr$.
The remaining cases can be shown exacly in the same way.
\end{proof}
The phase factors given by (\ref{te phase}) are now $\phaseI=\phaseII=\phaseIII=0$.
Then, (\ref{te general}) is specialized as follows:
\begin{align}
&\sum
\tdr_{x_1,x_2,x_3}^{o_1,o_2,o_3}\tdr_{i_1,x_4,x_5}^{x_1,o_4,o_5}\tdr_{i_2,i_4,x_6}^{x_2,x_4,o_6}\tdr_{i_3,i_5,i_6}^{x_3,x_5,x_6}
=
\sum
\tdr_{x_3,x_5,x_6}^{o_3,o_5,o_6}\tdr_{x_2,x_4,i_6}^{o_2,o_4,x_6}\tdr_{x_1,i_4,i_5}^{o_1,x_4,x_5}\tdr_{i_1,i_2,i_3}^{x_1,x_2,x_3}
,
\end{align}
where all indices are defined on $\mathbb{Z}_{\geq 0}$.
This is exactly the tetrahedron equation (\ref{KV94 te}):
\begin{align}
\tdr_{123}\tdr_{145}\tdr_{246}\tdr_{356}=\tdr_{356}\tdr_{246}\tdr_{145}\tdr_{123}
.
\end{align}
We then get the following result:
\begin{corollary}\label{tA r3 mythm1}
The tetrahedron equation (\ref{KV94 te}) is characterized as the identity of the transition matrices of the quantum superalgebra associated with \ddthreeA{\cc}{\cc}{\cc}.
\end{corollary}
We note that although Corollary \ref{tA r3 mythm1} is a corollary of the Kuniba-Okado-Yamada theorem\cite{KOY13}, the above calculation gives a direct derivation of the tetrahedron equation (\ref{KV94 te}) without using any results for quantum coordinate rings.
This is a key for the generalization of earlier results to super cases.
\subsubsection{The case (II) $\vcenter{\hbox{\protect\includegraphics{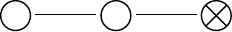}}}$}
In this case, the corresponding symmetrized Cartan matrix is given by
\begin{align}
DA
=
\begin{pmatrix}
2 & -1 & 0 \\
-1 & 2 & -1 \\
0 & -1 & 0
\end{pmatrix}
,
\end{align}
and the corresponding positive roots are given by
\begin{align}
\prer&=\{\alpha_1,\alpha_2,\alpha_1+\alpha_2\}
,
\\
\prir&=\{\alpha_3,\alpha_2+\alpha_3,\alpha_1+\alpha_2+\alpha_3\}
.
\end{align}
Similarly to Lemma \ref{tA trans mat r3 c1 lemma}, by using Proposition \ref{tA qroot rels} and \ref{tA qroot rels2}, we can show the following lemma:
\begin{lemma}\label{tA trans mat r3 c2 lemma}
For the quantum superalgebra associated with \ddthreeA{\cc}{\cc}{\cct}, we have (\ref{tA trans mat r3 eq1}) $\sim$ (\ref{tA trans mat r3 eq6}) where $\stm^{(x)}$ are given by
\begin{align}
\stm^{(2|1)}
=\tdl
,\quad
\stm^{(3|2)}
=\stm^{(23|1)}
=\stm^{(32|1)}
=\stm^{(3|12)}
=\stm^{(3|21)}
=\tdr
.
\label{tA trans mat r3 c2}
\end{align}
\end{lemma}
The phase factors given by (\ref{te phase}) are now $\phaseI=\phaseII=\phaseIII=0$.
Then, (\ref{te general}) is specialized as follows:
\begin{align}
&\sum
\tdl_{x_1,x_2,x_3}^{o_1,o_2,o_3}\tdl_{i_1,x_4,x_5}^{x_1,o_4,o_5}\tdl_{i_2,i_4,x_6}^{x_2,x_4,o_6}\tdr_{i_3,i_5,i_6}^{x_3,x_5,x_6}
=
\sum
\tdr_{x_3,x_5,x_6}^{o_3,o_5,o_6}\tdl_{x_2,x_4,i_6}^{o_2,o_4,x_6}\tdl_{x_1,i_4,i_5}^{o_1,x_4,x_5}\tdl_{i_1,i_2,i_3}^{x_1,x_2,x_3}
,
\end{align}
where $o_k,i_k,x_k\in\{0,1\}{\ }(k=1,2,4)$ and the other indices are defined on $\mathbb{Z}_{\geq 0}$.
This is exactly the tetrahedron equation (\ref{BS06 te}):
\begin{align}
\tdl_{123}\tdl_{145}\tdl_{246}\tdr_{356}=\tdr_{356}\tdl_{246}\tdl_{145}\tdl_{123}
.
\end{align}
We then get the following result:
\begin{corollary}\label{tA r3 mythm2}
The tetrahedron equation (\ref{BS06 te}) is characterized as the identity of the transition matrices of the quantum superalgebra associated with \ddthreeA{\cc}{\cc}{\cct}.
\end{corollary}
\subsubsection{The case (III) $\vcenter{\hbox{\protect\includegraphics{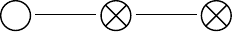}}}$}
In this case, the corresponding symmetrized Cartan matrix is given by
\begin{align}
DA
=
\begin{pmatrix}
2 & -1 & 0 \\
-1 & 0 & 1 \\
0 & 1 & 0
\end{pmatrix}
,
\end{align}
and the corresponding positive roots are given by
\begin{align}
\prer&=\{\alpha_1,\alpha_2+\alpha_3,\alpha_1+\alpha_2+\alpha_3\}
,
\\
\prir&=\{\alpha_2,\alpha_3,\alpha_1+\alpha_2\}
.
\end{align}
Similarly to Lemma \ref{tA trans mat r3 c1 lemma}, by using Proposition \ref{tA qroot rels} and \ref{tA qroot rels2}, we can show the following lemma:
\begin{lemma}\label{tA trans mat r3 c3 lemma}
For the quantum superalgebra associated with \ddthreeA{\cc}{\cct}{\cct}, we have (\ref{tA trans mat r3 eq1}) $\sim$ (\ref{tA trans mat r3 eq6}) where $\stm^{(x)}$ are given by
\begin{align}
\stm^{(2|1)}
=\tdl
,\quad
\stm^{(3|2)}
=\stm^{(3|12)}
=\stm^{(3|21)}
=\tdn(q^{-1})
,\quad
=\stm^{(23|1)}
=\stm^{(32|1)}
=\tdr
.
\label{tA trans mat r3 c3}
\end{align}
\end{lemma}
The phase factors given by (\ref{te phase}) are now $\phaseI=\phaseII=\phaseIII=0$.
Then, (\ref{te general}) is specialized as follows:
\begin{align}
&\sum
\tdn(q^{-1})_{x_1,x_2,x_3}^{o_1,o_2,o_3}\tdn(q^{-1})_{i_1,x_4,x_5}^{x_1,o_4,o_5}\tdr_{i_2,i_4,x_6}^{x_2,x_4,o_6}\tdl_{i_3,i_5,i_6}^{x_3,x_5,x_6}
=
\sum
\tdl_{x_3,x_5,x_6}^{o_3,o_5,o_6}\tdr_{x_2,x_4,i_6}^{o_2,o_4,x_6}\tdn(q^{-1})_{x_1,i_4,i_5}^{o_1,x_4,x_5}\tdn(q^{-1})_{i_1,i_2,i_3}^{x_1,x_2,x_3}
,
\end{align}
where $o_k,i_k,x_k\in\{0,1\}{\ }(k=1,3,5)$ and the other indices are defined on $\mathbb{Z}_{\geq 0}$.
We then get the following result, which gives a new solution to the tetrahedron equation.
\begin{corollary}\label{tA r3 mythm3}
As the identity of the transition matrices of the quantum superalgebra associated with \ddthreeA{\cc}{\cct}{\cct}, we have the tetrahedron equation given by
\begin{align}
\tdn(q^{-1})_{123}\tdn(q^{-1})_{145}\tdr_{246}\tdl_{356}=\tdl_{356}\tdr_{246}\tdn(q^{-1})_{145}\tdn(q^{-1})_{123}
.
\label{my te c3}
\end{align}
\end{corollary}
\subsubsection{The case (IV) $\vcenter{\hbox{\protect\includegraphics{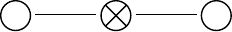}}}$}
In this case, the corresponding symmetrized Cartan matrix is given by
\begin{align}
DA
=
\begin{pmatrix}
2 & -1 & 0 \\
-1 & 0 & 1 \\
0 & 1 & -2
\end{pmatrix}
,
\end{align}
and the corresponding positive roots are given by
\begin{align}
\prer&=\{\alpha_1,\alpha_3\}
,
\\
\prir&=\{\alpha_2,\alpha_1+\alpha_2,\alpha_2+\alpha_3,\alpha_1+\alpha_2+\alpha_3\}
.
\end{align}
Similarly to Lemma \ref{tA trans mat r3 c1 lemma}, by using Proposition \ref{tA qroot rels} and \ref{tA qroot rels2}, we can show the following lemma:
\begin{lemma}\label{tA trans mat r3 c4 lemma}
For the quantum superalgebra associated with \ddthreeA{\cc}{\cct}{\cc}, we have (\ref{tA trans mat r3 eq1}) $\sim$ (\ref{tA trans mat r3 eq6}) where $\stm^{(x)}$ are given by
\begin{align}
\stm^{(2|1)}
=\stm^{(23|1)}
=\stm^{(32|1)}
=\tdl
,\quad
\stm^{(3|2)}
=\stm^{(3|12)}
=\stm^{(3|21)}
=\tdm(q^{-1})
.
\label{tA trans mat r3 c4}
\end{align}
\end{lemma}
The phase factors given by (\ref{te phase}) are now $\phaseI=0$, $\phaseII=\phaseIII=1$.
Then, (\ref{te general}) is specialized as follows:
\begin{align}
\begin{split}
&\sum
(-1)^{x_2x_5+i_3i_4}\tdm(q^{-1})_{x_1,x_2,x_3}^{o_1,o_2,o_3}\tdm(q^{-1})_{i_1,x_4,x_5}^{x_1,o_4,o_5}\tdl_{i_2,i_4,x_6}^{x_2,x_4,o_6}\tdl_{i_3,i_5,i_6}^{x_3,x_5,x_6}
\\
&=
\sum
(-1)^{x_2x_5+o_3o_4}\tdl_{x_3,x_5,x_6}^{o_3,o_5,o_6}\tdl_{x_2,x_4,i_6}^{o_2,o_4,x_6}\tdm(q^{-1})_{x_1,i_4,i_5}^{o_1,x_4,x_5}\tdm(q^{-1})_{i_1,i_2,i_3}^{x_1,x_2,x_3}
,
\end{split}
\label{my te c4}
\end{align}
where $o_k,i_k,x_k\in\{0,1\}{\ }(k=2,3,4,5)$ and the other indices are defined on $\mathbb{Z}_{\geq 0}$.
As we explained in Remark \ref{LLMM remark}, this equation resembles the tetrahedron equation (\ref{LLMM te}), but we can not eliminate the sign factors at present.
Anyway, we then get the following result:
\begin{corollary}\label{tA r3 mythm4}
As the identity of the transition matrices of the quantum superalgebra associated with \ddthreeA{\cc}{\cct}{\cc}, we have the tetrahedron equation up to sign factors given by (\ref{my te c4}).
\end{corollary}
\subsubsection{The case (V) $\vcenter{\hbox{\protect\includegraphics{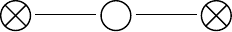}}}$}
In this case, the corresponding symmetrized Cartan matrix is given by
\begin{align}
DA
=
\begin{pmatrix}
0 & -1 & 0 \\
-1 & 2 & -1 \\
0 & -1 & 0
\end{pmatrix}
,
\end{align}
and the corresponding positive roots are given by
\begin{align}
\prer&=\{\alpha_2,\alpha_1+\alpha_2+\alpha_3\}
,
\\
\prir&=\{\alpha_1,\alpha_3,\alpha_1+\alpha_2,\alpha_2+\alpha_3\}
.
\end{align}
Similarly to Lemma \ref{tA trans mat r3 c1 lemma}, by using Proposition \ref{tA qroot rels} and \ref{tA qroot rels2}, we can show the following lemma:
\begin{lemma}\label{tA trans mat r3 c5 lemma}
For the quantum superalgebra associated with \ddthreeA{\cct}{\cc}{\cct}, we have (\ref{tA trans mat r3 eq1}) $\sim$ (\ref{tA trans mat r3 eq6}) where $\stm^{(x)}$ are given by
\begin{align}
\stm^{(2|1)}
=\tdm
,\quad
\stm^{(3|2)}
=\tdl
,\quad
\stm^{(23|1)}
=\stm^{(32|1)}
=\stm^{(3|12)}
=\stm^{(3|21)}
=\tdn
.
\label{tA trans mat r3 c5}
\end{align}
\end{lemma}
The phase factors given by (\ref{te phase}) are now $\phaseI=\phaseII=1$, $\phaseIII=0$.
Then, (\ref{te general}) is specialized as follows:
\begin{align}
\begin{split}
&\sum
(-1)^{i_1o_6+x_4+x_2x_5}\tdl_{x_1,x_2,x_3}^{o_1,o_2,o_3}\tdn_{i_1,x_4,x_5}^{x_1,o_4,o_5}\tdn_{i_2,i_4,x_6}^{x_2,x_4,o_6}\tdm_{i_3,i_5,i_6}^{x_3,x_5,x_6}
\\
&=
\sum
(-1)^{o_1i_6+x_4+x_2x_5}\tdm_{x_3,x_5,x_6}^{o_3,o_5,o_6}\tdn_{x_2,x_4,i_6}^{o_2,o_4,x_6}\tdn_{x_1,i_4,i_5}^{o_1,x_4,x_5}\tdl_{i_1,i_2,i_3}^{x_1,x_2,x_3}
,
\end{split}
\label{my te c5}
\end{align}
where $o_k,i_k,x_k\in\{0,1\}{\ }(k=1,2,5,6)$ and the other indices are defined on $\mathbb{Z}_{\geq 0}$.
We then get the following result:
\begin{corollary}\label{tA r3 mythm5}
As the identity of the transition matrices of the quantum superalgebra associated with \ddthreeA{\cct}{\cc}{\cct}, we have the tetrahedron equation up to sign factors given by (\ref{my te c5}).
\end{corollary}
\subsubsection{The case (VI) $\vcenter{\hbox{\protect\includegraphics{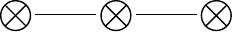}}}$}
In this case, the corresponding symmetrized Cartan matrix is given by
\begin{align}
DA
=
\begin{pmatrix}
0 & 1 & 0 \\
1 & 0 & -1 \\
0 & -1 & 0
\end{pmatrix}
,
\end{align}
and the corresponding positive roots are given by
\begin{align}
\prer&=\{\alpha_2,\alpha_3\}
,
\\
\prir&=\{\alpha_1,\alpha_2,\alpha_3,\alpha_1+\alpha_2+\alpha_3\}
.
\end{align}
Similarly to Lemma \ref{tA trans mat r3 c1 lemma}, by using Proposition \ref{tA qroot rels} and \ref{tA qroot rels2}, we can show the following lemma:
\begin{lemma}\label{tA trans mat r3 c6 lemma}
For the quantum superalgebra associated with \ddthreeA{\cct}{\cct}{\cct}, we have (\ref{tA trans mat r3 eq1}) $\sim$ (\ref{tA trans mat r3 eq6}) where $\stm^{(x)}$ are given by
\begin{align}
\stm^{(2|1)}
=\tdn(q^{-1})
,\quad
\stm^{(3|2)}
=\tdn
,\quad
\stm^{(23|1)}
=\stm^{(32|1)}
=\tdm(q^{-1})
,\quad
\stm^{(3|12)}
=\stm^{(3|21)}
=\tdl
.
\label{tA trans mat r3 c6}
\end{align}
\end{lemma}
The phase factors given by (\ref{te phase}) are now $\phaseI=\phaseIII=1$, $\phaseII=0$.
Then, (\ref{te general}) is specialized as follows:
\begin{align}
\begin{split}
&\sum
(-1)^{i_1o_6+x_4+i_3i_4}\tdn_{x_1,x_2,x_3}^{o_1,o_2,o_3}\tdl_{i_1,x_4,x_5}^{x_1,o_4,o_5}\tdm(q^{-1})_{i_2,i_4,x_6}^{x_2,x_4,o_6}\tdn(q^{-1})_{i_3,i_5,i_6}^{x_3,x_5,x_6}
\\
&=
\sum
(-1)^{o_1i_6+x_4+o_3o_4}\tdn(q^{-1})_{x_3,x_5,x_6}^{o_3,o_5,o_6}\tdm(q^{-1})_{x_2,x_4,i_6}^{o_2,o_4,x_6}\tdl_{x_1,i_4,i_5}^{o_1,x_4,x_5}\tdn_{i_1,i_2,i_3}^{x_1,x_2,x_3}
,
\end{split}
\label{my te c6}
\end{align}
where $o_k,i_k,x_k\in\{0,1\}{\ }(k=1,3,4,6)$ and the other indices are defined on $\mathbb{Z}_{\geq 0}$.
We then get the following result:
\begin{corollary}\label{tA r3 mythm6}
As the identity of the transition matrices of the quantum superalgebra associated with \ddthreeA{\cct}{\cct}{\cct}, we have the tetrahedron equation up to sign factors given by (\ref{my te c6}).
\end{corollary}
\section{PBW bases of type B and 3D reflection equation}\label{sec 5}
\subsection{PBW bases of type B of rank 2 and 3}\label{tB qroot rels subsec}
In this section, we focus on quantum superalgebras of type B in the case of rank 2 and 3.
Here, we introduce some notations to briefly describe the PBW bases of the nilpotent subalgebra of them, and show higher-order relations for them.
We define $e'_{ij},e'_{(ij)k},e'_{i(jk)}\in\UtBp$ in the same way as (\ref{tA notations}).
Let us recall the normalization $e'_{\beta}\mapsto e_{\beta}$ given in Definition \ref{qrv}(ii).
Corresponding to the normalization, we use a simplified rule as follows, which is enough to our description.
We set $e_{x}=e'_{x}/(q^{1/2}+q^{-1/2}){\ }(x=ij,(ij)k,i(jk))$ for the case $x$ involves the letter ``$r$'' twice where $r$ is the rank of $\UtBp$ defined in Section \ref{sec 21}, and we set $e_{x}=e'_{x}{\ }(x=ij,(ij)k,i(jk))$ otherwise.
By considering Corollary \ref{qcom jacobi coro}(1), we simply write $e_{ijk}=e_{(ij)k}$ for the case $(\alpha_i,\alpha_k)=0$.
\begin{example}
The indices of the element $e'_{(12)2}\in\UtBp$ for rank $r=m+n=2$ involve the letter ``$2$'' twice.
We then set $e_{(12)2}=e'_{(12)2}/(q^{1/2}+q^{-1/2})$.
\end{example}
\par
We also define elements with more $q$-commutators as well.
For example, we define $e'_{((jk)k)(ji)}\in\UtBp$ by
\begin{align}
e'_{((jk)k)(ji)}
=[[[e_j,e_k]_q,e_k]_q,[e_j,e_i]_q]_q
.
\end{align}
Then, we similarly set $e_{x}=e'_{x}/(q^{1/2}+q^{-1/2})$ for the case $x$ involves two letters ``$r$'', and we set $e_{x}=e'_{x}$ otherwise.
These elements satisfy the following higher-order relations, where we only consider the case of rank 3, which is enough for our purpose in Section \ref{sec 53}.
\begin{proposition}\label{tB qroot rels1}
For the case of rank 3, we have
\begin{align}
&\ltrlo=(-1)^{p(\alpha_1)p(\alpha_3)}\ltolr
\label{tB qroot rel2}
,\\
&\oot=\olltrlrl
\label{tB qroot rel3}
,\\
&\lltrlrlo=\lltolrlr
\label{tB qroot rel4}
,\\
&\loltrllltrl=(-1)^{(p(\alpha_1)+p(\alpha_2)+p(\alpha_3))p(\alpha_2)}\tlllotlrlrl
\label{tB qroot rel5}
,\\
&\ltrllltrlol=(-1)^{p(\alpha_1)p(\alpha_3)+(p(\alpha_1)+p(\alpha_2))(p(\alpha_2)+p(\alpha_3))}\ltollltrlrl
\label{tB qroot rel6}
,\\
&\lltrlrlltol=\tlrlrltolll
\label{tB qroot rel7}
.
\end{align}
\end{proposition}
\begin{proof}
(\ref{tB qroot rel2}) is obtained by Corollary \ref{qcom jacobi coro}(2) because $[e_{1},e_{3}]=0$ and $(\alpha_{1},\alpha_{3})=0$.
(\ref{tB qroot rel3}) is obtained by Corollary \ref{qcom jacobi coro}(1).
(\ref{tB qroot rel4}) is obtained by
\begin{align}
[[[e_{2},e_{3}]_{q},e_{3}]_{q},e_{1}]_{q}
=
(-1)^{p(\alpha_1)p(\alpha_3)}[[[e_{2},e_{3}]_{q},e_{1}]_{q},e_{3}]_{q}
=
[[[e_{2},e_{1}]_{q},e_{3}]_{q},e_{3}]_{q}
,
\label{tB qroot rels1 proof1}
\end{align}
where we used Corollary \ref{qcom jacobi coro}(2).
For (\ref{tB qroot rel7}), it is sufficient to show $e_{(12)(3(32))}=e_{(((12)3)3)2}$ by considering the anti-algebra automorphism $\chi$ given by (\ref{anti chi}).
Actually, we have $e_{(12)(3(32))}=e_{((12)3)(32)}=e_{(((12)3)3)2}$ where we used Corollary \ref{qcom jacobi coro}(1), (\ref{tB qroot rel9}) and Lemma \ref{qroot general lemma}(2).
(\ref{tB qroot rel5}) and (\ref{tB qroot rel6}) are shown in the same way.
Here, we only present the proof for (\ref{tB qroot rel5}).
We can calculate the left hand side of (\ref{tB qroot rel5}) as
\begin{align}
\begin{split}
[e_{123},[e_{2},e_{3}]_{q}]_{q}/(q^{1/2}+q^{-1/2})
=&
(-1)^{(p(\alpha_1)+p(\alpha_2)+p(\alpha_3))p(\alpha_2)}q^{-(\alpha_1+\alpha_2+\alpha_3,\alpha_2)}e_{2}e_{((12)3)3}
\\
&-(-1)^{p(\alpha_2)p(\alpha_3)}q^{-(\alpha_2,\alpha_3)}e_{((12)3)3}e_{2}
\end{split}
\\
\begin{split}
=&
(-1)^{(p(\alpha_1)+p(\alpha_2)+p(\alpha_3))p(\alpha_2)}
q^{-(\alpha_1+\alpha_2+\alpha_3,\alpha_2)}
\\
&\times
\left(
e_{2}e_{((12)3)3}
-(-1)^{(p(\alpha_1)+p(\alpha_2))p(\alpha_2)}q^{(\alpha_1+\alpha_2,\alpha_2)}e_{((12)3)3}e_{2}
\right)
,
\end{split}
\label{tB qroot rels1 proof2}
\end{align}
where we used Lemma \ref{qcom jacobi} and (\ref{tB qroot rel8}).
One can see $(\alpha_1+\alpha_2+\alpha_3,\alpha_2)=(\BB_1,\BB_2-\BB_3)=0$ and $(\alpha_1+\alpha_2,\alpha_2)=(\BB_3,\BB_3)=-(\alpha_1+\alpha_2+2\alpha_3,\alpha_2)$.
Then, the right hand side of (\ref{tB qroot rels1 proof2}) is exactly the right hand side of (\ref{tB qroot rel5}).
\end{proof}
Here, we cite a lemma from \cite[Lemma 6.3.1(i)]{Yam94} used below.
\begin{lemma}\label{Yama lemma1}
For $1\leq i\leq r-1$, $[e_{\BB_{i}},e_{\BB_{i}+\BB_{r}}]_{q}=0$.
Especially, if $r=3$, this gives $[e_{123},e_{((12)3)3}]_{q}=0$ for $i=1$ and $[e_{23},e_{(23)3}]_{q}=0$ for $i=2$.
\end{lemma}
\begin{proposition}\label{tB qroot rels2}
For the case of rank 3, we have
\begin{align}
&[e_{\jj},\ooo]=0
\label{tB qroot rel8}
,\\
&[\toX,\poo]=0
\label{tB qroot rel9}
,\\
&[\poo,\oot]=0
\label{tB qroot rel10}
,\\
&[\ltrlo,\pot]=0
\label{tB qroot rel11}
,\\
&[e_{\kk},\lltrlrlltol]=0
\label{tB qroot rel12}
.
\end{align}
\end{proposition}
\begin{proof}
(\ref{tB qroot rel8}) and (\ref{tB qroot rel9}) are obtained  exactly in the same way as (\ref{tA qroot rel2}) and (\ref{tA qroot rel3}), where we use the anti-algebra automorphism $\chi$ given by (\ref{anti chi}).
(\ref{tB qroot rel10}) is obtained by $[e_{((12)3)3},e_{23}]=[e_{(((12)3)3)2},e_{3}]=0$ where we used Corollary \ref{qcom jacobi coro}(1) and Lemma \ref{qroot general lemma}(2).
(\ref{tB qroot rel11}) is obtained by
\begin{align}
\begin{split}
[[e_{23},e_{1}]_{q},e_{(23)3}]_{q}
=&
[e_{23},[e_{1},e_{(23)3}]]
+(-1)^{p(\alpha_1)p(\alpha_2)}q^{-(\alpha_1,\alpha_2+\alpha_3)}
[e_{23},e_{(23)3}]e_{1}
\\
&+(-1)^{p(\alpha_1)(p(\alpha_2)+p(\alpha_3))}q^{-(\alpha_1,\alpha_2+\alpha_3)}
e_{1}[e_{23},e_{(23)3}]
,
\end{split}
\end{align}
where we used Lemma \ref{qcom jacobi}.
This is actually equal to $0$ by Lemma \ref{Yama lemma1} for $i=2$, (\ref{tB qroot rel3}) and (\ref{tB qroot rel10}).
(\ref{tB qroot rel12}) is obtained by
\begin{align}
\chi([e_{\kk},\lltrlrlltol])
=\chi([e_{\kk},\tlrlrltolll])
=[e_{(((12)3)3)2},e_{\kk}]
=0,
\end{align}
where we used (\ref{tB qroot rel7}) and Lemma \ref{qroot general lemma}(2).
\end{proof}
\begin{proposition}\label{tB qroot rels3}
For the case of rank 3, we have
\begin{align}
&e_{\ii}^2\poo-(q+q^{-1})e_{\ii}\poo e_{\ii}+\poo e_{\ii}^2=0
\quad (\alpha_{\ii}\in\prer)
\label{tB qroot serre 1}
,\\
&\poo^3e_{\ii}-(q+1+q^{-1})\poo^2 e_{\ii}\poo+(q+1+q^{-1})\poo e_{\ii}\poo^2-e_{\ii}\poo^3=0
\quad (\alpha_{\jj}+\alpha_{\kk}\in\prer)
\label{tB qroot serre 2}
,\\
&e_{\kk}^3\toX-(q+1+q^{-1})e_{\kk}^2\toX e_{\kk}+(q+1+q^{-1})e_{\kk}\toX e_{\kk}^2-\toX e_{\kk}^3=0
\quad (\alpha_{\kk}\in\prer)
\label{tB qroot serre 3}
,\\
&\toX^2e_{\kk}-(q+q^{-1})\toX e_{\kk}\toX+e_{\kk}\toX^2=0
\quad (\alpha_{\ii}+\alpha_{\jj}\in\prer)
\label{tB qroot serre 4}
,\\
&e_{\ii}^2\pot-(q+q^{-1})e_{\ii}\pot e_{\ii}+\pot e_{\ii}^2=0
\quad (\alpha_{\ii}\in\prer)
\label{tB qroot serre 5}
,\\
&\pot^2e_{\ii}-(q+q^{-1})\pot e_{\ii}\pot+e_{\ii}\pot^2=0
\quad (\alpha_{\jj}+2\alpha_{\kk}\in\prer)
\label{tB qroot serre 6}
,\\
&e_{\jj}^2\oot-(q+q^{-1})e_{\jj}\oot e_{\jj}+\oot e_{\jj}^2=0
\quad (\alpha_{\jj}\in\prer)
\label{tB qroot serre 7}
,\\
&\oot^2e_{\jj}-(q+q^{-1})\oot e_{\jj}\oot +e_{\jj}\oot^2=0
\quad (\alpha_{\ii}+\alpha_{\jj}+2\alpha_{\kk}\in\prer)
\label{tB qroot serre 8}
,\\
&\toX^2\pot-(q+q^{-1})\toX\pot\toX+\pot\toX=0
\quad (\alpha_{\ii}+\alpha_{\jj}\in\prer)
\label{tB qroot serre 9}
,\\
&\pot^2\toX-(q+q^{-1})\pot\toX\pot+\toX\pot^2=0
\quad (\alpha_{\jj}+2\alpha_{\kk}\in\prer)
\label{tB qroot serre 10}
.
\end{align}
\end{proposition}
\begin{proof}
The proof of this proposition will be presented together with the next proposition.
\end{proof}
\begin{proposition}\label{tB qroot rels4}
For the case of rank 3, we have
\begin{align}
&\toX^2=0
\quad (\alpha_{\ii}+\alpha_{\jj}\in\prir)
\label{tB qroot serre 11}
,\\
&\pot^2=0
\quad (\alpha_{\jj}+2\alpha_{\kk}\in\prir)
\label{tB qroot serre 12}
,\\
&\oot^2=0
\quad (\alpha_{\ii}+\alpha_{\jj}+2\alpha_{\kk}\in\prir)
\label{tB qroot serre 13}
,\\
\begin{split}
&\poo^3e_{\ii}+(-1)^{p(\alpha_{\ii})}(1-q-q^{-1})\poo^2 e_{\ii}\poo
\\
&+(1-q-q^{-1})\poo e_{\ii}\poo^2+(-1)^{p(\alpha_{\ii})}e_{\ii}\poo^3=0
\quad (\alpha_{\jj}+\alpha_{\kk}\in\prar)
,
\end{split}
\label{tB qroot serre 14}
\\
\begin{split}
&e_{\kk}^3\toX+(-1)^{p(\alpha_{\ii})+p(\alpha_{\jj})}(1-q-q^{-1})e_{\kk}^2\toX e_{\kk}
\\
&+(1-q-q^{-1})e_{\kk}\toX e_{\kk}^2+(-1)^{p(\alpha_{\ii})+p(\alpha_{\jj})}\toX e_{\kk}^3=0
\quad (\alpha_{\kk}\in\prar)
.
\end{split}
\label{tB qroot serre 15}
\end{align}
\end{proposition}
\begin{proof}
(\ref{tB qroot serre 11}), (\ref{tB qroot serre 12}) and (\ref{tB qroot serre 13}) are cororallies of Lemma \ref{qroot general lemma}(1).
(\ref{tB qroot serre 3}) and (\ref{tB qroot serre 15}) can be written together as $[[[e_{21},e_{3}]_{q},e_{3}]_{q},e_{3}]_{q}$.
Then, (\ref{tB qroot serre 1}), (\ref{tB qroot serre 3}), (\ref{tB qroot serre 15}) and (\ref{tB qroot serre 5}) are obtained exactly in the same way as (\ref{tA qroot serre 1}).
(\ref{tB qroot serre 4}) and (\ref{tB qroot serre 6}) are obtained in the same way as (\ref{tA qroot serre 2}), where we use (\ref{tB qroot rel11}) for (\ref{tB qroot serre 6}).
The left hand side of (\ref{tB qroot serre 7}) can be written as $[[e_{((12)3)3},e_{2}]_{q},e_{2}]_{q}=[e_{(((12)3)3)2},e_{2}]_{q}$.
This is equal to $0$ by Lemma \ref{qroot general lemma}.
\par
(\ref{tB qroot serre 2}) and (\ref{tB qroot serre 14}) can be written together as $[[[e_{1},e_{23}]_{q},e_{23}]_{q},e_{23}]_{q}$.
Then, we have
\begin{align}
[[[e_{1},e_{23}]_{q},e_{23}]_{q},e_{23}]_{q}/(q^{1/2}+q^{-1/2})
=&
[e_{(1(23))(23)},e_{23}]_{q}
\\
=&
(-1)^{(p(\alpha_1)+p(\alpha_2)+p(\alpha_3))p(\alpha_2)}
[e_{2(((12)3)3)},e_{23}]_{q}
\\
=&
(-1)^{(p(\alpha_1)+p(\alpha_2)+p(\alpha_3))p(\alpha_2)}
[e_{2},[e_{((12)3)3},e_{23}]_{q}]_{q}
=
0
,
\end{align}
where we used (\ref{tB qroot rel5}), Corollary \ref{qcom jacobi coro}(1), Corollary \ref{generalized serre} and Lemma \ref{qroot general lemma}(2).
\par
The left hand side of (\ref{tB qroot serre 8}) can be written as $[e_{2(((12)3)3)},e_{((12)3)3}]_{q}$.
Then we have
\begin{align}
[e_{2(((12)3)3)},e_{((12)3)3}]_{q}
=&
(-1)^{(p(\alpha_1)+p(\alpha_2)+p(\alpha_3))p(\alpha_2)}
[e_{(1(23))(23)},e_{((12)3)3}]_{q}
\\
=&
(-1)^{(p(\alpha_1)+p(\alpha_2)+p(\alpha_3))p(\alpha_2)}
[e_{123},[e_{23},e_{((12)3)3}]_{q}]_{q}/(q^{1/2}+q^{-1/2})
\\
=&
(-1)^{(p(\alpha_1)+p(\alpha_2)+p(\alpha_3))p(\alpha_2)}
[e_{123},[e_{23},e_{1((23)3)}]_{q}]_{q}/(q^{1/2}+q^{-1/2})
\\
=&
(-1)^{(p(\alpha_1)+p(\alpha_2)+p(\alpha_3))p(\alpha_2)}
[e_{123},[e_{(23)1},e_{(23)3}]_{q}]_{q}/(q^{1/2}+q^{-1/2})
=0
,
\end{align}
where we used (\ref{tB qroot rel5}), Corollary \ref{qcom jacobi coro}(1), Lemma \ref{Yama lemma1} for $i=1$, (\ref{tB qroot rel3}), Lemma \ref{Yama lemma1} for $i=2$ and (\ref{tB qroot rel11}).
\par
The left hand side of (\ref{tB qroot serre 9}) can be written as $[e_{21},e_{(21)((23)3)}]_{q}$.
Then we have
\begin{align}
[e_{21},e_{(21)((23)3)}]_{q}
=&
(-1)^{(p(\alpha_1)+p(\alpha_2))(p(\alpha_2)+p(\alpha_3))}
[e_{21},e_{(23)((23)1)}]_{q}]_{q}
\\
\begin{split}
=&
1/(q^{1/2}+q^{-1/2})\left(
(-1)^{(p(\alpha_1)+p(\alpha_2))(p(\alpha_2)+p(\alpha_3))}
(
[[e_{21},e_{23}]_{q},e_{(23)1}]_{q}
\right.
\\
&-(-1)^{(p(\alpha_2)+p(\alpha_3))(p(\alpha_1)+p(\alpha_2)+p(\alpha_3))}q^{-(\alpha_2+\alpha_3,\alpha_1+\alpha_2+\alpha_3)}[e_{21},e_{(23)1}]_{q}e_{23}
\\
&+
\left.
(-1)^{(p(\alpha_1)+p(\alpha_2))(p(\alpha_2)+p(\alpha_3))}
q^{-(\alpha_1+\alpha_2,\alpha_2+\alpha_3)}
e_{23}[e_{21},e_{(23)1}]_{q}
)
\right)
,
\end{split}
\end{align}
where we used (\ref{tB qroot rel6}) and Lemma \ref{qcom jacobi}.
This is actually equal to $0$ because we have $[e_{21},e_{23}]_{q}=0$ by (\ref{tB qroot rel9}) and $[e_{21},e_{(23)1}]_{q}=[[e_{21},e_{23}]_{q},e_{1}]_{q}=0$ by Corollary \ref{qcom jacobi coro}(1), Corollary \ref{generalized serre} and (\ref{tB qroot rel9}).
\par
We prove (\ref{tB qroot serre 10 chi}) instead of (\ref{tB qroot serre 10}).
The left hand side of (\ref{tB qroot serre 10 chi}) can be written as $[[e_{12},e_{3(32)}]_{q},e_{3(32)}]_{q}$.
Here, we have $[e_{12},e_{3(32)}]_{q}=[e_{123},e_{32}]_{q}/(q^{1/2}+q^{-1/2})=e_{(((12)3)3)2}$ where we used Corollary \ref{qcom jacobi coro}(1), (\ref{tB qroot rel9 chi}) and (\ref{tB qroot rel8}).
Then, it is sufficient to show $[e_{(((12)3)3)2},e_{3(32)}]_{q}=0$.
This actually holds by repeated use of Lemma \ref{qcom jacobi} and Lemma \ref{qroot general lemma}(2).
\end{proof}
Also, by applying the anti-algebra automorphism $\chi$ given by (\ref{anti chi}) on the above propositions, we obtain the following relations:
\begin{proposition}\label{tB qroot rels1 chi}
For the case of rank 3, we have
\begin{align}
&e_{\ii(\kk\jj)}=(-1)^{p(\alpha_1)p(\alpha_3)}e_{\kk(\ii\jj)}
\label{tB qroot rel2 chi}
,\\
&e_{\kk(\kk(\jj\ii))}=e_{(\kk(\kk\jj))\ii}
\label{tB qroot rel3 chi}
,\\
&e_{\ii(\kk(\kk\jj))}=e_{\kk(\kk(\ii\jj))}
\label{tB qroot rel4 chi}
,\\
&e_{(\kk\jj)((\kk\jj)\ii)}=(-1)^{(p(\alpha_1)+p(\alpha_2)+p(\alpha_3))p(\alpha_2)}e_{(\kk(\kk(\jj\ii)))\jj}
\label{tB qroot rel5 chi}
,\\
&e_{(\ii(\kk\jj))(\kk\jj)}=(-1)^{p(\alpha_1)p(\alpha_3)+(p(\alpha_1)+p(\alpha_2))(p(\alpha_2)+p(\alpha_3))}e_{(\kk(\kk\jj))(\ii\jj)}
\label{tB qroot rel6 chi}
,\\
&e_{(\ii\jj)(\kk(\kk\jj))}=e_{(((\ii\jj)\kk)\kk)\jj}
\label{tB qroot rel7 chi}
.
\end{align}
\end{proposition}
\begin{proposition}\label{tB qroot rels2 chi}
For the case of rank 3, we have
\begin{align}
&[e_{\jj},e_{\kk\jj\ii}]=0
\label{tB qroot rel8 chi}
,\\
&[e_{\ii\jj},e_{\kk\jj}]=0
\label{tB qroot rel9 chi}
,\\
&[e_{\kk\jj},e_{\kk(\kk(\jj\ii))}]=0
\label{tB qroot rel10 chi}
,\\
&[e_{\kk(\kk\jj)},e_{\ii(\kk\jj)}]=0
\label{tB qroot rel11 chi}
,\\
&[e_{\kk},e_{(\ii\jj)(\kk(\kk\jj))}]=0
\label{tB qroot rel12 chi}
.
\end{align}
\end{proposition}
\begin{proposition}\label{tB qroot rels3 chi}
For the case of rank 3, we have
\begin{align}
&e_{\ii}^2e_{\kk\jj}-(q+q^{-1})e_{\ii}e_{\kk\jj} e_{\ii}+e_{\kk\jj} e_{\ii}^2=0
\quad (\alpha_{\ii}\in\prer)
\label{tB qroot serre 1 chi}
,\\
&e_{\kk\jj}^3e_{\ii}-(q+1+q^{-1})e_{\kk\jj}^2 e_{\ii}e_{\kk\jj}+(q+1+q^{-1})e_{\kk\jj} e_{\ii}e_{\kk\jj}^2-e_{\ii}e_{\kk\jj}^3=0
\quad (\alpha_{\jj}+\alpha_{\kk}\in\prer)
\label{tB qroot serre 2 chi}
,\\
&e_{\kk}^3e_{\ii\jj}-(q+1+q^{-1})e_{\kk}^2e_{\ii\jj} e_{\kk}+(q+1+q^{-1})e_{\kk}e_{\ii\jj} e_{\kk}^2-e_{\ii\jj} e_{\kk}^3=0
\quad (\alpha_{\kk}\in\prer)
\label{tB qroot serre 3 chi}
,\\
&e_{\ii\jj}^2e_{\kk}-(q+q^{-1})e_{\ii\jj} e_{\kk}e_{\ii\jj}+e_{\kk}e_{\ii\jj}^2=0
\quad (\alpha_{\ii}+\alpha_{\jj}\in\prer)
\label{tB qroot serre 4 chi}
,\\
&e_{\ii}^2e_{\kk(\kk\jj)}-(q+q^{-1})e_{\ii}e_{\kk(\kk\jj)} e_{\ii}+e_{\kk(\kk\jj)} e_{\ii}^2=0
\quad (\alpha_{\ii}\in\prer)
\label{tB qroot serre 5 chi}
,\\
&e_{\kk(\kk\jj)}^2e_{\ii}-(q+q^{-1})e_{\kk(\kk\jj)} e_{\ii}e_{\kk(\kk\jj)}+e_{\ii}e_{\kk(\kk\jj)}^2=0
\quad (\alpha_{\jj}+2\alpha_{\kk}\in\prer)
\label{tB qroot serre 6 chi}
,\\
&e_{\jj}^2e_{\kk(\kk(\jj\ii))}-(q+q^{-1})e_{\jj}e_{\kk(\kk(\jj\ii))} e_{\jj}+e_{\kk(\kk(\jj\ii))} e_{\jj}^2=0
\quad (\alpha_{\jj}\in\prer)
\label{tB qroot serre 7 chi}
,\\
&e_{\kk(\kk(\jj\ii))}^2e_{\jj}-(q+q^{-1})e_{\kk(\kk(\jj\ii))} e_{\jj}e_{\kk(\kk(\jj\ii))} +e_{\jj}e_{\kk(\kk(\jj\ii))}^2=0
\quad (\alpha_{\ii}+\alpha_{\jj}+2\alpha_{\kk}\in\prer)
\label{tB qroot serre 8 chi}
,\\
&e_{\ii\jj}^2e_{\kk(\kk\jj)}-(q+q^{-1})e_{\ii\jj}e_{\kk(\kk\jj)}e_{\ii\jj}+e_{\kk(\kk\jj)}e_{\ii\jj}=0
\quad (\alpha_{\ii}+\alpha_{\jj}\in\prer)
\label{tB qroot serre 9 chi}
,\\
&e_{\kk(\kk\jj)}^2e_{\ii\jj}-(q+q^{-1})e_{\kk(\kk\jj)}e_{\ii\jj}e_{\kk(\kk\jj)}+e_{\ii\jj}e_{\kk(\kk\jj)}^2=0
\quad (\alpha_{\jj}+2\alpha_{\kk}\in\prer)
\label{tB qroot serre 10 chi}
.
\end{align}
\end{proposition}
\begin{proposition}\label{tB qroot rels4 chi}
For the case of rank 3, we have
\begin{align}
&e_{\ii\jj}^2=0
\quad (\alpha_{\ii}+\alpha_{\jj}\in\prir)
\label{tB qroot serre 11 chi}
,\\
&e_{\kk(\kk\jj)}^2=0
\quad (\alpha_{\jj}+2\alpha_{\kk}\in\prir)
\label{tB qroot serre 12 chi}
,\\
&e_{\kk(\kk(\jj\ii))}^2=0
\quad (\alpha_{\ii}+\alpha_{\jj}+2\alpha_{\kk}\in\prir)
\label{tB qroot serre 13 chi}
,\\
\begin{split}
&e_{\kk\jj}^3e_{\ii}+(-1)^{p(\alpha_{\ii})}(1-q-q^{-1})e_{\kk\jj}^2 e_{\ii}e_{\kk\jj}
\\
&+(1-q-q^{-1})e_{\kk\jj} e_{\ii}e_{\kk\jj}^2+(-1)^{p(\alpha_{\ii})}e_{\ii}e_{\kk\jj}^3=0
\quad (\alpha_{\jj}+\alpha_{\kk}\in\prar)
,
\end{split}
\label{tB qroot serre 14 chi}
\\
\begin{split}
&e_{\kk}^3e_{\ii\jj}+(-1)^{p(\alpha_{\ii})+p(\alpha_{\jj})}(1-q-q^{-1})e_{\kk}^2e_{\ii\jj} e_{\kk}
\\
&+(1-q-q^{-1})e_{\kk}e_{\ii\jj} e_{\kk}^2+(-1)^{p(\alpha_{\ii})+p(\alpha_{\jj})}e_{\ii\jj} e_{\kk}^3=0
\quad (\alpha_{\kk}\in\prar)
.
\end{split}
\label{tB qroot serre 15 chi}
\end{align}
\end{proposition}
\par
By writing down quantum root vectors given by Definition \ref{qrv} for the case of rank 2, we find they are given by
\begin{align}
&B_{1}:
\quad e_{\beta_1}=e_{1}
,\quad e_{\beta_2}=e_{21}
,\quad e_{\beta_3}=e_{2(21)}
,\quad e_{\beta_3}=e_{2}
,
\label{tB r2 qroot vec1}
\\
&B_{2}:
\quad e_{\beta_1}=e_{2}
,\quad e_{\beta_2}=e_{(12)2}
,\quad e_{\beta_3}=e_{12}
,\quad e_{\beta_3}=e_{1}
,
\label{tB r2 qroot vec2}
\end{align}
where $\beta_{t}{\ }(t=1,\cdots,4)$ are the same as Theorem \ref{Yam PBW thm}.
For non-super case, (\ref{tB r2 qroot vec1}) and (\ref{tB r2 qroot vec2}) concide with quantum root vectors given by (\ref{qroot Lus}) with the reduced expressions $w_0=s_{1}s_{2}s_{1}s_{2},s_{2}s_{1}s_{2}s_{1}$ of the longest element of the Weyl group, respectively.
\par
Similarly, by writing down quantum root vectors given by Definition \ref{qrv} for the case of rank 3, we find they are given by
\begin{align}
\begin{split}
B_{1}:
\quad &e_{\beta_1}=e_{1}
,\quad e_{\beta_2}=\toX
,\quad e_{\beta_3}=\rltol
,\quad e_{\beta_4}=\rlrltoll
,\quad e_{\beta_5}=\tlrlrltolll
,
\\
&e_{\beta_6}=e_{2}
,\quad e_{\beta_7}=\rtX
,\quad e_{\beta_8}=\rlrtl
,\quad e_{\beta_9}=e_{3}
,
\end{split}
\label{tB r3 qroot vec1}
\\
\begin{split}
B_{2}:
\quad &e_{\beta_1}=e_{3}
,\quad e_{\beta_2}=\pot
,\quad e_{\beta_3}=\poo
,\quad e_{\beta_4}=e_{2}
,\quad e_{\beta_5}=\ott
,
\\
&e_{\beta_6}=\oot
,\quad e_{\beta_7}=\ooo
,\quad e_{\beta_8}=\oop
,\quad e_{\beta_9}=e_{1}
,
\end{split}
\label{tB r3 qroot vec2}
\end{align}
where $\beta_{t}{\ }(t=1,\cdots,9)$ are the same as Theorem \ref{Yam PBW thm}.
For non-super case, (\ref{tB r3 qroot vec1}) and (\ref{tB r3 qroot vec2}) concide with quantum root vectors given by (\ref{qroot Lus}) with the reduced expressions $w_0=s_{1}s_{2}s_{3}s_{2}s_{1}s_{2}s_{3}s_{2}s_{3}$, $s_{3}s_{2}s_{3}s_{2}s_{1}s_{2}s_{3}s_{2}s_{1}$ of the longest element of the Weyl group, respectively.
\subsection{Transition matrices of PBW bases of type B of rank 2}
In this section, we consider transition matrices of the PBW bases of $\UtBp$ of rank 2, so $m+n=2$.
All possible Dynkin diagrams associated with admissible realizations are given in Table \ref{tB all dynkin}, where they are distinguished except \ddtwoB{\cct}{\ccb}, in the sense defined in Section \ref{sec 22}.
\begin{table}[htb]
\centering
\caption{}
\label{tB all dynkin}
\vspace{2mm}
\begin{tabular}{c|c}\hline
$\mathfrak{g}$ & Dynkin diagram \\\hline
$osp(5|0)$ & \ddtwoWul{Rightarrow}{\cc}{\cc}{\epsilon_1-\epsilon_2}{\epsilon_2} \\\hline
$osp(3|2)$ & \ddtwoWul{Rightarrow}{\cct}{\cc}{\delta_1-\epsilon_2}{\epsilon_2}
\quad
\ddtwoWul{Rightarrow}{\cct}{\ccb}{\epsilon_1-\delta_2}{\delta_2} \\\hline
$osp(1|4)$ & \ddtwoWul{Rightarrow}{\cc}{\ccb}{\delta_1-\delta_2}{\delta_2} \\\hline
\end{tabular}
\end{table}
For the case of rank 2, quantum root vecotrs are given by (\ref{tB r2 qroot vec1}),(\ref{tB r2 qroot vec2}), so the transition matrix (\ref{tm1}) is given as follows:
\begin{align}
e_{2}^{(a)}e_{(12)2}^{(b)}e_{12}^{(c)}e_{1}^{(d)}
&=\sum_{i,j,k,l}
\tm_{i,j,k,l}^{a,b,c,d}
e_{1}^{(l)}e_{21}^{(k)}e_{2(21)}^{(j)}e_{2}^{(i)}
,
\label{tB trans mat def1}
\end{align}
where the domain of indices is specified below.
Hereafter, we consider each case.
Sometimes, we abbreviate simple roots for Dynkin diagrams, but we always assume that they are given as Table \ref{tB all dynkin}.
\subsubsection{The case (I) $\vcenter{\hbox{\protect\includegraphics{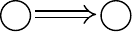}}}$}
In this case, the corresponding symmetrized Cartan matrix is given by
\begin{align}
DA
=
\begin{pmatrix}
2 & -1 \\
-1 & 1
\end{pmatrix}
,
\end{align}
and the corresponding positive roots are given by
\begin{align}
\prer&=\{\alpha_1,\alpha_2,\alpha_1+\alpha_2,\alpha_1+2\alpha_2\}
,
\\
\prir&=\{\}
,\\
\prar&=\{\}
.
\end{align}
Then, indices are specified as $i,j,k,l,a,b,c,d\in\mathbb{Z}_{\geq 0}$ for (\ref{tB trans mat def1}).
The transition matrix in (\ref{tB trans mat def1}) is explicitly given as the consequence of the Kuniba-Okado-Yamada theorem\cite{KOY13}:
\begin{theorem}[\cite{KO13,KOY13}]\label{KOY13 thm 3dJ}
For the quantum superalgebra associated with \ddtwoB{\cc}{\cc}, the transition matrix in (\ref{tB trans mat def1}) is given by
\begin{align}
\tm_{i,j,k,l}^{a,b,c,d}
=\tdj_{i,j,k,l}^{a,b,c,d}
\label{KOY13 3dJ result}
,
\end{align}
where $\tdj$ is the 3D J given by (\ref{3dJ mat el}).
\end{theorem}
\subsubsection{The case (II) $\vcenter{\hbox{\protect\includegraphics{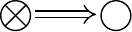}}}$}
In this case, the corresponding symmetrized Cartan matrix is given by
\begin{align}
DA
=
\begin{pmatrix}
0 & -1 \\
-1 & 1
\end{pmatrix}
,
\end{align}
and the corresponding positive roots are given by
\begin{align}
\prer&=\{\alpha_2\}
,
\\
\prir&=\{\alpha_1,\alpha_1+2\alpha_2\}
,\\
\prar&=\{\alpha_1+\alpha_2\}
.
\end{align}
Then, indices are specified as $i,k,a,c\in\mathbb{Z}_{\geq 0},{\ }j,l,b,d\in\{0,1\}$ for (\ref{tB trans mat def1}).
We set $\tdx(q)\in\mathrm{End}(F\otimes V\otimes F\otimes V)$ by
\begin{align}
&\tdx(q)(\ket{i}\otimes u_j\otimes \ket{k}\otimes u_l)
=\sum_{a,c\in\mathbb{Z}_{\geq 0},b,d\in\{0,1\}}
\tdx(q)_{i,j,k,l}^{a,b,c,d}
\ket{a}\otimes u_b\otimes \ket{c}\otimes u_d
,
\label{3dX}
\\
&\tdx(q)_{i,0,k,0}^{a,0,c,0}
=\delta_{i,a}\delta_{k,c}\left(1-(1-(-q)^{c})q^a\right)
,
\label{3dX mat el 1}
\\
&\tdx(q)_{i,1,k,0}^{a,0,c,0}
=\delta_{i,a-1}\delta_{k,c-1}(-1)^{c}q^{(a+c-1)/2}(1+q)
,
\\
&\tdx(q)_{i,0,k,1}^{a,0,c,0}
=\delta_{i,a+1}\delta_{k,c-1}(-1)^{c+1}q^{(a+c-1)/2}(1+q)(1-q^{a+1})
,\\
&\tdx(q)_{i,1,k,1}^{a,0,c,0}
=\delta_{i,a}\delta_{k,c-2}q^{a+c-1}(1+q)^2
,\\
&\tdx(q)_{i,0,k,0}^{a,1,c,0}
=\delta_{i,a+1}\delta_{k,c+1}(-1)^{c+1}q^{(a-c+1)/2}\frac{(1-q^{a+1})(1-(-q)^{c+1})}{1+q}
,\\
&\tdx(q)_{i,1,k,0}^{a,1,c,0}
=\delta_{i,a}\delta_{k,c}q^{a+1}
,\\
&\tdx(q)_{i,0,k,1}^{a,1,c,0}
=\delta_{i,a+2}\delta_{k,c}(1-q^{a+1})(1-q^{a+2})
,\\
&\tdx(q)_{i,1,k,1}^{a,1,c,0}
=\delta_{i,a+1}\delta_{k,c-1}(-1)^{c}q^{(a+c+1)/2}(1+q)(1-q^{a+1})
,\\
&\tdx(q)_{i,0,k,0}^{a,0,c,1}
=\delta_{i,a-1}\delta_{k,c+1}(-1)^{c}q^{(a-c-1)/2}\frac{1-(-q)^{c+1}}{1+q}
,\\
&\tdx(q)_{i,1,k,0}^{a,0,c,1}
=\delta_{i,a-2}\delta_{k,c}
,\\
&\tdx(q)_{i,0,k,1}^{a,0,c,1}
=\delta_{i,a}\delta_{k,c}q^a
,\\
&\tdx(q)_{i,1,k,1}^{a,0,c,1}
=\delta_{i,a-1}\delta_{k,c-1}(-1)^{c+1}q^{(a+c-1)/2}(1+q)
,\\
&\tdx(q)_{i,0,k,0}^{a,1,c,1}
=\delta_{i,a}\delta_{k,c+2}q^{a-c}\frac{(1-(-q)^{c+1})(1-(-q)^{c+2})}{(1+q)^2}
,\\
&\tdx(q)_{i,1,k,0}^{a,1,c,1}
=\delta_{i,a-1}\delta_{k,c+1}(-1)^{c+1}q^{(a-c+1)/2}\frac{1-(-q)^{c+1}}{1+q}
,\\
&\tdx(q)_{i,0,k,1}^{a,1,c,1}
=\delta_{i,a+1}\delta_{k,c+1}(-1)^{c}q^{(a-c+1)/2}\frac{(1-q^{a+1})(1-(-q)^{c+1})}{1+q}
,\\
&\tdx(q)_{i,1,k,1}^{a,1,c,1}
=\delta_{i,a}\delta_{k,c}\left(1-(1-(-q)^{c+1})q^{a+1}\right)
.
\label{3dX mat el 8}
\end{align}
For simplicity, we also use the abbreviated notation $\tdx=\tdx(q)$.
We simply call $\tdx$ as the 3D X.
Then, the transition matrix in (\ref{tB trans mat def1}) is explicitly given as follows:
\begin{theorem}\label{my 3dX result}
For the quantum superalgebra associated with \ddtwoB{\cct}{\cc}, the transition matrix in (\ref{tB trans mat def1}) is given by
\begin{align}
\tm_{i,j,k,l}^{a,b,c,d}
=\tdx_{i,j,k,l}^{a,b,c,d}
,
\end{align}
where $\tdx$ is the 3D X given by (\ref{3dX mat el 1}) $\sim$ (\ref{3dX mat el 8}).
\end{theorem}
The proof of Theorem \ref{my 3dX result}. is available in Appendix \ref{app 3dX}.
\begin{corollary}\label{my inv cor 3dX}
\begin{align}
\tdx^{-1}=\tdx
.
\end{align}
\end{corollary}
\begin{proof}
This is shown exactly in the same way as Corollary \ref{my inv cor}.
\end{proof}
\subsubsection{The case (III) $\vcenter{\hbox{\protect\includegraphics{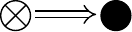}}}$}
In this case, the corresponding symmetrized Cartan matrix is given by
\begin{align}
DA
=
\begin{pmatrix}
0 & -1 \\
-1 & 1
\end{pmatrix}
,
\end{align}
and the corresponding positive roots are given by
\begin{align}
\prer&=\{\alpha_1+\alpha_2\}
,
\\
\prir&=\{\alpha_1,\alpha_1+2\alpha_2\}
,\\
\prar&=\{\alpha_2\}
.
\end{align}
Then, indices are specified as $i,k,a,c\in\mathbb{Z}_{\geq 0},{\ }j,l,b,d\in\{0,1\}$ for (\ref{tB trans mat def1}).
We set $\tdy(q)\in\mathrm{End}(F\otimes V\otimes F\otimes V)$ by
\begin{align}
&\tdy(q)(\ket{i}\otimes u_j\otimes \ket{k}\otimes u_l)
=\sum_{a,c\in\mathbb{Z}_{\geq 0},b,d\in\{0,1\}}
\tdy(q)_{i,j,k,l}^{a,b,c,d}
\ket{a}\otimes u_b\otimes \ket{c}\otimes u_d
,
\label{3dY}
\\
&\tdy(q)_{i,0,k,0}^{a,0,c,0}
=\delta_{i,a}\delta_{k,c}\left(1-(1-q^{c})(-q)^{a}\right)
,
\label{3dY mat el 1}
\\
&\tdy(q)_{i,1,k,0}^{a,0,c,0}
=\delta_{i,a-1}\delta_{k,c-1}q^{(a+c-1)/2}(1+q)
,
\\
&\tdy(q)_{i,0,k,1}^{a,0,c,0}
=\delta_{i,a+1}\delta_{k,c-1}(-1)^{a}q^{(a+c-1)/2}(1-q)(1-(-q)^{a+1})
,\\
&\tdy(q)_{i,1,k,1}^{a,0,c,0}
=\delta_{i,a}\delta_{k,c-2}(-1)^{a}q^{a+c-1}(1-q^{2})
,\\
&\tdy(q)_{i,0,k,0}^{a,1,c,0}
=\delta_{i,a+1}\delta_{k,c+1}(-1)^{a+1}q^{(a-c+1)/2}\frac{(1-(-q)^{a+1})(1-q^{c+1})}{1+q}
,\\
&\tdy(q)_{i,1,k,0}^{a,1,c,0}
=\delta_{i,a}\delta_{k,c}(-q)^{a+1}
,\\
&\tdy(q)_{i,0,k,1}^{a,1,c,0}
=\delta_{i,a+2}\delta_{k,c}(-1)^{a}\frac{(1-q)(1-(-q)^{a+1})(1-(-q)^{a+2})}{1+q}
,\\
&\tdy(q)_{i,1,k,1}^{a,1,c,0}
=\delta_{i,a+1}\delta_{k,c-1}(-1)^{a}q^{(a+c+1)/2}(1-q)(1-(-q)^{a+1})
,\\
&\tdy(q)_{i,0,k,0}^{a,0,c,1}
=\delta_{i,a-1}\delta_{k,c+1}q^{(a-c-1)/2}\frac{1-q^{c+1}}{1-q}
,\\
&\tdy(q)_{i,1,k,0}^{a,0,c,1}
=\delta_{i,a-2}\delta_{k,c}(-1)^{a}\frac{1+q}{1-q}
,\\
&\tdy(q)_{i,0,k,1}^{a,0,c,1}
=\delta_{i,a}\delta_{k,c}(-q)^{a}
,\\
&\tdy(q)_{i,1,k,1}^{a,0,c,1}
=-\delta_{i,a-1}\delta_{k,c-1}q^{(a+c-1)/2}(1+q)
,\\
&\tdy(q)_{i,0,k,0}^{a,1,c,1}
=\delta_{i,a}\delta_{k,c+2}(-1)^{a+1}q^{a-c}\frac{(1-q^{c+1})(1-q^{c+2})}{1-q^{2}}
,\\
&\tdy(q)_{i,1,k,0}^{a,1,c,1}
=\delta_{i,a-1}\delta_{k,c+1}q^{(a-c+1)/2}\frac{1-q^{c+1}}{1-q}
,\\
&\tdy(q)_{i,0,k,1}^{a,1,c,1}
=\delta_{i,a+1}\delta_{k,c+1}(-1)^{a}q^{(a-c+1)/2}\frac{(1-(-q)^{a+1})(1-q^{c+1})}{1+q}
,\\
&\tdy(q)_{i,1,k,1}^{a,1,c,1}
=\delta_{i,a}\delta_{k,c}\left(1-(1-q^{c+1})(-q)^{a+1}\right)
.
\label{3dY mat el 8}
\end{align}
For simplicity, we also use the abbreviated notation $\tdy=\tdy(q)$.
We simply call $\tdy$ as the 3D Y.
Then, the transition matrix in (\ref{tB trans mat def1}) is explicitly given as follows:
\begin{theorem}\label{my 3dY result}
For the quantum superalgebra associated with \ddtwoB{\cct}{\ccb}, the transition matrix in (\ref{tB trans mat def1}) is given by
\begin{align}
\tm_{i,j,k,l}^{a,b,c,d}
=\tdy_{i,j,k,l}^{a,b,c,d}
,
\end{align}
where $\tdy$ is the 3D Y given by (\ref{3dY mat el 1}) $\sim$ (\ref{3dY mat el 8}).
\end{theorem}
\begin{proof}
The relations $e_1,e_2$ satisfy are
\begin{align}
e_{1}^2=0
,\quad
e_{2}^3e_{1}
-(1-q-q^{-1})e_{2}^2e_{1}e_{2}
+(1-q-q^{-1})e_{2}e_{1}e_{2}^2
-e_{1}e_{2}^3
=0
,
\label{my3dY proof1}
\end{align}
Also, the quantum root vectors are given by
\begin{alignat}{2}
&e_{12}
=e_{1}e_{2}+qe_{2}e_{1}
,&\quad
&e_{21}
=e_{2}e_{1}+qe_{1}e_{2}
\label{my3dY proof2}
,\\
&e_{(12)2}
=\frac{e_{12}e_{2}-e_{2}e_{12}}{q^{1/2}+q^{-1/2}}
,&\quad
&e_{2(21)}
=\frac{e_{2}e_{21}-e_{21}e_{2}}{q^{1/2}+q^{-1/2}}
.
\label{my3dY proof3}
\end{alignat}
On the other hand, for the case (II) \ddtwoB{\cct}{\cc}, the relations for the generators are
\begin{align}
\geneBoe_{1}^2=0
,\quad
\geneBoe_{2}^3\geneBoe_{1}
-(1+q+q^{-1})\geneBoe_{2}^2\geneBoe_{1}\geneBoe_{2}
+(1+q+q^{-1})\geneBoe_{2}\geneBoe_{1}\geneBoe_{2}^2
-\geneBoe_{1}\geneBoe_{2}^3
=0
,
\label{my3dY proof4}
\end{align}
where we write the generators by $\geneBoe_{i}$ instead of $e_{i}$ to avoid confusion.
Also, the quantum root vectors for the case (II) \ddtwoB{\cct}{\cc} are given by
\begin{alignat}{2}
&\geneBoe_{12}
=\geneBoe_{1}\geneBoe_{2}-q\geneBoe_{2}\geneBoe_{1}
,&\quad
&\geneBoe_{21}
=\geneBoe_{2}\geneBoe_{1}-q\geneBoe_{1}\geneBoe_{2}
\label{my3dY proof5}
,\\
&\geneBoe_{(12)2}
=\frac{\geneBoe_{12}\geneBoe_{2}-\geneBoe_{2}\geneBoe_{12}}{q^{1/2}+q^{-1/2}}
,&\quad
&\geneBoe_{2(21)}
=\frac{\geneBoe_{2}\geneBoe_{21}-\geneBoe_{21}\geneBoe_{2}}{q^{1/2}+q^{-1/2}}
.
\label{my3dY proof6}
\end{alignat}
Apparently, (\ref{my3dY proof4}),(\ref{my3dY proof5}) and the numerators of (\ref{my3dY proof6}) correspond to (\ref{my3dY proof1}),(\ref{my3dY proof2}) and the numerators of (\ref{my3dY proof3}) with a replacement $q\to -q$.
\par
Here, (\ref{tB trans mat def1}) for the case (II) \ddtwoB{\cct}{\cc} and (III) \ddtwoB{\cct}{\ccb} are explicitly written as follows:
\begin{gather}
\begin{split}
&
\frac{\geneBoe_{2}^{a}}{[a]_{q^{1/2}}!}
\left(
\frac{\geneBoe_{12}\geneBoe_{2}-\geneBoe_{2}\geneBoe_{12}}{q^{1/2}+q^{-1/2}}
\right)^{{\!\!}b}
\frac{(\geneBoe_{1}\geneBoe_{2}-q\geneBoe_{2}\geneBoe_{1})^{c}}{[c]_{q^{-1/2},(-1)}!}
\geneBoe_{1}^{d}
\\
&=
\sum_{i,j,k,l}
\tdx_{i,j,k,l}^{a,b,c,d}
\geneBoe_{1}^{l}
\frac{(\geneBoe_{2}\geneBoe_{1}-q\geneBoe_{1}\geneBoe_{2})^{k}}{[k]_{q^{-1/2},(-1)}!}
\left(
\frac{\geneBoe_{2}\geneBoe_{21}-\geneBoe_{21}\geneBoe_{2}}{q^{1/2}+q^{-1/2}}
\right)^{{\!\!}j}
\frac{\geneBoe_{2}^{i}}{[i]_{q^{1/2}}!}
.
\end{split}
\label{my3dY proof7}
\\
\begin{split}
&
\frac{e_{2}^{a}}{[a]_{q^{1/2},(-1)}!}
\left(
\frac{e_{12}e_{2}-e_{2}e_{12}}{q^{1/2}+q^{-1/2}}
\right)^{{\!\!}b}
\frac{(e_{1}e_{2}-qe_{2}e_{1})^{c}}{[c]_{q^{-1/2}}!}
e_{1}^{d}
\\
&=
\sum_{i,j,k,l}
\tm_{i,j,k,l}^{a,b,c,d}
e_{1}^{l}
\frac{(e_{2}e_{1}-qe_{1}e_{2})^{k}}{[k]_{q^{-1/2}}!}
\left(
\frac{e_{2}e_{21}-e_{21}e_{2}}{q^{1/2}+q^{-1/2}}
\right)^{{\!\!}j}
\frac{e_{2}^{i}}{[i]_{q^{1/2},(-1)}!}
.
\end{split}
\label{my3dY proof8}
\end{gather}
Comparing (\ref{my3dY proof7}) with $q\to -q$ and (\ref{my3dY proof8}), we obtain the following relation:
\begin{align}
\begin{split}
\tm_{i,j,k,l}^{a,b,c,d}
=&
\left.
\left(
\frac{1}{q^{1/2}+q^{-1/2}}
\right)^{{\!\! j-b}}
\frac{[a]_{q^{1/2}}![c]_{q^{-1/2},(-1)}!}{[i]_{q^{1/2}}![k]_{q^{-1/2},(-1)}!}
\tdx_{i,j,k,l}^{a,b,c,d}
\right|_{q\to -q}
\\
&\times
(q^{1/2}+q^{-1/2})^{j-b}
\frac{[i]_{q^{1/2},(-1)}![k]_{q^{-1/2}}!}{[a]_{q^{1/2},(-1)}![c]_{q^{-1/2}}!}
\end{split}
\label{my3dY proof9}
\\
=&
(-1)^{i(i-1)/4-k(k-1)/4-a(a-1)/4+c(c-1)/4+j/2-b/2}
\left(
\frac{1+q}{1-q}
\right)^{{\!\!}j-b}
\tdx(-q)_{i,j,k,l}^{a,b,c,d}
,
\label{my3dY proof10}
\end{align}
where we used
\begin{align}
&[m]_{(-q)^{1/2}}!
=(-1)^{m(m-1)/4}[m]_{q^{1/2},(-1)}!
\label{my3dY proof11}
,\\
&[m]_{(-q)^{-1/2},(-1)}!
=(-1)^{-m(m-1)/4}[m]_{q^{-1/2}}!
\label{my3dY proof12}
,\\
&\frac{1}{((-q)^{1/2}+(-q)^{-1/2})^{m}}
=(-1)^{m/2}
\left(
\frac{1+q}{1-q}
\right)^{{\!\!}m}
\frac{1}{(q^{1/2}+q^{-1/2})^{m}}
.
\label{my3dY proof13}
\end{align}
We then obtain the desired result by direct calculations.
We note that (\ref{my3dY proof10}) involves $(-1)^{1/2}$, but no matrix elements of the 3D Y involve it.
\end{proof}
\begin{corollary}\label{my inv cor 3dY}
\begin{align}
\tdy^{-1}=\tdy
.
\end{align}
\end{corollary}
\begin{proof}
This is shown exactly in the same way as Corollary \ref{my inv cor}.
\end{proof}
\subsubsection{The case (IV) $\vcenter{\hbox{\protect\includegraphics{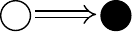}}}$}
In this case, the corresponding symmetrized Cartan matrix is given by
\begin{align}
DA
=
\begin{pmatrix}
2 & -1 \\
-1 & 1
\end{pmatrix}
,
\end{align}
and the corresponding positive roots are given by
\begin{align}
\prer&=\{\alpha_1,\alpha_1+2\alpha_2\}
,
\\
\prir&=\{\}
,\\
\prar&=\{\alpha_2,\alpha_1+\alpha_2\}
.
\end{align}
Then, the indices are specified as $i,j,k,l,a,b,c,d\in\mathbb{Z}_{\geq 0}$ for (\ref{tB trans mat def1}).
We write matrix elements of the transition matrix in (\ref{tB trans mat def1}) by
\begin{align}
\tm_{i,j,k,l}^{a,b,c,d}
=\tdz(q)_{i,j,k,l}^{a,b,c,d}
.
\label{3dZ mat el}
\end{align}
For simplicity, we also use the abbreviated notation $\tdz=\tdz(q)$.
We simply call $\tdz$ as the 3D Z.
\par
At present, an explict formula for the 3D Z is yet to be constructed.
In Appendix \ref{app 3dZ}, we present recurrence equations for the 3D Z.
We can calculate any matrix elements of the 3D Z by using a computer program via these equations.
\begin{example}\label{3dZ example}
The following is the list of all the non-zero elements of $\tdz_{0,1,1,2}^{a,b,c,d}$:
\begin{align}
&\tdz_{0,1,1,2}^{0,0,3,1}
=q^4(1+q)^2(1+q^2)
,\\
&\tdz_{0,1,1,2}^{0,1,1,2}
=-q^2(1-q^4-q^7)
,\\
&\tdz_{0,1,1,2}^{1,0,2,2}
=-q^3(1+q)(1+q+q^2+q^3+q^5)
,\\
&\tdz_{0,1,1,2}^{1,1,0,3}
=1-q^4-q^7
,\\
&\tdz_{0,1,1,2}^{2,0,1,3}
=-\frac{q^5(1+q^3)}{1-q}
,\\
&\tdz_{0,1,1,2}^{3,0,0,4}
=\frac{q^2(1+q)}{1-q}
.
\end{align}
\end{example}
\subsection{Transition matrices of PBW bases of type B of rank 3 and 3D reflection equation}\label{sec 53}
In this section, we condider the transition matrices of the PBW bases of $\UtBp$ of rank 3, so $m+n=3$.
All possible Dynkin diagrams associated with admissible realizations are given in Table \ref{tB all dynkin r3}.
\begin{table}[htb]
\centering
\caption{}
\label{tB all dynkin r3}
\begin{tabular}{c|c}\hline
$\GG$ & Dynkin diagram \\\hline
$\tB[7|0]$ &  \ddthreeBWul{\cc}{\cc}{\cc}{\ee_1-\ee_2}{\ee_2-\ee_3}{\ee_3}\\\hline
$\tB[5|2]$ &  \ddthreeBWul{\cc}{\cct}{\ccb}{\ee_1-\ee_2}{\ee_2-\dd_3}{\dd_3}
\quad \ddthreeBWul{\cct}{\cct}{\cc}{\ee_1-\dd_2}{\dd_2-\ee_3}{\ee_3}
\quad \ddthreeBWul{\cct}{\cc}{\cc}{\dd_1-\ee_2}{\ee_2-\ee_3}{\ee_3}
\\\hline
$\tB[3|4]$ &  \ddthreeBWul{\cct}{\cc}{\ccb}{\ee_1-\dd_2}{\dd_2-\dd_3}{\dd_3}
\quad \ddthreeBWul{\cct}{\cct}{\ccb}{\dd_1-\ee_2}{\ee_2-\dd_3}{\dd_3}
\quad \ddthreeBWul{\cc}{\cct}{\cc}{\dd_1-\dd_2}{\dd_2-\ee_3}{\ee_3}
\\\hline
$\tB[1|6]$ &  \ddthreeBWul{\cc}{\cc}{\ccb}{\dd_1-\dd_2}{\dd_2-\dd_3}{\dd_3}\\\hline
\end{tabular}
\end{table}
For the case of rank 3, quantum root vectors are given by (\ref{tB r3 qroot vec1}) and (\ref{tB r3 qroot vec2}), so the transition matrix in (\ref{tm1}) is given as follows:
\begin{align}
\begin{split}
&e_{3}^{(o_1)}
\pot^{(o_2)}
\poo^{(o_3)}
e_{2}^{(o_4)}
\ott^{(o_5)}
\oot^{(o_6)}
\ooo^{(o_7)}
\oop^{(o_8)}
e_{1}^{(o_9)}
\\
&=
\sum_{i_1,i_2,i_3,i_4,i_5,i_6,i_7,i_8,i_9}
\gamma_{i_1,i_2,i_3,i_4,i_5,i_6,i_7,i_8,i_9}^{o_1,o_2,o_3,o_4,o_5,o_6,o_7,o_8,o_9}
e_{1}^{(i_9)}
\toX^{(i_8)}
\rltol^{(i_7)}
\rlrltoll^{(i_6)}
\tlrlrltolll^{(i_5)}
e_{2}^{(i_4)}
\rtX^{(i_3)}
\rlrtl^{(i_2)}
e_{3}^{(i_1)}
,
\end{split}
\label{tB trans mat def r3}
\end{align}
where the domain of indices is specified below.
In order to attribute the transition matrix in (\ref{tB trans mat def r3}) to a composition of transition matrices of rank 2, we exploit the following transition matrices $\stmB^{(x)}$ and $\stmC^{(y)}$:
\begin{align}
e_{2}^{(a)}e_{12}^{(b)}e_{1}^{(c)}
&=\sum_{i,j,k}
\stmB^{(2|1)}{}_{i,j,k}^{a,b,c}
e_{1}^{(k)}e_{21}^{(j)}e_{2}^{(i)}
,
\label{tB trans mat r3 eq1}
\\
e_{2}^{(a)}e_{(((12)3)3)2}^{(b)}e_{((12)3)3}^{(c)}
&=\sum_{i,j,k}
\stmB^{(2|1233)}{}_{i,j,k}^{a,b,c}
e_{((12)3)3}^{(k)}e_{2(((12)3)3)}^{(j)}e_{2}^{(i)}
,
\label{tB trans mat r3 eq2}
\\
e_{2}^{(a)}e_{((3(32))1)2}^{(b)}e_{(3(32))1}^{(c)}
&=\sum_{i,j,k}
\stmB^{(2|3321)}{}_{i,j,k}^{a,b,c}
e_{(3(32))1}^{(k)}e_{2((3(32))1)}^{(j)}e_{2}^{(i)}
,
\label{tB trans mat r3 eq3}
\\
e_{(23)3}^{(a)}e_{1((23)3)}^{(b)}e_{1}^{(c)}
&=\sum_{i,j,k}
\stmB^{(233|1)}{}_{i,j,k}^{a,b,c}
e_{1}^{(k)}e_{((23)3)1}^{(j)}e_{(23)3}^{(i)}
,
\label{tB trans mat r3 eq4}
\\
e_{3(32)}^{(a)}e_{1(3(32))}^{(b)}e_{1}^{(c)}
&=\sum_{i,j,k}
\stmB^{(332|1)}{}_{i,j,k}^{a,b,c}
e_{1}^{(k)}e_{(3(32))1}^{(j)}e_{3(32)}^{(i)}
,
\label{tB trans mat r3 eq5}
\\
e_{(23)3}^{(a)}e_{(21)((23)3)}^{(b)}e_{21}^{(c)}
&=\sum_{i,j,k}
\stmB^{(233|21)}{}_{i,j,k}^{a,b,c}
e_{21}^{(k)}e_{((23)3)(21)}^{(j)}e_{(23)3}^{(i)}
,
\label{tB trans mat r3 eq6}
\\
e_{3(32)}^{(a)}e_{(12)(3(32))}^{(b)}e_{12}^{(c)}
&=\sum_{i,j,k}
\stmB^{(332|12)}{}_{i,j,k}^{a,b,c}
e_{12}^{(k)}e_{(3(32))(12)}^{(j)}e_{3(32)}^{(i)}
,
\label{tB trans mat r3 eq7}
\\
e_{3}^{(a)}e_{(23)3}^{(b)}e_{23}^{(c)}e_{2}^{(d)}
&=\sum_{i,j,k,l}
\stmC^{(3|2)}{}_{i,j,k,l}^{a,b,c,d}
e_{2}^{(l)}e_{32}^{(k)}e_{3(32)}^{(j)}e_{3}^{(i)}
,
\label{tB trans mat r3 eq8}
\\
e_{3}^{(a)}e_{((21)3)3}^{(b)}e_{(21)3}^{(c)}e_{21}^{(d)}
&=\sum_{i,j,k,l}
\stmC^{(3|21)}{}_{i,j,k,l}^{a,b,c,d}
e_{21}^{(l)}e_{3(21)}^{(k)}e_{3(3(21))}^{(j)}e_{3}^{(i)}
,
\label{tB trans mat r3 eq9}
\\
e_{3}^{(a)}e_{((12)3)3}^{(b)}e_{(12)3}^{(c)}e_{12}^{(d)}
&=\sum_{i,j,k,l}
\stmC^{(3|12)}{}_{i,j,k,l}^{a,b,c,d}
e_{12}^{(l)}e_{3(12)}^{(k)}e_{3(3(12))}^{(j)}e_{3}^{(i)}
,
\label{tB trans mat r3 eq10}
\\
e_{23}^{(a)}e_{(1(23))(23)}^{(b)}e_{1(23)}^{(c)}e_{1}^{(d)}
&=\sum_{i,j,k,l}
\stmC^{(23|1)}{}_{i,j,k,l}^{a,b,c,d}
e_{1}^{(l)}e_{(23)1}^{(k)}e_{(23)((23)1)}^{(j)}e_{23}^{(i)}
,
\label{tB trans mat r3 eq11}
\\
e_{32}^{(a)}e_{(1(32))(32)}^{(b)}e_{1(32)}^{(c)}e_{1}^{(d)}
&=\sum_{i,j,k,l}
\stmC^{(32|1)}{}_{i,j,k,l}^{a,b,c,d}
e_{1}^{(l)}e_{(32)1}^{(k)}e_{(32)((32)1)}^{(j)}e_{32}^{(i)}
,
\label{tB trans mat r3 eq12}
\end{align}
where the domain of indices will be specified.
For each case in Table \ref{tB all dynkin r3}, $\stmB^{(x)}$ and $\stmC^{(y)}$ are identified with the 3D operators we have already introduced.
\par
Then, by using $\stmB^{(x)},\stmC^{(y)}$, Proposition \ref{tB qroot rels1}, \ref{tB qroot rels2} and Proposition \ref{tB qroot rels1 chi}, \ref{tB qroot rels2 chi}, we can construct the transition matrix in (\ref{tB trans mat def r3}) in two ways.
The first way is given by
\begin{align}
&e_{3}^{(o_1)}
\pot^{(o_2)}
\poo^{(o_3)}
\underline{
e_{2}^{(o_4)}
\ott^{(o_5)}
\oot^{(o_6)}
}
\ooo^{(o_7)}
\oop^{(o_8)}
e_{1}^{(o_9)}
\label{tB lhs1 eq1}
\\
&=\sum
\stmB^{(2|1233)}{}_{y_4,y_5,y_6}^{o_4,o_5,o_6}
e_{3}^{(o_1)}
\pot^{(o_2)}
\underline{
\poo^{(o_3)}
\oot^{(y_6)}
}
\tlllotlrlrl^{(y_5)}
\underline{
e_{2}^{(y_4)}
\ooo^{(o_7)}
}
\oop^{(o_8)}
e_{1}^{(o_9)}
\label{tB lhs1 eq2}
\\
&=\sum
(-1)^{\phaseIII y_4o_7+\phaseIIb o_3y_6}
\stmB^{(2|1233)}{}_{y_4,y_5,y_6}^{o_4,o_5,o_6}
e_{3}^{(o_1)}
\pot^{(o_2)}
\oot^{(y_6)}
\poo^{(o_3)}
\underline{
\tlllotlrlrl^{(y_5)}
}
\ooo^{(o_7)}
\underline{
e_{2}^{(y_4)}
\oop^{(o_8)}
e_{1}^{(o_9)}
}
\label{tB lhs1 eq3}
\\
\begin{split}
&=\sum
(-1)^{\phaseIII (y_4o_7+y_5)+\phaseIIb o_3y_6}
\stmB^{(2|1233)}{}_{y_4,y_5,y_6}^{o_4,o_5,o_6}
\stmB^{(2|1)}{}_{x_4,y_8,y_9}^{y_4,o_8,o_9}
\\
&\spaceDb
\times
e_{3}^{(o_1)}
\pot^{(o_2)}
{\ }
\underline{
\oot^{(y_6)}
}
{\ }
\underline{
\poo^{(o_3)}
e_{(1(23))(23)}^{(y_5)}
\ooo^{(o_7)}
e_{1}^{(y_9)}
}
\toX^{(y_8)}
e_{2}^{(x_4)}
\end{split}
\label{tB lhs1 eq4}
\\
\begin{split}
&=\sum
(-1)^{\phaseIII (y_4o_7+y_5)+\phaseIIb o_3y_6}
\stmB^{(2|1233)}{}_{y_4,y_5,y_6}^{o_4,o_5,o_6}
\stmB^{(2|1)}{}_{x_4,y_8,y_9}^{y_4,o_8,o_9}
\stmC^{(23|1)}{}_{x_3,x_5,x_7,x_9}^{o_3,y_5,o_7,y_9}
\\
&\spaceDb
\times
e_{3}^{(o_1)}
\underline{
\pot^{(o_2)}
e_{1((23)3)}^{(y_6)}
e_{1}^{(x_9)}
}
\ltrlo^{(x_7)}
\ltrllltrlol^{(x_5)}
\poo^{(x_3)}
\toX^{(y_8)}
e_{2}^{(x_4)}
\end{split}
\label{tB lhs1 eq5}
\\
\begin{split}
&=\sum
(-1)^{\phaseIII (y_4o_7+y_5)+\phaseIIb o_3y_6}
\stmB^{(2|1233)}{}_{y_4,y_5,y_6}^{o_4,o_5,o_6}
\stmB^{(2|1)}{}_{x_4,y_8,y_9}^{y_4,o_8,o_9}
\stmC^{(23|1)}{}_{x_3,x_5,x_7,x_9}^{o_3,y_5,o_7,y_9}
\stmB^{(233|1)}{}_{y_2,x_6,i_9}^{o_2,y_6,x_9}
\\
&\spaceDb
\times
\underline{
e_{3}^{(o_1)}
e_{1}^{(i_9)}
}
\lltrlrlo^{(x_6)}
\underline{
\pot^{(y_2)}
\ltrlo^{(x_7)}
}
{\ }
\underline{
\ltrllltrlol^{(x_5)}
}
{\ }
\underline{
\poo^{(x_3)}
\toX^{(y_8)}
}
e_{2}^{(x_4)}
\end{split}
\label{tB lhs1 eq6}
\\
\begin{split}
&=\sum
(-1)^{\phaseI (o_1i_9+x_5)+\phaseII (x_3y_8+x_5)+\phaseIII (y_4o_7+y_5)+\phaseIIb o_3y_6+\phaseIIIb y_2x_7}
\\
&\spaceDb
\times
\stmB^{(2|1233)}{}_{y_4,y_5,y_6}^{o_4,o_5,o_6}
\stmB^{(2|1)}{}_{x_4,y_8,y_9}^{y_4,o_8,o_9}
\stmC^{(23|1)}{}_{x_3,x_5,x_7,x_9}^{o_3,y_5,o_7,y_9}
\stmB^{(233|1)}{}_{y_2,x_6,i_9}^{o_2,y_6,x_9}
\\
&\spaceDb
\times
e_{1}^{(i_9)}
e_{3}^{(o_1)}
{\ }
\underline{
\lltrlrlo^{(x_6)}
}
{\ }
\underline{
\ltrlo^{(x_7)}
}
{\ }
\underline{
\pot^{(y_2)}
e_{(21)((23)3)}^{(x_5)}
\toX^{(y_8)}
}
\poo^{(x_3)}
e_{2}^{(x_4)}
\end{split}
\label{tB lhs1 eq7}
\\
\begin{split}
&=\sum
(-1)^{\phaseI (o_1i_9+x_5+x_7)+\phaseII (x_3y_8+x_5)+\phaseIII (y_4o_7+y_5)+\phaseIIb o_3y_6+\phaseIIIb y_2x_7}
\\
&\spaceDb
\times
\stmB^{(2|1233)}{}_{y_4,y_5,y_6}^{o_4,o_5,o_6}
\stmB^{(2|1)}{}_{x_4,y_8,y_9}^{y_4,o_8,o_9}
\stmC^{(23|1)}{}_{x_3,x_5,x_7,x_9}^{o_3,y_5,o_7,y_9}
\stmB^{(233|1)}{}_{y_2,x_6,i_9}^{o_2,y_6,x_9}
\stmB^{(233|21)}{}_{x_2,i_5,x_8}^{y_2,x_5,y_8}
\\
&\spaceDb
\times
e_{1}^{(i_9)}
\underline{
e_{3}^{(o_1)}
e_{((21)3)3}^{(x_6)}
e_{(21)3}^{(x_7)}
\toX^{(x_8)}
}
\lltrlrlltol^{(i_5)}
\pot^{(x_2)}
\poo^{(x_3)}
e_{2}^{(x_4)}
\end{split}
\label{tB lhs1 eq8}
\\
\begin{split}
&=\sum
(-1)^{\phaseI (o_1i_9+x_5+x_7)+\phaseII (x_3y_8+x_5)+\phaseIII (y_4o_7+y_5)+\phaseIIb o_3y_6+\phaseIIIb y_2x_7}
\stmB^{(2|1233)}{}_{y_4,y_5,y_6}^{o_4,o_5,o_6}
\stmB^{(2|1)}{}_{x_4,y_8,y_9}^{y_4,o_8,o_9}
\\
&\spaceDb
\times
\stmC^{(23|1)}{}_{x_3,x_5,x_7,x_9}^{o_3,y_5,o_7,y_9}
\stmB^{(233|1)}{}_{y_2,x_6,i_9}^{o_2,y_6,x_9}
\stmB^{(233|21)}{}_{x_2,i_5,x_8}^{y_2,x_5,y_8}
\stmC^{(3|21)}{}_{x_1,i_6,i_7,i_8}^{o_1,x_6,x_7,x_8}
\\
&\spaceDb
\times
e_{1}^{(i_9)}
\toX^{(i_8)}
\rltol^{(i_7)}
\rlrltoll^{(i_6)}
\underline{
e_{3}^{(x_1)}
\lltrlrlltol^{(i_5)}
}
\pot^{(x_2)}
\poo^{(x_3)}
e_{2}^{(x_4)}
\end{split}
\label{tB lhs1 eq9}
\\
\begin{split}
&=\sum
(-1)^{\phaseI (o_1i_9+x_5+x_7)+\phaseII (x_3y_8+x_5)+\phaseIII (y_4o_7+y_5)+\phaseIb x_1i_5+\phaseIIb o_3y_6+\phaseIIIb y_2x_7}
\stmB^{(2|1233)}{}_{y_4,y_5,y_6}^{o_4,o_5,o_6}
\stmB^{(2|1)}{}_{x_4,y_8,y_9}^{y_4,o_8,o_9}
\\
&\spaceDb
\times
\stmC^{(23|1)}{}_{x_3,x_5,x_7,x_9}^{o_3,y_5,o_7,y_9}
\stmB^{(233|1)}{}_{y_2,x_6,i_9}^{o_2,y_6,x_9}
\stmB^{(233|21)}{}_{x_2,i_5,x_8}^{y_2,x_5,y_8}
\stmC^{(3|21)}{}_{x_1,i_6,i_7,i_8}^{o_1,x_6,x_7,x_8}
\\
&\spaceDb
\times
e_{1}^{(i_9)}
\toX^{(i_8)}
\rltol^{(i_7)}
\rlrltoll^{(i_6)}
\lltrlrlltol^{(i_5)}
\underline{
e_{3}^{(x_1)}
\pot^{(x_2)}
\poo^{(x_3)}
e_{2}^{(x_4)}
}
\end{split}
\label{tB lhs1 eq10}
\\
\begin{split}
&=\sum
(-1)^{\phaseI (o_1i_9+x_5+x_7)+\phaseII (x_3y_8+x_5)+\phaseIII (y_4o_7+y_5)+\phaseIb x_1i_5+\phaseIIb o_3y_6+\phaseIIIb y_2x_7}
\stmB^{(2|1233)}{}_{y_4,y_5,y_6}^{o_4,o_5,o_6}
\stmB^{(2|1)}{}_{x_4,y_8,y_9}^{y_4,o_8,o_9}
\\
&\spaceDb
\times
\stmC^{(23|1)}{}_{x_3,x_5,x_7,x_9}^{o_3,y_5,o_7,y_9}
\stmB^{(233|1)}{}_{y_2,x_6,i_9}^{o_2,y_6,x_9}
\stmB^{(233|21)}{}_{x_2,i_5,x_8}^{y_2,x_5,y_8}
\stmC^{(3|21)}{}_{x_1,i_6,i_7,i_8}^{o_1,x_6,x_7,x_8}
\stmC^{(3|2)}{}_{i_1,i_2,i_3,i_4}^{x_1,x_2,x_3,x_4}
\\
&\spaceDb
\times
e_{1}^{(i_9)}
\toX^{(i_8)}
\rltol^{(i_7)}
\rlrltoll^{(i_6)}
\underline{
\lltrlrlltol^{(i_5)}
}
e_{2}^{(i_4)}
\rtX^{(i_3)}
\rlrtl^{(i_2)}
e_{3}^{(i_1)}
\end{split}
\label{tB lhs1 eq11}
\\
\begin{split}
&=\sum
(-1)^{\phaseI (o_1i_9+x_5+x_7)+\phaseII (x_3y_8+x_5)+\phaseIII (y_4o_7+y_5)+\phaseIb x_1i_5+\phaseIIb o_3y_6+\phaseIIIb y_2x_7}
\stmB^{(2|1233)}{}_{y_4,y_5,y_6}^{o_4,o_5,o_6}
\stmB^{(2|1)}{}_{x_4,y_8,y_9}^{y_4,o_8,o_9}
\\
&\spaceDb
\times
\stmC^{(23|1)}{}_{x_3,x_5,x_7,x_9}^{o_3,y_5,o_7,y_9}
\stmB^{(233|1)}{}_{y_2,x_6,i_9}^{o_2,y_6,x_9}
\stmB^{(233|21)}{}_{x_2,i_5,x_8}^{y_2,x_5,y_8}
\stmC^{(3|21)}{}_{x_1,i_6,i_7,i_8}^{o_1,x_6,x_7,x_8}
\stmC^{(3|2)}{}_{i_1,i_2,i_3,i_4}^{x_1,x_2,x_3,x_4}
\\
&\spaceDb
\times
e_{1}^{(i_9)}
\toX^{(i_8)}
\rltol^{(i_7)}
\rlrltoll^{(i_6)}
\tlrlrltolll^{(i_5)}
e_{2}^{(i_4)}
\rtX^{(i_3)}
\rlrtl^{(i_2)}
e_{3}^{(i_1)}
,
\end{split}
\label{tB lhs1 eq12}
\end{align}
where summations are taken on $i_k,x_k{\ }(k=1,\cdots,9)$, $y_k{\ }(k=2,4,5,6,8,9)$.
We have put the underlines to the parts to be rewritten.
We used $\phaseI,\phaseII,\phaseIII$ given by (\ref{te phase}) and we have also set
\begin{align}
\begin{split}
\phaseIb
&=p(\alpha_3)p(\alpha_1+2\alpha_2+2\alpha_3)
,\\
\phaseIIb
&=p(\alpha_2+\alpha_3)p(\alpha_1+\alpha_2+2\alpha_3)
,\\
\phaseIIIb
&=p(\alpha_1+\alpha_2+\alpha_3)p(\alpha_2+2\alpha_3)
.
\end{split}
\label{tre phase}
\end{align}
We note that $\phaseIb$, $\phaseIIb$ and $\phaseIIIb$ are actually equal to $\phaseI$, $\phaseII$ and $\phaseIII$, respectively because the parts whose coefficients are 2 do not contribute the parity.
We exploit both of them for a better understanding.
The details of the above procedure are as follows.
For (\ref{tB lhs1 eq1}), we used (\ref{tB trans mat r3 eq2}).
For (\ref{tB lhs1 eq2}), we used (\ref{tB qroot rel10}) and (\ref{tB qroot rel8}).
For (\ref{tB lhs1 eq3}), we used (\ref{tB trans mat r3 eq1}) and (\ref{tB qroot rel4}).
For (\ref{tB lhs1 eq4}), we used (\ref{tB trans mat r3 eq11}) and (\ref{tB qroot rel3}).
For (\ref{tB lhs1 eq5}), we used (\ref{tB trans mat r3 eq4}).
For (\ref{tB lhs1 eq6}), we used $[e_1,e_3]=0$, (\ref{tB qroot rel6}), (\ref{tB qroot rel9}) and (\ref{tB qroot rel11}).
For (\ref{tB lhs1 eq7}), we used (\ref{tB qroot rel2}), (\ref{tB qroot rel4}) and (\ref{tB trans mat r3 eq6}).
For (\ref{tB lhs1 eq8}), we used (\ref{tB trans mat r3 eq9}).
For (\ref{tB lhs1 eq9}), we used (\ref{tB qroot rel12}).
For (\ref{tB lhs1 eq10}), we used (\ref{tB trans mat r3 eq8}).
For (\ref{tB lhs1 eq11}), we used (\ref{tB qroot rel7}).
\par
Similarly, the second way is given by
\begin{align}
&\underline{
e_{3}^{(o_1)}
\pot^{(o_2)}
\poo^{(o_3)}
e_{2}^{(o_4)}
}
\ott^{(o_5)}
\oot^{(o_6)}
\ooo^{(o_7)}
\oop^{(o_8)}
e_{1}^{(o_9)}
\label{tB rhs1 eq1}
\\
&=\sum
\stmC^{(3|2)}{}_{x_1,y_2,x_3,y_4}^{o_1,o_2,o_3,o_4}
e_{2}^{(y_4)}
\rtX^{(x_3)}
\rlrtl^{(y_2)}
\underline{
e_{3}^{(x_1)}
\ott^{(o_5)}
}
\oot^{(o_6)}
\ooo^{(o_7)}
\oop^{(o_8)}
e_{1}^{(o_9)}
\label{tB rhs1 eq2}
\\
&=\sum
(-1)^{\phaseIb x_1o_5}
\stmC^{(3|2)}{}_{x_1,y_2,x_3,y_4}^{o_1,o_2,o_3,o_4}
e_{2}^{(y_4)}
\rtX^{(x_3)}
\rlrtl^{(y_2)}
{\ }
\underline{
\ott^{(o_5)}
}
{\ }
\underline{
e_{3}^{(x_1)}
\oot^{(o_6)}
\ooo^{(o_7)}
\oop^{(o_8)}
}
e_{1}^{(o_9)}
\label{tB rhs1 eq3}
\\
\begin{split}
&=\sum
(-1)^{\phaseIb x_1o_5}
\stmC^{(3|2)}{}_{x_1,y_2,x_3,y_4}^{o_1,o_2,o_3,o_4}
\stmC^{(3|12)}{}_{i_1,y_6,x_7,y_8}^{x_1,o_6,o_7,o_8}
\\
&\spaceDb
\times
e_{2}^{(y_4)}
\rtX^{(x_3)}
\underline{
\rlrtl^{(y_2)}
e_{(12)(3(32))}^{(o_5)}
\oop^{(y_8)}
}
\rlotl^{(x_7)}
\rlrlotll^{(y_6)}
e_{3}^{(i_1)}
e_{1}^{(o_9)}
\end{split}
\label{tB rhs1 eq4}
\\
\begin{split}
&=\sum
(-1)^{\phaseIb x_1o_5}
\stmC^{(3|2)}{}_{x_1,y_2,x_3,y_4}^{o_1,o_2,o_3,o_4}
\stmC^{(3|12)}{}_{i_1,y_6,x_7,y_8}^{x_1,o_6,o_7,o_8}
\stmB^{(332|12)}{}_{x_2,y_5,x_8}^{y_2,o_5,y_8}
\\
&\spaceDb
\times
e_{2}^{(y_4)}
\underline{
\rtX^{(x_3)}
\oop^{(x_8)}
}
\lrlrtlllotl^{(y_5)}
\underline{
\rlrtl^{(x_2)}
\rlotl^{(x_7)}
}
{\ }
\underline{
\rlrlotll^{(y_6)}
}
{\ }
\underline{
e_{3}^{(i_1)}
e_{1}^{(o_9)}
}
\end{split}
\label{tB rhs1 eq5}
\\
\begin{split}
&=\sum
(-1)^{\phaseI i_1o_9+\phaseII x_3x_8+\phaseIb x_1o_5+\phaseIIIb x_2x_7}
\stmC^{(3|2)}{}_{x_1,y_2,x_3,y_4}^{o_1,o_2,o_3,o_4}
\stmC^{(3|12)}{}_{i_1,y_6,x_7,y_8}^{x_1,o_6,o_7,o_8}
\stmB^{(332|12)}{}_{x_2,y_5,x_8}^{y_2,o_5,y_8}
\\
&\spaceDb
\times
e_{2}^{(y_4)}
\oop^{(x_8)}
\rtX^{(x_3)}
{\ }
\underline{
\lrlrtlllotl^{(y_5)}
}
{\ }
\underline{
\rlotl^{(x_7)}
}
{\ }
\underline{
\rlrtl^{(x_2)}
e_{1(3(32))}^{(y_6)}
e_{1}^{(o_9)}
}
e_{3}^{(i_1)}
\end{split}
\label{tB rhs1 eq6}
\\
\begin{split}
&=\sum
(-1)^{\phaseI (i_1o_9+x_7+y_5)+\phaseII (x_3x_8+y_5)+\phaseIb x_1o_5+\phaseIIIb x_2x_7}
\\
&\spaceDb
\times
\stmC^{(3|2)}{}_{x_1,y_2,x_3,y_4}^{o_1,o_2,o_3,o_4}
\stmC^{(3|12)}{}_{i_1,y_6,x_7,y_8}^{x_1,o_6,o_7,o_8}
\stmB^{(332|12)}{}_{x_2,y_5,x_8}^{y_2,o_5,y_8}
\stmB^{(332|1)}{}_{i_2,x_6,y_9}^{x_2,y_6,o_9}
\\
&\spaceDb
\times
e_{2}^{(y_4)}
\oop^{(x_8)}
\underline{
\rtX^{(x_3)}
e_{(1(32))(32)}^{(y_5)}
e_{1(32)}^{(x_7)}
e_{1}^{(y_9)}
}
\lrlrtllo^{(x_6)}
\rlrtl^{(i_2)}
e_{3}^{(i_1)}
\end{split}
\label{tB rhs1 eq7}
\\
\begin{split}
&=\sum
(-1)^{\phaseI (i_1o_9+x_7+y_5)+\phaseII (x_3x_8+y_5)+\phaseIb x_1o_5+\phaseIIIb x_2x_7}
\\
&\spaceDb
\times
\stmC^{(3|2)}{}_{x_1,y_2,x_3,y_4}^{o_1,o_2,o_3,o_4}
\stmC^{(3|12)}{}_{i_1,y_6,x_7,y_8}^{x_1,o_6,o_7,o_8}
\stmB^{(332|12)}{}_{x_2,y_5,x_8}^{y_2,o_5,y_8}
\stmB^{(332|1)}{}_{i_2,x_6,y_9}^{x_2,y_6,o_9}
\stmC^{(32|1)}{}_{i_3,x_5,i_7,x_9}^{x_3,y_5,x_7,y_9}
\\
&\spaceDb
\times
\underline{
e_{2}^{(y_4)}
\oop^{(x_8)}
e_{1}^{(x_9)}
}
\lrtlo^{(i_7)}
\lrtlllrtlol^{(x_5)}
\rtX^{(i_3)}
\lrlrtllo^{(x_6)}
\rlrtl^{(i_2)}
e_{3}^{(i_1)}
\end{split}
\label{tB rhs1 eq8}
\\
\begin{split}
&=\sum
(-1)^{\phaseI (i_1o_9+x_7+y_5)+\phaseII (x_3x_8+y_5)+\phaseIb x_1o_5+\phaseIIIb x_2x_7}
\stmC^{(3|2)}{}_{x_1,y_2,x_3,y_4}^{o_1,o_2,o_3,o_4}
\stmC^{(3|12)}{}_{i_1,y_6,x_7,y_8}^{x_1,o_6,o_7,o_8}
\\
&\spaceDb
\times
\stmB^{(332|12)}{}_{x_2,y_5,x_8}^{y_2,o_5,y_8}
\stmB^{(332|1)}{}_{i_2,x_6,y_9}^{x_2,y_6,o_9}
\stmC^{(32|1)}{}_{i_3,x_5,i_7,x_9}^{x_3,y_5,x_7,y_9}
\stmB^{(2|1)}{}_{x_4,i_8,i_9}^{y_4,x_8,x_9}
\\
&\spaceDb
\times
e_{1}^{(i_9)}
\toX^{(i_8)}
\underline{
e_{2}^{(x_4)}
\lrtlo^{(i_7)}
}
{\ }
\underline{
\lrtlllrtlol^{(x_5)}
}
{\ }
\underline{
\rtX^{(i_3)}
\lrlrtllo^{(x_6)}
}
\rlrtl^{(i_2)}
e_{3}^{(i_1)}
\end{split}
\label{tB rhs1 eq9}
\\
\begin{split}
&=\sum
(-1)^{\phaseI (i_1o_9+x_7+y_5)+\phaseII (x_3x_8+y_5)+\phaseIII (x_4i_7+x_5)+\phaseIb x_1o_5+\phaseIIb i_3x_6+\phaseIIIb x_2x_7}
\stmC^{(3|2)}{}_{x_1,y_2,x_3,y_4}^{o_1,o_2,o_3,o_4}
\\
&\spaceDb
\times
\stmC^{(3|12)}{}_{i_1,y_6,x_7,y_8}^{x_1,o_6,o_7,o_8}
\stmB^{(332|12)}{}_{x_2,y_5,x_8}^{y_2,o_5,y_8}
\stmB^{(332|1)}{}_{i_2,x_6,y_9}^{x_2,y_6,o_9}
\stmC^{(32|1)}{}_{i_3,x_5,i_7,x_9}^{x_3,y_5,x_7,y_9}
\stmB^{(2|1)}{}_{x_4,i_8,i_9}^{y_4,x_8,x_9}
\\
&\spaceDb
\times
e_{1}^{(i_9)}
\toX^{(i_8)}
\lrtlo^{(i_7)}
\underline{
e_{2}^{(x_4)}
e_{((3(32))1)2}^{(x_5)}
\lrlrtllo^{(x_6)}
}
\rtX^{(i_3)}
\rlrtl^{(i_2)}
e_{3}^{(i_1)}
\end{split}
\label{tB rhs1 eq10}
\\
\begin{split}
&=\sum
(-1)^{\phaseI (i_1o_9+x_7+y_5)+\phaseII (x_3x_8+y_5)+\phaseIII (x_4i_7+x_5)+\phaseIb x_1o_5+\phaseIIb i_3x_6+\phaseIIIb x_2x_7}
\stmC^{(3|2)}{}_{x_1,y_2,x_3,y_4}^{o_1,o_2,o_3,o_4}
\\
&\spaceDb
\times
\stmC^{(3|12)}{}_{i_1,y_6,x_7,y_8}^{x_1,o_6,o_7,o_8}
\stmB^{(332|12)}{}_{x_2,y_5,x_8}^{y_2,o_5,y_8}
\stmB^{(332|1)}{}_{i_2,x_6,y_9}^{x_2,y_6,o_9}
\stmC^{(32|1)}{}_{i_3,x_5,i_7,x_9}^{x_3,y_5,x_7,y_9}
\stmB^{(2|1)}{}_{x_4,i_8,i_9}^{y_4,x_8,x_9}
\\
&\spaceDb
\times
\stmB^{(2|3321)}{}_{i_4,i_5,i_6}^{x_4,x_5,x_6}
e_{1}^{(i_9)}
\toX^{(i_8)}
\lrtlo^{(i_7)}
{\ }
\underline{
\lrlrtllo^{(i_6)}
}
{\ }
\underline{
\tllrlrtllol^{(i_5)}
}
e_{2}^{(i_4)}
\rtX^{(i_3)}
\rlrtl^{(i_2)}
e_{3}^{(i_1)}
\end{split}
\label{tB rhs1 eq11}
\\
\begin{split}
&=\sum
(-1)^{\phaseI (i_1o_9+x_7+y_5)+\phaseII (x_3x_8+y_5)+\phaseIII (x_4i_7+x_5)+\phaseIb x_1o_5+\phaseIIb i_3x_6+\phaseIIIb x_2x_7}
\stmC^{(3|2)}{}_{x_1,y_2,x_3,y_4}^{o_1,o_2,o_3,o_4}
\\
&\spaceDb
\times
\stmC^{(3|12)}{}_{i_1,y_6,x_7,y_8}^{x_1,o_6,o_7,o_8}
\stmB^{(332|12)}{}_{x_2,y_5,x_8}^{y_2,o_5,y_8}
\stmB^{(332|1)}{}_{i_2,x_6,y_9}^{x_2,y_6,o_9}
\stmC^{(32|1)}{}_{i_3,x_5,i_7,x_9}^{x_3,y_5,x_7,y_9}
\stmB^{(2|1)}{}_{x_4,i_8,i_9}^{y_4,x_8,x_9}
\\
&\spaceDb
\times
\stmB^{(2|3321)}{}_{i_4,i_5,i_6}^{x_4,x_5,x_6}
e_{1}^{(i_9)}
\toX^{(i_8)}
\rltol^{(i_7)}
\rlrltoll^{(i_6)}
\tlrlrltolll^{(i_5)}
e_{2}^{(i_4)}
\rtX^{(i_3)}
\rlrtl^{(i_2)}
e_{3}^{(i_1)}
,
\end{split}
,
\label{tB rhs1 eq12}
\end{align}
where summations are taken on $i_k,x_k{\ }(k=1,\cdots,9)$, $y_k{\ }(k=2,4,5,6,8,9)$.
Again, we have put the underlines to the parts to be rewritten.
The details of the above procedure are as follows.
For (\ref{tB rhs1 eq1}), we used (\ref{tB trans mat r3 eq8}).
For (\ref{tB rhs1 eq2}), we used (\ref{tB qroot rel7 chi}) and (\ref{tB qroot rel12 chi}).
For (\ref{tB rhs1 eq3}), we used (\ref{tB trans mat r3 eq10}) and (\ref{tB qroot rel7 chi}).
For (\ref{tB rhs1 eq4}), we used (\ref{tB trans mat r3 eq7}).
For (\ref{tB rhs1 eq5}), we used $[e_1,e_3]=0$, (\ref{tB qroot rel2 chi}), (\ref{tB qroot rel4 chi}), (\ref{tB qroot rel9 chi}) and (\ref{tB qroot rel11 chi}).
For (\ref{tB rhs1 eq6}), we used (\ref{tB trans mat r3 eq5}), (\ref{tB qroot rel2 chi}) and (\ref{tB qroot rel6 chi}).
For (\ref{tB rhs1 eq7}), we used (\ref{tB trans mat r3 eq12}).
For (\ref{tB rhs1 eq8}), we used (\ref{tB trans mat r3 eq1}).
For (\ref{tB rhs1 eq9}), we used (\ref{tB qroot rel2 chi}), (\ref{tB qroot rel3 chi}) and (\ref{tB qroot rel5 chi}), (\ref{tB qroot rel8 chi}) and (\ref{tB qroot rel10 chi}).
For (\ref{tB rhs1 eq10}), we used (\ref{tB trans mat r3 eq3}).
For (\ref{tB rhs1 eq11}), we used (\ref{tB qroot rel3 chi}).
\par
Now, $\{e_{1}^{(i_9)}\toX^{(i_8)}\rltol^{(i_7)}\rlrltoll^{(i_6)}\tlrlrltolll^{(i_5)}e_{2}^{(i_4)}\rtX^{(i_3)}\rlrtl^{(i_2)}e_{3}^{(i_1)}\}$ are linearly independent by Theorem \ref{Yam PBW thm}.
Then, by comparing (\ref{tB lhs1 eq12}) and (\ref{tB rhs1 eq12}), we obtain the following result:
\begin{theorem}\label{tre general thm}
As the identity of transition matrices of quantum superalgebras associated with type B, we have
\begin{align}
\begin{split}
&\sum
(-1)^{\phaseI (o_1i_9+x_5+x_7)+\phaseII (x_3y_8+x_5)+\phaseIII (y_4o_7+y_5)+\phaseIb x_1i_5+\phaseIIb o_3y_6+\phaseIIIb y_2x_7}
\stmB^{(2|1233)}{}_{y_4,y_5,y_6}^{o_4,o_5,o_6}
\stmB^{(2|1)}{}_{x_4,y_8,y_9}^{y_4,o_8,o_9}
\\
&\spaceDb
\times
\stmC^{(23|1)}{}_{x_3,x_5,x_7,x_9}^{o_3,y_5,o_7,y_9}
\stmB^{(233|1)}{}_{y_2,x_6,i_9}^{o_2,y_6,x_9}
\stmB^{(233|21)}{}_{x_2,i_5,x_8}^{y_2,x_5,y_8}
\stmC^{(3|21)}{}_{x_1,i_6,i_7,i_8}^{o_1,x_6,x_7,x_8}
\stmC^{(3|2)}{}_{i_1,i_2,i_3,i_4}^{x_1,x_2,x_3,x_4}
\\
&=\sum
(-1)^{\phaseI (i_1o_9+x_7+y_5)+\phaseII (x_3x_8+y_5)+\phaseIII (x_4i_7+x_5)+\phaseIb x_1o_5+\phaseIIb i_3x_6+\phaseIIIb x_2x_7}
\stmC^{(3|2)}{}_{x_1,y_2,x_3,y_4}^{o_1,o_2,o_3,o_4}
\stmC^{(3|12)}{}_{i_1,y_6,x_7,y_8}^{x_1,o_6,o_7,o_8}
\\
&\spaceDb
\times
\stmB^{(332|12)}{}_{x_2,y_5,x_8}^{y_2,o_5,y_8}
\stmB^{(332|1)}{}_{i_2,x_6,y_9}^{x_2,y_6,o_9}
\stmC^{(32|1)}{}_{i_3,x_5,i_7,x_9}^{x_3,y_5,x_7,y_9}
\stmB^{(2|1)}{}_{x_4,i_8,i_9}^{y_4,x_8,x_9}
\stmB^{(2|3321)}{}_{i_4,i_5,i_6}^{x_4,x_5,x_6}
.
\end{split}
\label{tre general}
\end{align}
where summations are taken on $x_k{\ }(k=1,\cdots,9)$ and $y_k{\ }(k=2,4,5,6,8,9)$.
\end{theorem}
The above equation (\ref{tre general}) generally involves nonlocal sign factors.
Here, we group the Dynkin diagrams given by Table \ref{tB all dynkin r3} into the two families.
In Table \ref{tB all dynkin r3}, we have $\phaseI=\phaseII=\phaseIII=\phaseIb=\phaseIIb=\phaseIIIb=0$ for the following Dynkin diagrams given by (\ref{tB our dynkin r3 1}):
\begin{alignat}{4}
&\mathrm{(I)}&{\ }&\ddthreeBWul{\cc}{\cc}{\cc}{\ee_1-\ee_2}{\ee_2-\ee_3}{\ee_3}
&\qquad
&\mathrm{(II)}&{\ }&\ddthreeBWul{\cct}{\cc}{\cc}{\dd_1-\ee_2}{\ee_2-\ee_3}{\ee_3}
\nonumber
\\
&\mathrm{(III)}&{\ }&\ddthreeBWul{\cct}{\cct}{\cc}{\ee_1-\dd_2}{\dd_2-\ee_3}{\ee_3}
&\qquad
&\mathrm{(IV)}&{\ }&\ddthreeBWul{\cc}{\cc}{\ccb}{\ee_1-\ee_2}{\ee_2-\dd_3}{\dd_3}
\label{tB our dynkin r3 1}
\end{alignat}
\vspace{-1em}
\begin{align}
\mathrm{(V)}{\ }\ddthreeBWul{\cc}{\cct}{\ccb}{\ee_1-\ee_2}{\ee_2-\dd_3}{\dd_3}
\nonumber
\end{align}
where (I) (II) and (IV) are distinguished, in the sense defined in Section \ref{sec 22}.
For each case of (\ref{tB our dynkin r3 1}), (\ref{tre general}) exactly gives the 3D reflection equation.
On the other hand, there are non-trivial nonlocal sign factors for the following Dynkin diagrams given by (\ref{tB our dynkin r3 2}):
\begin{align}
\begin{gathered}
\mathrm{(VI)}{\ }\ddthreeBWul{\cc}{\cct}{\cc}{\dd_1-\dd_2}{\dd_2-\ee_3}{\ee_3}
\qquad
\mathrm{(VII)}{\ }\ddthreeBWul{\cct}{\cc}{\ccb}{\ee_1-\dd_2}{\dd_2-\dd_3}{\dd_3}
\\
\mathrm{(VIII)}{\ }\ddthreeBWul{\cct}{\cct}{\ccb}{\dd_1-\ee_2}{\ee_2-\dd_3}{\dd_3}
\end{gathered}
\label{tB our dynkin r3 2}
\end{align}
where (VI) is distinguished.
Then, as we will see later, the associated equations are the 3D reflection equation up to sign factors.
Hereafter, we specialize Theorem \ref{tre general thm} for each case given in (\ref{tB our dynkin r3 1}) and (\ref{tB our dynkin r3 2}).
\subsubsection{The case (I) $\vcenter{\hbox{\protect\includegraphics{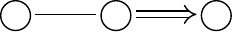}}}$}
In this case, the corresponding symmetrized Cartan matrix is given by
\begin{align}
DA
=
\begin{pmatrix}
2 & -1 & 0 \\
-1 & 2 & -1 \\
0 & -1 & 1
\end{pmatrix}
,
\end{align}
and the corresponding positive roots are given by
\begin{align}
\begin{split}
\prer=\{&\alpha_1,\alpha_2,\alpha_3,\alpha_1+\alpha_2,\alpha_2+\alpha_3,\alpha_2+2\alpha_3,
\\
&\alpha_1+\alpha_2+\alpha_3,\alpha_1+\alpha_2+2\alpha_3,\alpha_1+2\alpha_2+2\alpha_3\}
,
\end{split}
\\
\prir=\{&\}
,
\\
\prar=\{&\}
.
\end{align}
\par
Now, $\stmB^{(x)},\stmC^{(y)}$ defined by (\ref{tB trans mat r3 eq1}) $\sim$ (\ref{tB trans mat r3 eq12}) are specified as follows:
\begin{lemma}\label{tB trans mat r3 c1 lemma}
For the quantum superalgebra associated with \ddthreeB{\cc}{\cc}{\cc}, we have (\ref{tB trans mat r3 eq1}) $\sim$ (\ref{tB trans mat r3 eq12}) where $\stmB^{(x)},\stmC^{(y)}$ are given by
\begin{align}
\begin{split}
\stmB^{(x)}
&=\tdr
\quad
(x=2|1,2|1233,2|3321,233|1,332|1,233|21,332|12)
,\\
\stmC^{(y)}
&=\tdj
\quad
(y=3|2,3|21,3|12,23|1,32|1)
.
\end{split}
\label{tB trans mat r3 c1}
\end{align}
\end{lemma}
\begin{proof}
$\stmB^{(2|1)}$ and $\stmC^{(y)}$ can be obtained in the same way as Lemma \ref{tA trans mat r3 c1 lemma} via the propositions given in Section \ref{tB qroot rels subsec}.
The remaining cases are also obtained almost in the same way, but we have to care the normalization factor $q^{1/2}+q^{-1/2}$ of quantum root vectors, which is given by the begining of Section \ref{tB qroot rels subsec}.
Here, we only present the proof for $\stmB^{(233|1)}$.
Similarly to Lemma \ref{tA trans mat r3 c1 lemma}, by considering (\ref{tB qroot serre 5}) and (\ref{tB qroot serre 6}), $h:\UtAp[3]\to\UtBpn[5]$ defined by $e_1\mapsto e_{1}$, $e_{2}\mapsto e_{(23)3}$ gives an algebra homomorphism.
Also, $d_{\alpha_2+2\alpha_3}=d_{\alpha_2}$ is satisfied where the left hand side is for $\UtBpn[5]$ and the right hand side is for $\UtAp[3]$, so $[m]_{q^{d_{\alpha_2+\alpha_3}}}!=[m]_{q^{d_{\alpha_2}}}$! holds.
Therefore, by applying $h$ on (\ref{tA trans mat def1}) for the case \ddtwoA{\cc}{\cc}, we obtain
\begin{align}
e_{(23)3}^{(a)}e_{1((23)3)}^{(b)}e_{1}^{(c)}
&=\sum_{i,j,k}
(q^{1/2}+q^{-1/2})^{i+j-a-b}
\tdr_{i,j,k}^{a,b,c}
e_{1}^{(k)}e_{((23)3)1}^{(j)}e_{(23)3}^{(i)}
,
\label{tB trans mat r3 eq4 homo}
\end{align}
The 3D R satisfies the weight conservation law: $\tdr_{i,j,k}^{a,b,c}=0$ if $i+j\neq a+b$ or $j+k\neq b+c$.
We then obtain (\ref{tB trans mat r3 eq4}) for $\stmB^{(233|1)}=\tdr$.
\end{proof}
The phase factors given by (\ref{te phase}) and (\ref{tre phase}) are now $\phaseI=\phaseII=\phaseIII=0$ and $\phaseIb=\phaseIIb=\phaseIIIb=0$.
Then, (\ref{tre general}) is specialized as follows:
\begin{align}
\begin{split}
&\sum
\tdr_{y_4,y_5,y_6}^{o_4,o_5,o_6}
\tdr_{x_4,y_8,y_9}^{y_4,o_8,o_9}
\tdj_{x_3,x_5,x_7,x_9}^{o_3,y_5,o_7,y_9}
\tdr_{y_2,x_6,i_9}^{o_2,y_6,x_9}
\tdr_{x_2,i_5,x_8}^{y_2,x_5,y_8}
\tdj_{x_1,i_6,i_7,i_8}^{o_1,x_6,x_7,x_8}
\tdj_{i_1,i_2,i_3,i_4}^{x_1,x_2,x_3,x_4}
\\
&=\sum
\tdj_{x_1,y_2,x_3,y_4}^{o_1,o_2,o_3,o_4}
\tdj_{i_1,y_6,x_7,y_8}^{x_1,o_6,o_7,o_8}
\tdr_{x_2,y_5,x_8}^{y_2,o_5,y_8}
\tdr_{i_2,x_6,y_9}^{x_2,y_6,o_9}
\tdj_{i_3,x_5,i_7,x_9}^{x_3,y_5,x_7,y_9}
\tdr_{x_4,i_8,i_9}^{y_4,x_8,x_9}
\tdr_{i_4,i_5,i_6}^{x_4,x_5,x_6}
,
\end{split}
\end{align}
where all indices are defined on $\mathbb{Z}_{\geq 0}$.
This is exactly the 3D reflection equation (\ref{KO13 tre}):
\begin{align}
\tdr_{456}
\tdr_{489}
\tdj_{3579}
\tdr_{269}
\tdr_{258}
\tdj_{1678}
\tdj_{1234}
=
\tdj_{1234}
\tdj_{1678}
\tdr_{258}
\tdr_{269}
\tdj_{3579}
\tdr_{489}
\tdr_{456}
.
\end{align}
We then get the following result:
\begin{corollary}\label{tB r3 mythm1}
The 3D reflection equation (\ref{KO13 tre}) is characterized as the identity of the transition matrices of the quantum superalgebra associated with \ddthreeB{\cc}{\cc}{\cc}.
\end{corollary}
We note that although Corollary \ref{tB r3 mythm1} is a corollary of the Kuniba-Okado-Yamada theorem\cite{KOY13}, the above calculation gives a direct derivation of the 3D reflection equation (\ref{KO13 tre}) without using any results for quantum coordinate rings.
\subsubsection{The case (II) $\vcenter{\hbox{\protect\includegraphics{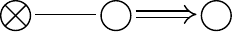}}}$}
In this case, the corresponding symmetrized Cartan matrix is given by
\begin{align}
DA
=
\begin{pmatrix}
0 & -1 & 0 \\
-1 & 2 & -1 \\
0 & -1 & 1
\end{pmatrix}
,
\end{align}
and the corresponding positive roots are given by
\begin{align}
\prer&=\{\alpha_2,\alpha_3,\alpha_2+\alpha_3,\alpha_2+2\alpha_3\}
,
\\
\prir&=\{\alpha_1,\alpha_1+\alpha_2,\alpha_1+\alpha_2+2\alpha_3\}
,
\\
\prar&=\{\alpha_1+\alpha_2+\alpha_3,\alpha_1+2\alpha_2+2\alpha_3\}
.
\end{align}
Similarly to Lemma \ref{tB trans mat r3 c1 lemma}., by using the propositions given in Section \ref{tB qroot rels subsec}., we can show the following lemma:
\begin{lemma}\label{tB trans mat r3 c2 lemma}
For the quantum superalgebra associated with \ddthreeB{\cct}{\cc}{\cc}, we have (\ref{tB trans mat r3 eq1}) $\sim$ (\ref{tB trans mat r3 eq12}) where $\stmB^{(x)},\stmC^{(y)}$ are given by
\begin{align}
\begin{split}
\stmB^{(x)}
&=
\tdm
\quad (x=2|1,2|1233,2|3321,233|1,332|1,233|21,332|12)
,\\
\stmC^{(y)}
&=
\begin{cases}
\tdj
\quad &(y=3|2),
\\
\tdx
\quad &(y=3|21,3|12,23|1,32|1).
\end{cases}
\end{split}
\label{tB trans mat r3 c2}
\end{align}
\end{lemma}
The phase factors given by (\ref{te phase}) and (\ref{tre phase}) are now $\phaseI=\phaseII=\phaseIII=0$ and $\phaseIb=\phaseIIb=\phaseIIIb=0$.
Then, (\ref{tre general}) is specialized as follows:
\begin{align}
\begin{split}
&
\sum
\tdm_{y_4,y_5,y_6}^{o_4,o_5,o_6}
\tdm_{x_4,y_8,y_9}^{y_4,o_8,o_9}
\tdx_{x_3,x_5,x_7,x_9}^{o_3,y_5,o_7,y_9}
\tdm_{y_2,x_6,i_9}^{o_2,y_6,x_9}
\tdm_{x_2,i_5,x_8}^{y_2,x_5,y_8}
\tdx_{x_1,i_6,i_7,i_8}^{o_1,x_6,x_7,x_8}
\tdj_{i_1,i_2,i_3,i_4}^{x_1,x_2,x_3,x_4}
\\
&=
\sum
\tdj_{x_1,y_2,x_3,y_4}^{o_1,o_2,o_3,o_4}
\tdx_{i_1,y_6,x_7,y_8}^{x_1,o_6,o_7,o_8}
\tdm_{x_2,y_5,x_8}^{y_2,o_5,y_8}
\tdm_{i_2,x_6,y_9}^{x_2,y_6,o_9}
\tdx_{i_3,x_5,i_7,x_9}^{x_3,y_5,x_7,y_9}
\tdm_{x_4,i_8,i_9}^{y_4,x_8,x_9}
\tdm_{i_4,i_5,i_6}^{x_4,x_5,x_6}
,
\end{split}
\end{align}
where $o_k,i_k,x_k,y_k\in\{0,1\}{\ }(k=5,6,8,9)$ and the other indices are defined on $\mathbb{Z}_{\geq 0}$.
We then get the following result, which gives a new solution to the 3D reflection equation.
\begin{corollary}\label{tB r3 mythm2}
As the identity of the transition matrices of the quantum superalgebra associated with \ddthreeB{\cct}{\cc}{\cc}, we have the 3D reflection equation given by
\begin{align}
\tdm_{456}
\tdm_{489}
\tdx_{3579}
\tdm_{269}
\tdm_{258}
\tdx_{1678}
\tdj_{1234}
=
\tdj_{1234}
\tdx_{1678}
\tdm_{258}
\tdm_{269}
\tdx_{3579}
\tdm_{489}
\tdm_{456}
.
\label{my tre can crys}
\end{align}
\end{corollary}
\subsubsection{The case (III) $\vcenter{\hbox{\protect\includegraphics{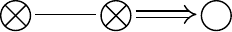}}}$}
In this case, the corresponding symmetrized Cartan matrix is given by
\begin{align}
DA
=
\begin{pmatrix}
0 & 1 & 0 \\
1 & 0 & -1 \\
0 & -1 & 1
\end{pmatrix}
,
\end{align}
and the corresponding positive roots are given by
\begin{align}
\prer&=\{\alpha_3,\alpha_1+\alpha_2,\alpha_1+\alpha_2+\alpha_3,\alpha_1+\alpha_2+2\alpha_3\}
,
\\
\prir&=\{\alpha_1,\alpha_2,\alpha_2+2\alpha_3,\alpha_1+2\alpha_2+2\alpha_3\}
,
\\
\prar&=\{\alpha_2+\alpha_3\}
.
\end{align}
Similarly to Lemma \ref{tB trans mat r3 c1 lemma}, by using the propositions given in Section \ref{tB qroot rels subsec}, we can show the following lemma:
\begin{lemma}\label{tB trans mat r3 c3 lemma}
For the quantum superalgebra associated with \ddthreeB{\cct}{\cct}{\cc}, we have (\ref{tB trans mat r3 eq1}) $\sim$ (\ref{tB trans mat r3 eq12}) where $\stmB^{(x)},\stmC^{(y)}$ are given by
\begin{align}
\begin{split}
\stmB^{(x)}
&=
\begin{cases}
\tdl
\quad &(x=2|1233,2|3321,233|21,332|12),
\\
\tdn(q^{-1})
\quad &(x=2|1,233|1,332|1),
\end{cases}
\\
\stmC^{(y)}
&=
\begin{cases}
\tdj
\quad &(y=3|21,3|12),
\\
\tdx
\quad &(y=3|2),
\\
\tdy(q^{-1})
\quad &(y=23|1,32|1).
\end{cases}
\end{split}
\label{tB trans mat r3 c3}
\end{align}
\end{lemma}
The phase factors given by (\ref{te phase}) and (\ref{tre phase}) are now $\phaseI=\phaseII=\phaseIII=0$ and $\phaseIb=\phaseIIb=\phaseIIIb=0$.
Then, (\ref{tre general}) is specialized as follows:
\begin{align}
\begin{split}
&
\sum
\tdl_{y_4,y_5,y_6}^{o_4,o_5,o_6}
\tdn(q^{-1})_{x_4,y_8,y_9}^{y_4,o_8,o_9}
\tdy(q^{-1})_{x_3,x_5,x_7,x_9}^{o_3,y_5,o_7,y_9}
\tdn(q^{-1})_{y_2,x_6,i_9}^{o_2,y_6,x_9}
\tdl_{x_2,i_5,x_8}^{y_2,x_5,y_8}
\tdj_{x_1,i_6,i_7,i_8}^{o_1,x_6,x_7,x_8}
\tdx_{i_1,i_2,i_3,i_4}^{x_1,x_2,x_3,x_4}
\\
&=
\sum
\tdx_{x_1,y_2,x_3,y_4}^{o_1,o_2,o_3,o_4}
\tdj_{i_1,y_6,x_7,y_8}^{x_1,o_6,o_7,o_8}
\tdl_{x_2,y_5,x_8}^{y_2,o_5,y_8}
\tdn(q^{-1})_{i_2,x_6,y_9}^{x_2,y_6,o_9}
\tdy(q^{-1})_{i_3,x_5,i_7,x_9}^{x_3,y_5,x_7,y_9}
\tdn(q^{-1})_{x_4,i_8,i_9}^{y_4,x_8,x_9}
\tdl_{i_4,i_5,i_6}^{x_4,x_5,x_6}
,
\end{split}
\end{align}
where $o_k,i_k,x_k,y_k\in\{0,1\}{\ }(k=2,4,5,9)$ and the other indices are defined on $\mathbb{Z}_{\geq 0}$.
We then get the following result, which gives a new solution to the 3D reflection equation.
\begin{corollary}\label{tB r3 mythm3}
As the identity of the transition matrices of the quantum superalgebra associated with \ddthreeB{\cct}{\cct}{\cc}, we have the 3D reflection equation given by
\begin{align}
\begin{split}
&\tdl_{456}
\tdn(q^{-1})_{489}
\tdy(q^{-1})_{3579}
\tdn(q^{-1})_{269}
\tdl_{258}
\tdj_{1678}
\tdx_{1234}
\\
&=
\tdx_{1234}
\tdj_{1678}
\tdl_{258}
\tdn(q^{-1})_{269}
\tdy(q^{-1})_{3579}
\tdn(q^{-1})_{489}
\tdl_{456}
.
\end{split}
\label{my tre c3}
\end{align}
\end{corollary}
\subsubsection{The case (IV) $\vcenter{\hbox{\protect\includegraphics{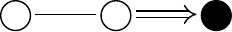}}}$}
In this case, the corresponding symmetrized Cartan matrix is given by
\begin{align}
DA
=
\begin{pmatrix}
2 & -1 & 0 \\
-1 & 2 & -1 \\
0 & -1 & 1
\end{pmatrix}
,
\end{align}
and the corresponding positive roots are given by
\begin{align}
\prer&=\{\alpha_1,\alpha_2,\alpha_1+\alpha_2,\alpha_2+2\alpha_3,\alpha_1+\alpha_2+2\alpha_3,\alpha_1+2\alpha_2+2\alpha_3\}
,
\\
\prir&=\{\}
,
\\
\prar&=\{\alpha_3,\alpha_2+\alpha_3,\alpha_1+\alpha_2+\alpha_3\}
.
\end{align}
Similarly to Lemma \ref{tB trans mat r3 c1 lemma}, by using the propositions given in Section \ref{tB qroot rels subsec}, we can show the following lemma:
\begin{lemma}\label{tB trans mat r3 c4 lemma}
For the quantum superalgebra associated with \ddthreeB{\cc}{\cc}{\ccb}, we have (\ref{tB trans mat r3 eq1}) $\sim$ (\ref{tB trans mat r3 eq12}) where $\stmB^{(x)},\stmC^{(y)}$ are given by
\begin{align}
\begin{split}
\stmB^{(x)}
&=\tdr
\quad
(x=2|1,2|1233,2|3321,233|1,332|1,233|21,332|12)
,\\
\stmC^{(y)}
&=\tdz
\quad
(y=3|2,3|21,3|12,23|1,32|1)
.
\end{split}
\label{tB trans mat r3 c4}
\end{align}
\end{lemma}
The phase factors given by (\ref{te phase}) and (\ref{tre phase}) are now $\phaseI=\phaseII=\phaseIII=0$ and $\phaseIb=\phaseIIb=\phaseIIIb=0$.
Then, (\ref{tre general}) is specialized as follows:
\begin{align}
\begin{split}
&
\sum
\tdr_{y_4,y_5,y_6}^{o_4,o_5,o_6}
\tdr_{x_4,y_8,y_9}^{y_4,o_8,o_9}
\tdz_{x_3,x_5,x_7,x_9}^{o_3,y_5,o_7,y_9}
\tdr_{y_2,x_6,i_9}^{o_2,y_6,x_9}
\tdr_{x_2,i_5,x_8}^{y_2,x_5,y_8}
\tdz_{x_1,i_6,i_7,i_8}^{o_1,x_6,x_7,x_8}
\tdz_{i_1,i_2,i_3,i_4}^{x_1,x_2,x_3,x_4}
\\
&=
\sum
\tdz_{x_1,y_2,x_3,y_4}^{o_1,o_2,o_3,o_4}
\tdz_{i_1,y_6,x_7,y_8}^{x_1,o_6,o_7,o_8}
\tdr_{x_2,y_5,x_8}^{y_2,o_5,y_8}
\tdr_{i_2,x_6,y_9}^{x_2,y_6,o_9}
\tdz_{i_3,x_5,i_7,x_9}^{x_3,y_5,x_7,y_9}
\tdr_{x_4,i_8,i_9}^{y_4,x_8,x_9}
\tdr_{i_4,i_5,i_6}^{x_4,x_5,x_6}
,
\end{split}
\end{align}
where all the indices are defined on $\mathbb{Z}_{\geq 0}$.
We then get the following result, which gives a new solution to the 3D reflection equation.
\begin{corollary}\label{tB r3 mythm4}
As the identity of the transition matrices of the quantum superalgebra associated with \ddthreeB{\cc}{\cc}{\ccb}, we have the 3D reflection equation given by
\begin{align}
\tdr_{456}
\tdr_{489}
\tdz_{3579}
\tdr_{269}
\tdr_{258}
\tdz_{1678}
\tdz_{1234}
=
\tdz_{1234}
\tdz_{1678}
\tdr_{258}
\tdr_{269}
\tdz_{3579}
\tdr_{489}
\tdr_{456}
.
\label{my tre c4}
\end{align}
\end{corollary}
\subsubsection{The case (V) $\vcenter{\hbox{\protect\includegraphics{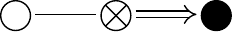}}}$}
In this case, the corresponding symmetrized Cartan matrix is given by
\begin{align}
DA
=
\begin{pmatrix}
2 & -1 & 0 \\
-1 & 0 & 1 \\
0 & 1 & -1
\end{pmatrix}
,
\end{align}
and the corresponding positive roots are given by
\begin{align}
\prer&=\{\alpha_1,\alpha_2+\alpha_3,\alpha_1+2\alpha_2+2\alpha_3,\alpha_1+\alpha_2+\alpha_3\}
,
\\
\prir&=\{\alpha_2,\alpha_1+\alpha_2,\alpha_2+2\alpha_3,\alpha_1+\alpha_2+2\alpha_3\}
,
\\
\prar&=\{\alpha_3\}
.
\end{align}
Similarly to Lemma \ref{tB trans mat r3 c1 lemma}, by using the propositions given in Section \ref{tB qroot rels subsec}, we can show the following lemma:
\begin{lemma}\label{tB trans mat r3 c5 lemma}
For the quantum superalgebra associated with \ddthreeB{\cc}{\cct}{\ccb}, we have (\ref{tB trans mat r3 eq1}) $\sim$ (\ref{tB trans mat r3 eq12}) where $\stmB^{(x)},\stmC^{(y)}$ are given by
\begin{align}
\begin{split}
\stmB^{(x)}
&=
\begin{cases}
\tdl
\quad &(x=2|1,233|1,332|1),
\\
\tdn(q^{-1})
\quad &(x=2|1233,2|3321,233|21,332|12),
\end{cases}
\\
\stmC^{(y)}
&=
\begin{cases}
\tdj
\quad &(y=23|1,32|1),
\\
\tdy(q^{-1})
\quad &(y=3|2,3|21,3|12).
\end{cases}
\end{split}
\label{tB trans mat r3 c5}
\end{align}
\end{lemma}
The phase factors given by (\ref{te phase}) and (\ref{tre phase}) are now $\phaseI=\phaseII=\phaseIII=0$ and $\phaseIb=\phaseIIb=\phaseIIIb=0$.
Then, (\ref{tre general}) is specialized as follows:
\begin{align}
\begin{split}
&
\sum
\tdn(q^{-1})_{y_4,y_5,y_6}^{o_4,o_5,o_6}
\tdl_{x_4,y_8,y_9}^{y_4,o_8,o_9}
\tdj_{x_3,x_5,x_7,x_9}^{o_3,y_5,o_7,y_9}
\tdl_{y_2,x_6,i_9}^{o_2,y_6,x_9}
\tdn(q^{-1})_{x_2,i_5,x_8}^{y_2,x_5,y_8}
\tdy(q^{-1})_{x_1,i_6,i_7,i_8}^{o_1,x_6,x_7,x_8}
\tdy(q^{-1})_{i_1,i_2,i_3,i_4}^{x_1,x_2,x_3,x_4}
\\
&=
\sum
\tdy(q^{-1})_{x_1,y_2,x_3,y_4}^{o_1,o_2,o_3,o_4}
\tdy(q^{-1})_{i_1,y_6,x_7,y_8}^{x_1,o_6,o_7,o_8}
\tdn(q^{-1})_{x_2,y_5,x_8}^{y_2,o_5,y_8}
\tdl_{i_2,x_6,y_9}^{x_2,y_6,o_9}
\tdj_{i_3,x_5,i_7,x_9}^{x_3,y_5,x_7,y_9}
\tdl_{x_4,i_8,i_9}^{y_4,x_8,x_9}
\tdn(q^{-1})_{i_4,i_5,i_6}^{x_4,x_5,x_6}
,
\end{split}
\end{align}
where $o_k,i_k,x_k,y_k\in\{0,1\}{\ }(k=2,4,6,8)$ and the other indices are defined on $\mathbb{Z}_{\geq 0}$.
We then get the following result, which gives a new solution to the 3D reflection equation.
\begin{corollary}\label{tB r3 mythm5}
As the identity of the transition matrices of the quantum superalgebra associated with \ddthreeB{\cc}{\cct}{\ccb}, we have the 3D reflection equation given by
\begin{align}
\begin{split}
&\tdn(q^{-1})_{456}
\tdl_{489}
\tdj_{3579}
\tdl_{269}
\tdn(q^{-1})_{258}
\tdy(q^{-1})_{1678}
\tdy(q^{-1})_{1234}
\\
&=
\tdy(q^{-1})_{1234}
\tdy(q^{-1})_{1678}
\tdn(q^{-1})_{258}
\tdl_{269}
\tdj_{3579}
\tdl_{489}
\tdn(q^{-1})_{456}
.
\end{split}
\label{my tre c5}
\end{align}
\end{corollary}
\subsubsection{The case (VI) $\vcenter{\hbox{\protect\includegraphics{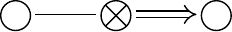}}}$}
In this case, the corresponding symmetrized Cartan matrix is given by
\begin{align}
DA
=
\begin{pmatrix}
2 & -1 & 0 \\
-1 & 0 & 1 \\
0 & 1 & -1
\end{pmatrix}
,
\end{align}
and the corresponding positive roots are given by
\begin{align}
\prer&=\{\alpha_1,\alpha_3,\alpha_1+2\alpha_2+2\alpha_3\}
,
\\
\prir&=\{\alpha_2,\alpha_1+\alpha_2,\alpha_2+2\alpha_3,\alpha_1+\alpha_2+2\alpha_3\}
,
\\
\prar&=\{\alpha_2+\alpha_3,\alpha_1+\alpha_2+\alpha_3\}
.
\end{align}
Similarly to Lemma \ref{tB trans mat r3 c1 lemma}, by using the propositions given in Section \ref{tB qroot rels subsec}, we can show the following lemma:
\begin{lemma}\label{tB trans mat r3 c6 lemma}
For the quantum superalgebra associated with \ddthreeB{\cc}{\cct}{\cc}, we have (\ref{tB trans mat r3 eq1}) $\sim$ (\ref{tB trans mat r3 eq12}) where $\stmB^{(x)},\stmC^{(y)}$ are given by
\begin{align}
\begin{split}
\stmB^{(x)}
&=
\begin{cases}
\tdl
\quad &(x=2|1,233|1,332|1),
\\
\tdn(q^{-1})
\quad &(x=2|1233,2|3321,233|21,332|12),
\end{cases}
\\
\stmC^{(y)}
&=
\begin{cases}
\tdx(q^{-1})
\quad &(y=3|2,3|21,3|12),
\\
\tdz
\quad &(y=23|1,32|1).
\end{cases}
\end{split}
\label{tB trans mat r3 c6}
\end{align}
\end{lemma}
The phase factors given by (\ref{te phase}) and (\ref{tre phase}) are now $\phaseI=\phaseIb=0$ and $\phaseII=\phaseIII=\phaseIIb=\phaseIIIb=1$.
Then, (\ref{tre general}) is specialized as follows:
\begin{align}
\begin{split}
&
\sum
(-1)^{x_3y_8+x_5+y_4o_7+y_5+o_3y_6+y_2x_7}
\\
&\spaceDb
\times
\tdn(q^{-1})_{y_4,y_5,y_6}^{o_4,o_5,o_6}
\tdl_{x_4,y_8,y_9}^{y_4,o_8,o_9}
\tdz_{x_3,x_5,x_7,x_9}^{o_3,y_5,o_7,y_9}
\tdl_{y_2,x_6,i_9}^{o_2,y_6,x_9}
\tdn(q^{-1})_{x_2,i_5,x_8}^{y_2,x_5,y_8}
\tdx(q^{-1})_{x_1,i_6,i_7,i_8}^{o_1,x_6,x_7,x_8}
\tdx(q^{-1})_{i_1,i_2,i_3,i_4}^{x_1,x_2,x_3,x_4}
\\
&=
\sum
(-1)^{x_3x_8+y_5+x_4i_7+x_5+i_3x_6+x_2x_7}
\\
&\spaceDb
\times
\tdx(q^{-1})_{x_1,y_2,x_3,y_4}^{o_1,o_2,o_3,o_4}
\tdx(q^{-1})_{i_1,y_6,x_7,y_8}^{x_1,o_6,o_7,o_8}
\tdn(q^{-1})_{x_2,y_5,x_8}^{y_2,o_5,y_8}
\tdl_{i_2,x_6,y_9}^{x_2,y_6,o_9}
\tdz_{i_3,x_5,i_7,x_9}^{x_3,y_5,x_7,y_9}
\tdl_{x_4,i_8,i_9}^{y_4,x_8,x_9}
\tdn(q^{-1})_{i_4,i_5,i_6}^{x_4,x_5,x_6}
,
\end{split}
\label{my tre c6}
\end{align}
where $o_k,i_k,x_k,y_k\in\{0,1\}{\ }(k=2,4,6,8)$ and the other indices are defined on $\mathbb{Z}_{\geq 0}$.
We then get the following result:
\begin{corollary}\label{tB r3 mythm6}
As the identity of the transition matrices of the quantum superalgebra associated with \ddthreeB{\cc}{\cct}{\cc}, we have the 3D reflection equation up to sign factors given by (\ref{my tre c6}).
\end{corollary}
\subsubsection{The case (VII) $\vcenter{\hbox{\protect\includegraphics{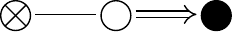}}}$}
In this case, the corresponding symmetrized Cartan matrix is given by
\begin{align}
DA
=
\begin{pmatrix}
0 & -1 & 0 \\
-1 & 2 & -1 \\
0 & -1 & 1
\end{pmatrix}
,
\end{align}
and the corresponding positive roots are given by
\begin{align}
\prer&=\{\alpha_2,\alpha_2+2\alpha_3,\alpha_1+\alpha_2+\alpha_3\}
,
\\
\prir&=\{\alpha_1,\alpha_1+\alpha_2,\alpha_1+\alpha_2+2\alpha_3,\alpha_1+2\alpha_2+2\alpha_3\}
,
\\
\prar&=\{\alpha_3,\alpha_2+\alpha_3\}
.
\end{align}
Similarly to Lemma \ref{tB trans mat r3 c1 lemma}, by using the propositions given in Section \ref{tB qroot rels subsec}, we can show the following lemma:
\begin{lemma}\label{tB trans mat r3 c7 lemma}
For the quantum superalgebra associated with \ddthreeB{\cct}{\cc}{\ccb}, we have (\ref{tB trans mat r3 eq1}) $\sim$ (\ref{tB trans mat r3 eq12}) where $\stmB^{(x)},\stmC^{(y)}$ are given by
\begin{align}
\begin{split}
\stmB^{(x)}
&=
\tdm
\quad (x=2|1,2|1233,2|3321,233|1,332|1,233|21,332|12)
,\\
\stmC^{(y)}
&=
\begin{cases}
\tdz
\quad &(y=3|2),
\\
\tdy
\quad &(y=3|21,3|12,23|1,32|1).
\end{cases}
\end{split}
\label{tB trans mat r3 c7}
\end{align}
\end{lemma}
The phase factors given by (\ref{te phase}) and (\ref{tre phase}) are now $\phaseIII=\phaseIIIb=0$ and $\phaseI=\phaseII=\phaseIb=\phaseIIb=1$.
Then, (\ref{tre general}) is specialized as follows:
\begin{align}
\begin{split}
&
\sum
(-1)^{o_1i_9+x_5+x_7+x_3y_8+x_5+x_1i_5+o_3y_6}
\\
&\spaceDb
\times
\tdm_{y_4,y_5,y_6}^{o_4,o_5,o_6}
\tdm_{x_4,y_8,y_9}^{y_4,o_8,o_9}
\tdy_{x_3,x_5,x_7,x_9}^{o_3,y_5,o_7,y_9}
\tdm_{y_2,x_6,i_9}^{o_2,y_6,x_9}
\tdm_{x_2,i_5,x_8}^{y_2,x_5,y_8}
\tdy_{x_1,i_6,i_7,i_8}^{o_1,x_6,x_7,x_8}
\tdz_{i_1,i_2,i_3,i_4}^{x_1,x_2,x_3,x_4}
\\
&=
\sum
(-1)^{i_1o_9+x_7+y_5+x_3x_8+y_5+x_1o_5+i_3x_6}
\\
&\spaceDb
\times
\tdz_{x_1,y_2,x_3,y_4}^{o_1,o_2,o_3,o_4}
\tdy_{i_1,y_6,x_7,y_8}^{x_1,o_6,o_7,o_8}
\tdm_{x_2,y_5,x_8}^{y_2,o_5,y_8}
\tdm_{i_2,x_6,y_9}^{x_2,y_6,o_9}
\tdy_{i_3,x_5,i_7,x_9}^{x_3,y_5,x_7,y_9}
\tdm_{x_4,i_8,i_9}^{y_4,x_8,x_9}
\tdm_{i_4,i_5,i_6}^{x_4,x_5,x_6}
,
\end{split}
\label{my tre c7}
\end{align}
where $o_k,i_k,x_k,y_k\in\{0,1\}{\ }(k=5,6,8,9)$ and the other indices are defined on $\mathbb{Z}_{\geq 0}$.
We then get the following result:
\begin{corollary}\label{tB r3 mythm7}
As the identity of the transition matrices of the quantum superalgebra associated with \ddthreeB{\cct}{\cc}{\ccb}, we have the 3D reflection equation up to sign factors given by (\ref{my tre c7}).
\end{corollary}
\subsubsection{The case (VIII) $\vcenter{\hbox{\protect\includegraphics{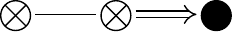}}}$}
In this case, the corresponding symmetrized Cartan matrix is given by
\begin{align}
DA
=
\begin{pmatrix}
0 & 1 & 0 \\
1 & 0 & -1 \\
0 & -1 & 1
\end{pmatrix}
,
\end{align}
and the corresponding positive roots are given by
\begin{align}
\prer&=\{\alpha_1+\alpha_2,\alpha_2+\alpha_3,\alpha_1+\alpha_2+2\alpha_3\}
,
\\
\prir&=\{\alpha_1,\alpha_2,\alpha_2+2\alpha_3,\alpha_1+2\alpha_2+2\alpha_3\}
,
\\
\prar&=\{\alpha_3,\alpha_1+\alpha_2+\alpha_3\}
.
\end{align}
Similarly to Lemma \ref{tB trans mat r3 c1 lemma}, by using the propositions given in Section \ref{tB qroot rels subsec}, we can show the following lemma:
\begin{lemma}\label{tB trans mat r3 c8 lemma}
For the quantum superalgebra associated with \ddthreeB{\cct}{\cct}{\ccb}, we have (\ref{tB trans mat r3 eq1}) $\sim$ (\ref{tB trans mat r3 eq12}) where $\stmB^{(x)},\stmC^{(y)}$ are given by
\begin{align}
\begin{split}
\stmB^{(x)}
&=
\begin{cases}
\tdl
\quad &(x=2|1233,2|3321,233|21,332|12),
\\
\tdn(q^{-1})
\quad &(x=2|1,233|1,332|1),
\end{cases}
\\
\stmC^{(y)}
&=
\begin{cases}
\tdx(q^{-1})
\quad &(y=23|1,32|1),
\\
\tdy
\quad &(y=3|2),
\\
\tdz
\quad &(y=3|21,3|12).
\end{cases}
\end{split}
\label{tB trans mat r3 c8}
\end{align}
\end{lemma}
The phase factors given by (\ref{te phase}) and (\ref{tre phase}) are now $\phaseII=\phaseIIb=0$ and $\phaseI=\phaseIII=\phaseIb=\phaseIIIb=1$.
Then, (\ref{tre general}) is specialized as follows:
\begin{align}
\begin{split}
&
\sum
(-1)^{o_1i_9+x_5+x_7+y_4o_7+y_5+x_1i_5+y_2x_7}
\\
&\spaceDb
\times
\tdl_{y_4,y_5,y_6}^{o_4,o_5,o_6}
\tdn(q^{-1})_{x_4,y_8,y_9}^{y_4,o_8,o_9}
\tdx(q^{-1})_{x_3,x_5,x_7,x_9}^{o_3,y_5,o_7,y_9}
\tdn(q^{-1})_{y_2,x_6,i_9}^{o_2,y_6,x_9}
\tdl_{x_2,i_5,x_8}^{y_2,x_5,y_8}
\tdz_{x_1,i_6,i_7,i_8}^{o_1,x_6,x_7,x_8}
\tdy_{i_1,i_2,i_3,i_4}^{x_1,x_2,x_3,x_4}
\\
&=
\sum
(-1)^{i_1o_9+x_7+y_5+x_4i_7+x_5+x_1o_5+x_2x_7}
\\
&\spaceDb
\times
\tdy_{x_1,y_2,x_3,y_4}^{o_1,o_2,o_3,o_4}
\tdz_{i_1,y_6,x_7,y_8}^{x_1,o_6,o_7,o_8}
\tdl_{x_2,y_5,x_8}^{y_2,o_5,y_8}
\tdn(q^{-1})_{i_2,x_6,y_9}^{x_2,y_6,o_9}
\tdx(q^{-1})_{i_3,x_5,i_7,x_9}^{x_3,y_5,x_7,y_9}
\tdn(q^{-1})_{x_4,i_8,i_9}^{y_4,x_8,x_9}
\tdl_{i_4,i_5,i_6}^{x_4,x_5,x_6}
,
\end{split}
\label{my tre c8}
\end{align}
where $o_k,i_k,x_k,y_k\in\{0,1\}{\ }(k=2,4,5,9)$ and the other indices are defined on $\mathbb{Z}_{\geq 0}$.
We then get the following result:
\begin{corollary}\label{tB r3 mythm8}
As the identity of the transition matrices of the quantum superalgebra associated with \ddthreeB{\cct}{\cct}{\ccb}, we have the 3D reflection equation up to sign factors given by (\ref{my tre c8}).
\end{corollary}
\section{Crystal limit}\label{sec 6}
\subsection{Crystal limit of transition matrices of rank 2}\label{sec 61}
In this section, we consider some transition matrices obtained in Section \ref{sec 4} and \ref{sec 5} at $q=0$, which is known as the \textit{crystal limit}\cite{Kas91}.
First, we note that the crystal limit of transition matrices for non-super cases reproduces so-called transition maps of Lusztig's parametrizations of the canonical basis\cite{Lus90,BZ01}.
For type A and B, we set the crystal limits of the 3D R and 3D J by
\begin{align}
\tdrC_{i,j,k}^{a,b,c}=\lim_{q\to 0}\tdr(q)_{i,j,k}^{a,b,c}
,\quad
\tdjC_{i,j,k,l}^{a,b,c,d}=\lim_{q\to 0}\tdj(q)_{i,j,k,l}^{a,b,c,d}
.
\end{align}
Then, these elements are explicitly given as follows\cite{BZ01}:
\begin{align}
&\tdrC_{i,j,k}^{a,b,c}
=
\delta_{a,i+j-\min(i,k)}
\delta_{b,\min(i,k)}
\delta_{c,j+k-\min(i,k)}
,
\label{3dR crys limit}
\\
&\tdjC_{i,j,k,l}^{a,b,c,d}
=
\delta_{a,i+2j+k-x_1}
\delta_{b,x_1-x_2}
\delta_{c,2x_2-x_1}
\delta_{d,j+k+l-x_2}
,
\label{3dJ crys limit}
\end{align}
where $x_1=\min(i+2\min(j,l),k+2l)$, $x_2=\min(i+\min(j,l),k+l)$.
(\ref{3dR crys limit}) and (\ref{3dJ crys limit}) follow from the fact that diagonal elements of transition matrices from PBW bases to the canonical basis is $1$ and off-diagonal elements are in $q\mathbb{Z}[q]$.
They define the non-trivial bijections on $(\mathbb{Z}_{\geq 0})^{3}$ and $(\mathbb{Z}_{\geq 0})^{4}$, respectively.
There also exists the crystal limit of the tetrahedron equation (\ref{BS06 te}) and the 3D reflection equation themselves, so they gives the \textit{combinatorial} solutions to them.
See also related results given in \cite{KO12}.
\par
Here, we present a super analog of these results.
Let us begin with the case of type A of rank 2.
We set the crystal limits of the 3D L, M and N by
\begin{align}
\tdlC_{i,j,k}^{a,b,c}
&=\lim_{q\to 0}\tdl(q)_{i,j,k}^{a,b,c}
\label{3dL crys setting}
,\\
\tdmC_{i,j,k}^{a,b,c}
&=\lim_{q\to 0}\tdm(q)_{i,j,k}^{a,b,c}
\label{3dM crys setting}
,\\
\tdnC_{i,j,k}^{a,b,c}
&=\lim_{q\to 0}\left(\frac{[b]_{q}!}{[j]_{q}!}\tdn(q)_{i,j,k}^{a,b,c}\right)
\label{3dN crys setting}
.
\end{align}
We note that the normalization change in (\ref{3dN crys setting}) corresponds to use \textit{unnormalized} PBW bases for \ddtwoA{\cct}{\cct}.
This is consistent with earlier observations given in \cite[Section 5.3]{Cla16}.
Then, we have the following results by direct calculations:
\begin{proposition}\label{3dL crys prop}
The crystal limit of the 3D L defines a non-trivial bijection on $\{0,1\}^{2}\times \mathbb{Z}_{\geq 0}$.
The elements are given by
\begin{align}
\tdlC_{0,0,k}^{0,0,c}=\tdlC_{1,1,k}^{1,1,c}=\delta_{k,c}
,\quad
\tdlC_{0,1,k}^{1,0,c}=\delta_{k+1,c}
,\quad
\tdlC_{1,0,0}^{1,0,0}=1
,\quad
\tdlC_{1,0,k}^{0,1,c}=\delta_{k-1,c}
,
\label{3dL mat el crys}
\end{align}
where $\tdlC_{i,j,k}^{a,b,c}=0$ other than (\ref{3dL mat el crys}).
\end{proposition}
\begin{corollary}
The crystal limit of the 3D M defines a non-trivial bijection on $\mathbb{Z}_{\geq 0}\times\{0,1\}^{2}$.
The elements are given by $\tdmC_{i,j,k}^{a,b,c}=\tdlC_{k,j,i}^{c,b,a}$.
\end{corollary}
\begin{proposition}\label{3dN crys prop}
The crystal limit of the 3D N defines a non-trivial bijection on $\{0,1\}\times \mathbb{Z}_{\geq 0}\times \{0,1\}$.
The elements are given by
\begin{align}
\tdnC_{0,j,1}^{0,b,1}=\tdnC_{1,j,0}^{1,b,0}=\delta_{j,b}
,\quad
\tdnC_{0,j,0}^{1,b,1}=\delta_{j-1,b}
,\quad
\tdnC_{0,0,0}^{0,0,0}=1
,\quad
\tdnC_{1,j,1}^{0,b,0}=\delta_{j+1,b}
,
\label{3dN mat el crys}
\end{align}
where $\tdnC_{i,j,k}^{a,b,c}=0$ other than (\ref{3dN mat el crys}).
\end{proposition}
\par
Next, we proceed to the case of type B of rank 2.
We set the crystal limits of the 3D X and 3D Y by
\begin{align}
\tdxC_{i,j,k,l}^{a,b,c,d}
&=\lim_{q\to 0}\left(
\frac{[c]_{q^{-1/2},(-1)}!}{[k]_{q^{-1/2},(-1)}!}\tdx(q)_{i,j,k,l}^{a,b,c,d}\right)
\label{3dX crys setting}
,\\
\tdyC_{i,j,k,l}^{a,b,c,d}
&=\lim_{q\to 0}\left(
\frac{[c]_{q^{-1/2}}!}{[k]_{q^{-1/2}}!}
\tdy(q)_{i,j,k,l}^{a,b,c,d}\right)
\label{3dY crys setting}
.
\end{align}
We note that the normalization changes in (\ref{3dX crys setting}) and (\ref{3dY crys setting}) correspond to use \textit{partially unnormalized} PBW bases for \ddtwoB{\cct}{\cc} and \ddtwoB{\cct}{\ccb}, respectively.
Then, we have the following results by direct calculations:
\begin{proposition}\label{3dX crys prop}
The crystal limit of the 3D X defines a non-trivial bijection on $\mathbb{Z}_{\geq 0}\times \{0,1\}\times\mathbb{Z}_{\geq 0}\times\{0,1\}$.
The matrix elements are given by
\begin{align}
\tdxC_{i,0,k,0}^{a,0,c,0}
=&\delta_{i,a}\delta_{k,c}
\ronri(a\geq 1{\ }\mathrm{or}{\ }a=c=0)
,
\label{3dX mat el 1 crys}
\\
\tdxC_{i,0,k,1}^{a,0,c,0}
=&\delta_{i,a+1}\delta_{k,c-1}
\ronri(a=0)
,\\
\tdxC_{i,0,k,1}^{a,1,c,0}
=&\delta_{i,a+2}\delta_{k,c}
,\\
\tdxC_{i,0,k,0}^{a,0,c,1}
=&\delta_{i,a-1}\delta_{k,c+1}
\ronri(a=1)
,\\
\tdxC_{i,1,k,0}^{a,0,c,1}
=&\delta_{i,a-2}\delta_{k,c}
,\\
\tdxC_{i,0,k,1}^{a,0,c,1}
=&\delta_{i,a}\delta_{k,c}
\ronri(a=0)
,\\
\tdxC_{i,1,k,1}^{a,1,c,1}
=&\delta_{i,a}\delta_{k,c}
,
\label{3dX mat el 8 crys}
\end{align}
where $\tdxC_{i,j,k,l}^{a,b,c,d}=0$ otherwise, and we used $\ronri$ defined by $\ronri(\mathrm{true})=1$ and $\ronri(\mathrm{false})=0$.
\end{proposition}
\begin{proposition}\label{3dY crys prop}
The crystal limit of the 3D Y defines a non-trivial bijection on $\mathbb{Z}_{\geq 0}\times \{0,1\}\times\mathbb{Z}_{\geq 0}\times\{0,1\}$.
The matrix elements are given by
\begin{align}
\tdyC_{i,0,k,0}^{a,0,c,0}
=&\delta_{i,a}\delta_{k,c}
\ronri(a\geq 1{\ }\mathrm{or}{\ }a=c=0)
,
\label{3dY mat el 1 crys}
\\
\tdyC_{i,0,k,1}^{a,0,c,0}
=&\delta_{i,a+1}\delta_{k,c-1}
\ronri(a=0)
,\\
\tdyC_{i,0,k,1}^{a,1,c,0}
=&\delta_{i,a+2}\delta_{k,c}(-1)^{a}
,\\
\tdyC_{i,0,k,0}^{a,0,c,1}
=&\delta_{i,a-1}\delta_{k,c+1}
\ronri(a=1)
,\\
\tdyC_{i,1,k,0}^{a,0,c,1}
=&\delta_{i,a-2}\delta_{k,c}(-1)^{a}
,\\
\tdyC_{i,0,k,1}^{a,0,c,1}
=&\delta_{i,a}\delta_{k,c}
\ronri(a=0)
,\\
\tdyC_{i,1,k,1}^{a,1,c,1}
=&\delta_{i,a}\delta_{k,c}
,
\label{3dY mat el 8 crys}
\end{align}
where $\tdyC_{i,j,k,l}^{a,b,c,d}=0$ otherwise, and we used $\ronri$ defined by $\ronri(\mathrm{true})=1$ and $\ronri(\mathrm{false})=0$.
\end{proposition}
Here, the bijections obtained by the crystal limit of the 3D X and 3D Y are actually same, but have different sign factors.
Actually, $\tdyC_{i,j,k,l}^{a,b,c,d}$ takes not only $0,1$ but also $-1$.
This is a new aspect not arising for non-super cases.
\par
We can also observe the crystal limit of the 3D Z although we do not have an explicit formula for it.
We set the crystal limit of the 3D Z by
\begin{align}
\tdzC_{i,j,k,l}^{a,b,c,d}
&=\lim_{q\to 0}\tdz(q)_{i,j,k.l}^{a,b,c,d}
\label{3dZ crys setting}
.
\end{align}
Supported by computer experiments, we conjecture the crystal limit of the 3D Z also defines a non-trivial bijection on $(\mathbb{Z}_{\geq 0})^{4}$.
For example, the list of all the non-zero elements of $\tdz_{0,1,1,2}^{a,b,c,d}$ is given in Example \ref{3dZ example}.
The crystal limit of them gives $\tdzC_{0,1,1,2}^{a,b,c,d}=\delta_{a,1}\delta_{b,1}\delta_{c,0}\delta_{d,3}$.
We note that the negative factor also appears for the 3D Z.
For example, the following is the list of all the non-zero elements of $\tdz_{2,0,1,0}^{a,b,c,d}$:
\begin{align}
&\tdz_{2,0,1,0}^{2,0,1,0}
=-(1-q^2+q^3)
,\\
&\tdz_{2,0,1,0}^{1,1,0,0}
=q(1-q)^2
,\\
&\tdz_{2,0,1,0}^{3,0,0,1}
=q
.
\end{align}
The crystal limit of them gives $\tdzC_{2,0,1,0}^{a,b,c,d}=-\delta_{a,2}\delta_{b,0}\delta_{c,1}\delta_{d,0}$.
\par
The above results give a super analog of Lusztig's parametrizations of the canonical basis.
To the best of my knowledge, there is no such a study considering transition maps for super cases at present.
We note that there are some earlier results attempting to construct the canonical basis from PBW bases for super cases recently\cite{Cla16,CHW16} although they mainly deal with the distinguished Dynkin diagrams and the canonical basis not depending on reduced expressions has obtained only for {\!\!}\ddgenArightmost{\!\!}.
As we considered for the 3D N, it seems our results are consistent with them.
On the other hand, further investigations should be done for negative factors, which is also remarked in \cite[Remark 7.10]{CHW16}.
\subsection{Crystal limit of transition matrices of rank 3}\label{sec 62}
Here, we remark for the case of rank 3.
In contrast to the case of rank 2, we can \textit{not} take the crystal limit for all cases.
For example, we obtained the tetrahedron equation for \ddthreeA{\cc}{\cct}{\cct} given by (\ref{my te c3}):
\begin{align}
\tdn(q^{-1})_{123}\tdn(q^{-1})_{145}\tdr_{246}\tdl_{356}=\tdl_{356}\tdr_{246}\tdn(q^{-1})_{145}\tdn(q^{-1})_{123}
.
\end{align}
This equation is not consistent with the crystal limits introduced in the previous section because of their \textit{staggered} $q$-dependence of the components.
\par
Among the Dynkin diagrams dealt with Section \ref{sec 43} and \ref{sec 53}, we can take the limit for the following cases:
\begin{gather}
\begin{gathered}
\ddthreeA{\cc}{\cc}{\cct}
,\quad
\ddthreeA{\cct}{\cc}{\cct}
,\\
\ddthreeB{\cct}{\cc}{\cc}
,\quad
\ddthreeB{\cc}{\cc}{\ccb}
,\quad
\ddthreeB{\cct}{\cc}{\ccb}
,
\end{gathered}
\end{gather}
where we omit the non-super cases.
Actually, by setting the normalization factors appropriately, we obtain solutions to the tetrahedron and 3D reflection equations which are compositions of bijections.
Such solutions are often called \textit{set-theoretical} or \textit{combinatorial}, here we use the latter term.
\begin{corollary}
We have the combinatorial solution to the tetrahedron equation given by
\begin{align}
\tdlC_{123}\tdlC_{145}\tdlC_{246}\tdrC_{356}
=\tdrC_{356}\tdlC_{246}\tdlC_{145}\tdlC_{123}
,
\label{BS06 te crys}
\end{align}
and the combinatorial solution up to sign factors given by
\begin{align}
\begin{split}
&\sum
(-1)^{i_1o_6+x_4+x_2x_5}\tdlC_{x_1,x_2,x_3}^{o_1,o_2,o_3}\tdnC_{i_1,x_4,x_5}^{x_1,o_4,o_5}\tdnC_{i_2,i_4,x_6}^{x_2,x_4,o_6}\tdmC_{i_3,i_5,i_6}^{x_3,x_5,x_6}
\\
&=
\sum
(-1)^{o_1i_6+x_4+x_2x_5}\tdmC_{x_3,x_5,x_6}^{o_3,o_5,o_6}\tdnC_{x_2,x_4,i_6}^{o_2,o_4,x_6}\tdnC_{x_1,i_4,i_5}^{o_1,x_4,x_5}\tdlC_{i_1,i_2,i_3}^{x_1,x_2,x_3}
.
\end{split}
\label{my te c5 crys}
\end{align}
where summations are taken on $x_k{\ }(k=1,\cdots,6)$.
\end{corollary}
\begin{corollary}
We have the combinatorial solution to the 3D reflection equation given by
\begin{align}
\tdmC_{456}
\tdmC_{489}
\tdxC_{3579}
\tdmC_{269}
\tdmC_{258}
\tdxC_{1678}
\tdjC_{1234}
=
\tdjC_{1234}
\tdxC_{1678}
\tdmC_{258}
\tdmC_{269}
\tdxC_{3579}
\tdmC_{489}
\tdmC_{456}
.
\label{my tre crys}
\end{align}
\end{corollary}
\begin{conjecture}
We have the combinatorial solution to the 3D reflection equation given by
\begin{align}
\tdrC_{456}
\tdrC_{489}
\tdzC_{3579}
\tdrC_{269}
\tdrC_{258}
\tdzC_{1678}
\tdzC_{1234}
=
\tdzC_{1234}
\tdzC_{1678}
\tdrC_{258}
\tdrC_{269}
\tdzC_{3579}
\tdrC_{489}
\tdrC_{456}
,
\label{my tre crys conj1}
\end{align}
and the combinatorial solution up to sign factors given by
\begin{align}
\begin{split}
&
\sum
(-1)^{o_1i_9+x_5+x_7+x_3y_8+x_5+x_1i_5+o_3y_6}
\\
&\spaceDb
\times
\tdmC_{y_4,y_5,y_6}^{o_4,o_5,o_6}
\tdmC_{x_4,y_8,y_9}^{y_4,o_8,o_9}
\tdyC_{x_3,x_5,x_7,x_9}^{o_3,y_5,o_7,y_9}
\tdmC_{y_2,x_6,i_9}^{o_2,y_6,x_9}
\tdmC_{x_2,i_5,x_8}^{y_2,x_5,y_8}
\tdyC_{x_1,i_6,i_7,i_8}^{o_1,x_6,x_7,x_8}
\tdzC_{i_1,i_2,i_3,i_4}^{x_1,x_2,x_3,x_4}
\\
&=
\sum
(-1)^{i_1o_9+x_7+y_5+x_3x_8+y_5+x_1o_5+i_3x_6}
\\
&\spaceDb
\times
\tdzC_{x_1,y_2,x_3,y_4}^{o_1,o_2,o_3,o_4}
\tdyC_{i_1,y_6,x_7,y_8}^{x_1,o_6,o_7,o_8}
\tdmC_{x_2,y_5,x_8}^{y_2,o_5,y_8}
\tdmC_{i_2,x_6,y_9}^{x_2,y_6,o_9}
\tdyC_{i_3,x_5,i_7,x_9}^{x_3,y_5,x_7,y_9}
\tdmC_{x_4,i_8,i_9}^{y_4,x_8,x_9}
\tdmC_{i_4,i_5,i_6}^{x_4,x_5,x_6}
.
\end{split}
\label{my tre crys conj2}
\end{align}
where summations are taken on $x_k{\ }(k=1,\cdots,9)$ and $y_k{\ }(k=2,4,5,6,8,9)$.
\end{conjecture}
\section{Concluding remarks}
In this paper, we studied transition matrices of PBW bases of the nilpotent subalgebra of quantum superalgebras of type A and B in the case of rank 2 and 3, and obtained explicit formulae for many cases.
By considering the case of rank 3, we obtained the ``mother'' solution to the tetrahedron equation (\ref{te general}) and 3D reflection equation (\ref{tre general}) as identities of transition matrices attributed to compositions of transition matrices of rank 2 in two ways.
Then, we reduced them to special cases and obtained several solutions to the tetrahedron and 3D reflection equation.
The parts of them are summarized as the following table:
\vspace{-0.5em}
\begin{table}[H]
\centering
\begin{tabular}{c}
\begin{minipage}{0.45\hsize}
\centering
\caption*{Type A:}
\vspace{-0.5em}
\begin{tabular}{c|c}\hline
Dynkin diagram & Transition matrix \\\hline
\ddtwoA{\cc}{\cc} & $\tdr${\ }(\ref{3dR mat el}) \\\hline
\ddtwoA{\cc}{\cct} & $\tdl${\ }(\ref{3dL mat el}) \\\hline
\ddtwoA{\cct}{\cc} & $\tdm${\ }(\ref{3dM mat el}) \\\hline
\ddtwoA{\cct}{\cct} & $\tdn${\ }(\ref{3dN mat el}) \\\hline
\end{tabular}
\end{minipage}
\begin{minipage}{0.45\hsize}
\centering
\caption*{Type B:}
\vspace{-0.5em}
\begin{tabular}{c|c}\hline
Dynkin diagram & Transition matrix \\\hline
\ddtwoB{\cc}{\cc} & $\tdj${\ }(\ref{3dJ mat el}) \\\hline
\ddtwoB{\cct}{\cc} & $\tdx${\ }(\ref{3dX mat el 1}) $\sim$ (\ref{3dX mat el 8}) \\\hline
\ddtwoB{\cct}{\ccb} & $\tdy${\ }(\ref{3dY mat el 1}) $\sim$ (\ref{3dY mat el 8}) \\\hline
\ddtwoB{\cc}{\ccb} & $\tdz${\ }(\ref{3dZ mat el}) \\\hline
\end{tabular}
\end{minipage}
\end{tabular}
\end{table}
\vspace{-0.5em}
\noindent
It is important that our proofs exploit higher-order relations for quantum superalgebras given in Section \ref{tA qroot rels subsec} and \ref{tB qroot rels subsec}, and did not use any result for quantum coordinate rings.
\par
For the case of \ddtwoA{\cc}{\cct}, our approach exactly reproduced matrix elements of the 3D L (\ref{3dL mat el}), and the associated tetrahedron equation (\ref{BS06 te}) for \ddthreeA{\cc}{\cc}{\cct}.
That is one of the remarkable result of this paper.
It was known that the 3D L also satisfies another tetrahedron equation (\ref{LLMM te}).
We obtained the similar equation (\ref{my te c4}) for \ddthreeA{\cc}{\cct}{\cc}, but it involves nonlocal sign factors, so we can not write it as a matrix equation at present, in the sense explained in Remark \ref{LLMM remark}.
It is open whether we can attribute (\ref{LLMM te}) to (\ref{my te c4}).
If we can, it is also interesting whether the procedure ``eliminating nonlocal sign factors'' can be applied to other tetrahedron equations (\ref{my te c5}) and (\ref{my te c6}), and the 3D reflection equations (\ref{my tre c6}), (\ref{my tre c7}) and (\ref{my tre c8}) for the case of type B.
\par
We further obtained the new solution to the tetrahedron equation by considering \ddtwoA{\cct}{\cct}, which we call the 3D N (\ref{3dN mat el}).
The associated equation (\ref{my te c3}) was obtained by considering \ddthreeA{\cc}{\cct}{\cct}.
We found matrix elements of the 3D N are related to ones of the 3D L as (\ref{3dl 3dn rel}).
It is interesting whether, in general, transition matrices associated with a pair of Cartan data mapped to each other via odd reflections are attributed to each other or not.
\par
Our framework also can be applied to the case of type B.
We derived the new solutions to the 3D reflection equation (\ref{my tre can crys}), (\ref{my tre c3}), (\ref{my tre c4}) and (\ref{my tre c5}).
As parts of the equations, we introduced the 3D X, Y and Z and obtained explicit formulae for the 3D X and 3D Y given by (\ref{3dX mat el 1}) $\sim$ (\ref{3dX mat el 8}) and (\ref{3dY mat el 1}) $\sim$ (\ref{3dY mat el 8}), respectively.
Although we did not for the 3D Z, we can calculate any matrix elements by recurrence equations like Example \ref{3dZ example}.
We hope to report an explicit formula for the 3D Z in a future publication.
\par
We also discussed the crystal limit of transition matrices for super cases, and obtained a super analog of transition maps of Lusztig's parametrizations of the canonical basis.
We hope that our result gives a new insight into recent studies for a super analog of the canonical basis\cite{Cla16,CHW16}.
It is also an interesting question whether a geometric lifting\cite{BZ01} for them exists or not.
\par
Our result stimulates to challenge whether the Kuniba-Okado-Yamada theorem can be generalized to the case of quantum superalgebras, or not.
This question is quite interesting but needs hard works because there is no theory about irreducible representations of quantum super coordinate rings like Soibelman's theory for the non-super case\cite{Soi92}.
To construct a super version of Soibelman's theory, it seems that the Weyl groupoid plays important roles\cite{HY08}.
More concretely, as we mentioned in Section \ref{sec 12}, \cite{Ser09} seems to give a related result.
We hope to report this issue in a future publication.
\addtocontents{toc}{\setcounter{tocdepth}{1}}
\appendix
\section{Proof of Theorem \ref{my 3dX result}}\label{app 3dX}
By considering (\ref{tm rel}), it is sufficient to prove $\tmt_{i,j,k,l}^{a,b,c,d}=\tdx_{l,k,j,i}^{d,c,b,a}$.
Our proof is motivated by the proof of Proposition 2.{\ }of \cite{KOY13}.
If we obtain an explicit formula for $\appGa_{i,j,k,l}^{a,b,c,d}$ defined by
\begin{align}
e_{1}^{a}e_{21}^{b}e_{2(21)}^{c}e_{2}^{d}
&=\sum_{i,k\in\mathbb{Z}_{\geq 0},j,l\in\{0,1\}}
\appGa_{i,j,k,l}^{a,b,c,d}
e_{2}^{l}e_{(12)2}^{k}e_{12}^{j}e_{1}^{i}
,
\label{tB trans mat def2 app}
\end{align}
we can obtain $\tmt_{i,j,k,l}^{a,b,c,d}$ by
\begin{align}
\tmt_{i,j,k,l}^{a,b,c,d}
=
\frac{[j]_{q^{-1/2},(-1)}![l]_{q^{1/2}}!}{[b]_{q^{-1/2},(-1)}![d]_{q^{1/2}}!}
\appGa_{i,j,k,l}^{a,b,c,d}
.
\end{align}
Then, in order to prove $\tmt_{i,j,k,l}^{a,b,c,d}=\tdx_{l,k,j,i}^{d,c,b,a}$, it is sufficient to show $\appGa_{i,j,k,l}^{a,b,c,d}$ is given by
\begin{alignat}{3}
\appGa_{0,j,0,l}^{0,b,0,d}
&=\delta_{j,b}\delta_{l,d}(1-(1-(-q)^{b})q^d)
,&\quad
\appGa_{0,j,1,l}^{0,b,0,d}
&=-\delta_{j,b-1}\delta_{l,d-1}\frac{q^{1/2}(1-(-q)^{b})(1-q^{d})}{1-q}
\label{3dX mat el 1 app}
,\\
\appGa_{1,j,0,l}^{0,b,0,d}
&=\delta_{j,b-1}\delta_{l,d+1}q^{d}(1-q)(1-(-q)^{b})
,&\quad
\appGa_{1,j,1,l}^{0,b,0,d}
&=-\delta_{j,b-2}\delta_{l,d}q^{d+1/2}(1-(-q)^{b-1})(1-(-q)^{b})
,\\
\appGa_{0,j,0,l}^{0,b,1,d}
&=-\delta_{j,b+1}\delta_{l,d+1}q^{d+1/2}(1-q)
,&\quad
\appGa_{0,j,1,l}^{0,b,1,d}
&=\delta_{j,b}\delta_{l,d}q^{d+1}
,\\
\appGa_{1,j,0,l}^{0,b,1,d}
&=\delta_{j,b}\delta_{l,d+2}q^{d+1/2}(1-q)^2
,&\quad
\appGa_{1,j,1,l}^{0,b,1,d}
&=-\delta_{j,b-1}\delta_{l,d+1}q^{d+1}(1-q)(1-(-q)^{b})
,\\
\appGa_{0,j,0,l}^{1,b,0,d}
&=\delta_{j,b+1}\delta_{l,d-1}\frac{1-q^{d}}{1-q}
,&\quad
\appGa_{0,j,1,l}^{1,b,0,d}
&=\delta_{j,b}\delta_{l,d-2}\frac{q^{-d+3/2}(1-q^{d-1})(1-q^{d})}{(1-q)^2}
,\\
\appGa_{1,j,0,l}^{1,b,0,d}
&=\delta_{j,b}\delta_{l,d}q^d
,&\quad
\appGa_{1,j,1,l}^{1,b,0,d}
&=\delta_{j,b-1}\delta_{l,d-1}\frac{q^{1/2}(1-(-q)^{b})(1-q^{d})}{1-q}
,\\
\appGa_{0,j,0,l}^{1,b,1,d}
&=-\delta_{j,b+2}\delta_{l,d}q^{d+1/2}
,&\quad
\appGa_{0,j,1,l}^{1,b,1,d}
&=-\delta_{j,b+1}\delta_{l,d-1}\frac{q(1-q^{d})}{1-q}
,\\
\appGa_{1,j,0,l}^{1,b,1,d}
&=\delta_{j,b+1}\delta_{l,d+1}q^{d+1/2}(1-q)
,&\quad
\appGa_{1,j,1,l}^{1,b,1,d}
&=\delta_{j,b}\delta_{l,d}(1-(1-(-q)^{b+1})q^{d+1})
.
\label{3dX mat el 8 app}
\end{alignat}
\subsection{Recurrence equations}
Our strategy to derive (\ref{3dX mat el 1 app}) $\sim$ (\ref{3dX mat el 8 app}) is using recurrence equations for $\appGa$.
For simplicity, we write $b_1=e_{1},{\ }b_2=e_{21},{\ }b_3=e_{2(21)},{\ }b_4=e_{2}$ and $F_1^{i,j,k,l}=b_1^{i}b_2^{j}b_3^{k}b_4^{l}$, $F_2^{l,k,j,i}=\chi(F_1^{i,j,k,l})$.
Then, (\ref{tB trans mat def2 app}) with $q\to q^{-1}$ is represented by
\begin{align}
F_1^{a,b,c,d}
&=\sum_{i,k\in\mathbb{Z}_{\geq 0},j,l\in\{0,1\}}
\appGa_{i,j,k,l}^{a,b,c,d}
F_2^{l,k,j,i}
.
\label{tB trans mat def2 app simp}
\end{align}
The elements $b_i$ satisfy the following relations:
\begin{alignat}{3}
&b_2b_1=-qb_1b_2
,&\quad
&b_3b_1=-q^2b_1b_3-q^{1/2}b_2^2
,&\quad
&b_4b_1=qb_1b_4+b_2
,\\
&b_3b_2=-qb_2b_3
,&\quad
&b_4b_2=b_2b_4+(q^{1/2}+q^{-1/2})b_3
,&\quad
&b_4b_3=q^{-1}b_3b_4
.
\end{alignat}
We can easily prove the following relations for $n\in\mathbb{N}$ by induction.
\begin{align}
b_2b_1^n
&=(-q)^nb_1^nb_2
,\\
b_4b_1^n
&=q^nb_1^nb_4+\frac{1-(-1)^n}{2}q^{n-1}b_1^{n-1}b_2
,\\
b_4b_2^n
&=b_2^nb_4+q^{-1/2}(1-(-1)^nq^n)b_2^{n-1}b_3
,\\
b_4b_3^n
&=q^{-n}b_3^nb_4
,\\
b_4^nb_1
&=q^{n}b_1b_4^n+\frac{1-q^n}{1-q}b_2b_4^{n-1}+\frac{(1+q)(1-q^n)(1-q^{n-1})}{(1-q)(1-q^2)}q^{-n+3/2}b_3b_4^{n-2}
,\\
b_3^nb_1
&=(-1)^nq^{2n}b_1b_3^n-\frac{1-(-1)^n}{2}q^{2n-3/2}b_2^2b_3^{n-1}
,\\
b_2^nb_1
&=(-q)^{n}b_1b_2^{n}
,\\
b_3^nb_2
&=(-q)^{n}b_2b_3^n
.
\end{align}
Then, the left multiplication of $b_1,b_2,b_4$ on $F_1^{a,b,c,d}$ and $F_2^{l,k,j,i}$ are given by
\begin{align}
b_1F^{a,b,c,d}_1
=&F^{a+1,b,c,d}_1
,\\
b_2F^{a,b,c,d}_1
=&(-q)^{a}F^{a,b+1,c,d}_1
,\\
b_4F^{a,b,c,d}_1
=&q^{a-c}F^{a,b,c,d+1}_1
+q^{a-1/2}(1-(-1)^bq^b)F^{a,b-1,c+1,d}_1
+\frac{1-(-1)^{a}}{2}q^{a-1}F^{a-1,b+1,c,d}_1
,\\
\begin{split}
b_1F^{l,k,j,i}_2
=&(-1)^{j+k}q^{j+2k+l}F^{l,k,j,i+1}_2
-\frac{1-(-1)^{k}}{2}q^{2k+l-3/2}F^{l,k-1,j+2,i}_2
\\
&+\frac{1-q^{l}}{1-q}(-q)^{k}F^{l-1,k,j+1,i}_2
+\frac{(1+q)(1-q^l)(1-q^{l-1})}{(1-q)(1-q^2)}q^{-l+3/2}F^{l-2,k+1,j,i}_2
,
\end{split}
\\
\begin{split}
b_2F^{l,k,j,i}_2
=&(-1)^{j+k}q^{j+2k+l}(1-q^2)F^{l+1,k,j,i+1}_2
-\frac{1-(-1)^{k}}{2}q^{2k+l-3/2}(1-q^2)F^{l+1,k-1,j+2,i}_2
\\
&+(-q)^{k}(1-(1+q)q^l)F^{l,k,j+1,i}_2
-q^{1/2}\frac{(1+q)(1-q^l)}{1-q}F^{l-1,k+1,j,i}_2
,
\end{split}
\\
b_4F^{l,k,j,i}_2
=&F^{l+1,k,j,i}_2
,
\end{align}
where the left multiplication of $b_2$ on $F_2^{l,k,j,i}$ can be calculated only using the right multiplication of $b_1,b_4$ on $F_1^{i,j,k,l}$ via $b_2F_2^{l,k,j,i}=\chi(F_1^{i,j,k,l}(b_1b_4-qb_4b_1))$.
By considering the left multiplication of $b_1,b_2,b_4$ on (\ref{tB trans mat def2 app simp}), we obtain the following recurrence equations:
\begin{align}
\begin{split}
\appGa_{i,j,k,l}^{a+1,b,c,d}
=&(-1)^{j+k}q^{j+2k+l}\appGa_{i-1,j,k,l}^{a,b,c,d}
-\frac{1-(-1)^{k+1}}{2}q^{2k+l+1/2}\appGa_{i,j-2,k+1,l}^{a,b,c,d}
\\
&+\frac{1-q^{l+1}}{1-q}(-q)^{k}\appGa_{i,j-1,k,l+1}^{a,b,c,d}
+\frac{(1+q)(1-q^{l+2})(1-q^{l+1})}{(1-q)(1-q^2)}q^{-l-1/2}\appGa_{i,j,k-1,l+2}^{a,b,c,d}
,
\end{split}
\label{rel b1}
\\
\begin{split}
\appGa_{i,j,k,l}^{a,b+1,c,d}
=&
(-1)^{j+k+a}q^{j+2k+l-a-1}(1-q^2)\appGa_{i-1,j,k,l-1}^{a,b,c,d}
\\
&
-\frac{1-(-1)^{k+1}}{2}(-q)^{-a}q^{2k+l-1/2}(1-q^2)\appGa_{i,j-2,k+1,l-1}^{a,b,c,d}
\\
&
+(-q)^{k-a}(1-(1+q)q^l)\appGa_{i,j-1,k,l}^{a,b,c,d}
-(-q)^{-a}q^{1/2}\frac{(1+q)(1-q^{l+1})}{1-q}\appGa_{i,j,k-1,l+1}^{a,b,c,d}
,
\end{split}
\label{rel b2}
\\
\appGa_{i,j,k,l}^{a,b,c,d+1}
=&q^{c-a}\appGa_{i,j,k,l-1}^{a,b,c,d}
-q^{c-1/2}(1-(-1)^bq^b)\appGa_{i,j,k,l}^{a,b-1,c+1,d}
-\frac{1-(-1)^{a}}{2}q^{c-1}\appGa_{i,j,k,l}^{a-1,b+1,c,d}
.
\label{rel b4}
\end{align}
\subsection{1-parameter family}
We first construct the 1-parameter family for $\appGa_{i,j,k,l}^{a,b,c,d}$ by using (\ref{rel b2}), which has generic $j,b$ and takes as small as possible $l,d$ for each $i,k,a,c\in\{0,1\}$.
By the same discussion as (\ref{my 3dL result proof2}), we obtain the following weight conservation:
\begin{align}
\appGa_{i,j,k,l}^{a,b,c,d}=0
\quad
(i+j+k\neq a+b+c\quad \mathrm{or}\quad j+2k+l\neq b+2c+d)
.
\label{wc appGa}
\end{align}
\begin{enumerate}[(1)]
\item
For the case $(i,k,a,c)=(0,0,0,0)$, the non-trivial case for (\ref{rel b2}) is $j=b+1,l=d$ by (\ref{wc appGa}).
Then, if we set $d=0$, (\ref{rel b2}) gives
\begin{align}
\appGa_{0,b+1,0,0}^{0,b+1,0,0}
=-q\appGa_{0,b,0,0}^{0,b,0,0}
=(-q)^{b+1}\appGa_{0,0,0,0}^{0,0,0,0}
.
\end{align}
It is easy to verify $\appGa_{0,0,0,0}^{0,0,0,0}=1$, so we obtain
\begin{align}
\appGa_{0,b,0,0}^{0,b,0,0}
=(-q)^{b}
.
\end{align}
Later, we use the case $d=1$.
By setting $d=1$, (\ref{rel b2}) gives
\begin{align}
\appGa_{0,b+1,0,1}^{0,b+1,0,1}
=-q^{1/2}(1-q^2)\appGa_{0,b-1,1,0}^{0,b,0,1}
+(1-q-q^2)\appGa_{0,b,0,1}^{0,b,0,1}
.
\label{1para case1 rel later}
\end{align}
\item
For the case $(i,k,a,c)=(0,1,0,0)$, the non-trivial case for (\ref{rel b2}) is $j=b,l=d-1$ by (\ref{wc appGa}).
Then, if we set $d=1$, (\ref{rel b2}) gives
\begin{align}
\appGa_{0,b,1,0}^{0,b+1,0,1}
=&q^2\appGa_{0,b-1,1,0}^{0,b,0,1}
-q^{1/2}(1+q)\appGa_{0,b,0,1}^{0,b,0,1}
\\
\begin{split}
=&q^2\left[q^2\appGa_{0,b-2,1,0}^{0,b-1,0,1}
-q^{1/2}(1+q)\appGa_{0,b-1,0,1}^{0,b-1,0,1}
\right]
\\
&-q^{1/2}(1+q)\left[
-q^{1/2}(1-q^2)\appGa_{0,b-2,1,0}^{0,b-1,0,1}
+(1-q-q^2)\appGa_{0,b-1,0,1}^{0,b-1,0,1}
\right]
\end{split}
\\
=&q(1+q-q^2)\appGa_{0,b-2,1,0}^{0,b-1,0,1}
-q^{1/2}(1-q^2)\appGa_{0,b-1,0,1}^{0,b-1,0,1}
\\
=&q(1-(-q)^{b-1}-(-q)^{b})\appGa_{0,0,1,0}^{0,1,0,1}
-q^{1/2}(1-(-q)^b)\appGa_{0,1,0,1}^{0,1,0,1}
,
\end{align}
where we use (\ref{1para case1 rel later}).
It is easy to verify $\appGa_{0,0,1,0}^{0,1,0,1}=-q^{1/2}(1+q)$, $\appGa_{0,1,0,1}^{0,1,0,1}=1-q-q^2$, so we obtain
\begin{align}
\appGa_{0,b,1,0}^{0,b+1,0,1}
&=-q^{3/2}(1+q)(1-(-q)^{b-1}-(-q)^{b})-q^{1/2}(1-(-q)^b)(1-q-q^2)
\\
&=-q^{1/2}(1-(-q)^{b+1})
.
\end{align}
\end{enumerate}
Similarly, we can derive the following formulae for other $i,k,a,c\in\{0,1\}$:
\begin{alignat}{2}
\appGa_{1,b,0,1}^{0,b+1,0,0}
&=(1-q)(1-(-q)^{b+1})
,&\quad
\appGa_{1,b,1,0}^{0,b+2,0,0}
&=-q^{1/2}(1-(-q)^{b+1})(1-(-q)^{b+2})
,\\
\appGa_{0,b+1,0,1}^{0,b,1,0}
&=-q^{1/2}(1-q)
,&\quad
\appGa_{0,b,1,0}^{0,b,1,0}
&=q
,\\
\appGa_{1,b,0,2}^{0,b,1,0}
&=q^{1/2}(1-q)^2
,&\quad
\appGa_{1,b,1,1}^{0,b+1,1,0}
&=-q(1-q)(1-(-q)^{b+1})
,\\
\appGa_{0,b+1,0,0}^{1,b,0,1}
&=1
,&\quad
\appGa_{0,b,1,0}^{1,b,0,2}
&=q^{-1/2}(1+q)
,\\
\appGa_{1,b,0,0}^{1,b,0,0}
&=1
,&\quad
\appGa_{1,b,1,0}^{1,b+1,0,1}
&=q^{1/2}(1-(-q)^{b+1})
,\\
\appGa_{0,b+2,0,0}^{1,b,1,0}
&=-q^{1/2}
,&\quad
\appGa_{0,b+1,1,0}^{1,b,1,1}
&=-q
,\\
\appGa_{1,b+1,0,1}^{1,b,1,0}
&=q^{1/2}(1-q)
,&\quad
\appGa_{1,b,1,0}^{1,b,1,0}
&=1-q-(-q)^{b+2}
.
\end{alignat}
\subsection{2-parameter family}
Next, we lift the 1-parameter family to the 2-parameter family by using (\ref{rel b4}) and (\ref{rel b1}), which has generic $j,k,b,d$ for each $i,k,a,c\in\{0,1\}$.
\begin{enumerate}[(I)]
\item
For the case $(a,c)=(0,1)$, (\ref{rel b4}) gives
\begin{align}
\appGa_{i,j,k,l}^{0,b,1,d+1}
=q\appGa_{i,j,k,l-1}^{0,b,1,d}
=q^{\alpha+1}\appGa_{i,j,k,l-\alpha-1}^{0,b,1,d-\alpha}
,
\label{2para case1}
\end{align}
where $\alpha$ is specified below.
\begin{enumerate}[(i)]
\item
For the case $(i,k)=(0,0)$, the non-trivial case for (\ref{2para case1}) is $j=b+1,l=d+2$ by (\ref{wc appGa}).
In that case, (\ref{2para case1}) gives
\begin{align}
\appGa_{0,b+1,0,d+2}^{0,b,1,d+1}
&=q^{d+1}\appGa_{0,b+1,0,1}^{0,b,1,0}
,\\
\therefore
\appGa_{0,b+1,0,d+1}^{0,b,1,d}
&=-q^{d+1/2}(1-q)
.
\label{appGa result1}
\end{align}
\item
For the case $(i,k)=(0,1)$, the non-trivial case for (\ref{2para case1}) is $j=b,l=d+1$ by (\ref{wc appGa}).
In that case, (\ref{2para case1}) gives
\begin{align}
\appGa_{0,b,1,d+1}^{0,b,1,d+1}
&=q^{d+1}\appGa_{0,b,1,0}^{0,b,1,0}
,\\
\therefore
\appGa_{0,b,1,d}^{0,b,1,d}
&=q^{d+1}
.
\label{appGa result2}
\end{align}
\item
For the case $(i,k)=(1,0)$, the non-trivial case for (\ref{2para case1}) is $j=b,l=d+3$ by (\ref{wc appGa}).
In that case, (\ref{2para case1}) gives
\begin{align}
\appGa_{1,b,0,d+3}^{0,b,1,d+1}
&=q^{d+1}\appGa_{1,b,0,2}^{0,b,1,0}
,\\
\therefore
\appGa_{1,b,0,d+2}^{0,b,1,d}
&=q^{d+1/2}(1-q)^2
.
\label{appGa result3}
\end{align}
\item
For the case $(i,k)=(1,1)$, the non-trivial case for (\ref{2para case1}) is $j=b-1,l=d+2$ by (\ref{wc appGa}).
In that case, (\ref{2para case1}) gives
\begin{align}
\appGa_{1,b-1,1,d+2}^{0,b,1,d+1}
&=q^{d+1}\appGa_{1,b-1,1,1}^{0,b,1,0}
,\\
\therefore
\appGa_{1,b,1,d+1}^{0,b+1,1,d}
&=-q^{d+1}(1-q)(1-(-q)^{b+1})
.
\label{appGa result4}
\end{align}
\end{enumerate}
\item\label{ac00}
For the case $(a,c)=(0,0)$, (\ref{rel b4}) gives
\begin{align}
\appGa_{i,j,k,l}^{0,b,0,d+1}
&=\appGa_{i,j,k,l-\alpha-1}^{0,b,0,d-\alpha}
-q^{-1/2}(1-(-1)^bq^b)\frac{1-q^{\alpha+1}}{1-q}\appGa_{i,j,k,l-\alpha}^{0,b-1,1,d-\alpha}
,
\label{2para case2}
\end{align}
where $\alpha$ is specified below.
\begin{enumerate}[(i)]
\item
For the case $(i,k)=(0,0)$, the non-trivial case for (\ref{2para case2}) is $j=b,l=d+1$ by (\ref{wc appGa}).
In that case, (\ref{2para case2}) gives
\begin{align}
\appGa_{0,b,0,d+1}^{0,b,0,d+1}
&=\appGa_{0,b,0,0}^{0,b,0,0}
-q^{-1/2}(1-(-1)^bq^b)\frac{1-q^{d+1}}{1-q}\appGa_{0,b,0,1}^{0,b-1,1,0}
,\\
\therefore
\appGa_{0,b,0,d}^{0,b,0,d}
&=(-q)^{b}
+(1-(-q)^b)(1-q^{d})
\\
&=1-(1-(-q)^{b})q^d
.
\label{appGa result5}
\end{align}
\item
For the case $(i,k)=(0,1)$, the non-trivial case for (\ref{2para case2}) is $j=b-1,l=d$ by (\ref{wc appGa}).
In that case, (\ref{2para case2}) gives
\begin{align}
\appGa_{0,b-1,1,d}^{0,b,0,d+1}
&=\appGa_{0,b-1,1,0}^{0,b,0,1}
-q^{-1/2}(1-(-1)^bq^b)\frac{1-q^{d}}{1-q}\appGa_{0,b-1,1,1}^{0,b-1,1,1}
\\
&=\appGa_{0,b-1,1,0}^{0,b,0,1}
-q^{1/2}(1-(-1)^bq^b)\frac{1-q^{d}}{1-q}\appGa_{0,b-1,1,0}^{0,b-1,1,0}
,\\
\therefore
\appGa_{0,b,1,d}^{0,b+1,0,d+1}
&=-q^{1/2}(1-(-q)^{b+1})
-q^{3/2}(1-(-q)^{b+1})\frac{1-q^{d}}{1-q}
\\
&=
-\frac{q^{1/2}(1-(-q)^{b+1})(1-q^{d+1})}{1-q}
.
\label{appGa result6}
\end{align}
\item
For the case $(i,k)=(1,0)$, the non-trivial case for (\ref{2para case2}) is $j=b-1,l=d+2$ by (\ref{wc appGa}).
In that case, (\ref{2para case2}) gives
\begin{align}
\appGa_{1,b-1,0,d+2}^{0,b,0,d+1}
&=\appGa_{1,b-1,0,1}^{0,b,0,0}
-q^{-1/2}(1-(-1)^bq^b)\frac{1-q^{d+1}}{1-q}\appGa_{1,b-1,0,2}^{0,b-1,1,0}
,\\
\therefore
\appGa_{1,b,0,d+1}^{0,b+1,0,d}
&=(1-q)(1-(-q)^{b+1})
-(1-q)(1-(-q)^{b+1})(1-q^{d})
\\
&=q^d(1-q)(1-(-q)^{b+1})
.
\label{appGa result7}
\end{align}
\item
For the case $(i,k)=(1,1)$, the non-trivial case for (\ref{2para case2}) is $j=b-2,l=d+1$ by (\ref{wc appGa}).
In that case, (\ref{2para case2}) gives
\begin{align}
\appGa_{1,b-2,1,d+1}^{0,b,0,d+1}
=&\appGa_{1,b-2,1,0}^{0,b,0,0}
-q^{-1/2}(1-(-1)^bq^b)\frac{1-q^{d+1}}{1-q}\appGa_{1,b-2,1,1}^{0,b-1,1,0}
,\\
\begin{split}
\therefore
\appGa_{1,b,1,d}^{0,b+2,0,d}
=&-q^{1/2}(1-(-q)^{b+1})(1-(-q)^{b+2})
\\
&+q^{1/2}(1-(-q)^{b+1})(1-(-q)^{b+2})(1-q^{d})
\end{split}
\\
=&-q^{d+1/2}(1-(-q)^{b+1})(1-(-q)^{b+2})
.
\label{appGa result8}
\end{align}
\end{enumerate}
\end{enumerate}
Similarly to (\ref{ac00}), we can derive the following formulae for $(a,c)=(1,1)$:
\begin{alignat}{2}
\appGa_{0,b+2,0,d}^{1,b,1,d}
&=-q^{d+1/2}
,&\quad
\appGa_{0,b+1,1,d}^{1,b,1,d+1}
&=-\frac{q(1-q^{d+1})}{1-q}
,
\label{appGa result9}
\\
\appGa_{1,b+1,0,d+1}^{1,b,1,d}
&=q^{d+1/2}(1-q)
,&\quad
\appGa_{1,b,1,d}^{1,b,1,d}
&=1-(1-(-q)^{b+1})q^{d+1}
.
\label{appGa result10}
\end{alignat}
Finally, we consider the case $(a,c)=(1,0)$.
(\ref{rel b1}) with $(a,c)=(0,0)$ gives
\begin{align}
\begin{split}
\appGa_{i,j,k,l}^{1,b,0,d}
=&(-1)^{j+k}q^{j+2k+l}\appGa_{i-1,j,k,l}^{0,b,0,d}
-\frac{1-(-1)^{k+1}}{2}q^{2k+l+1/2}\appGa_{i,j-2,k+1,l}^{0,b,0,d}
\\
&+\frac{1-q^{l+1}}{1-q}(-q)^{k}\appGa_{i,j-1,k,l+1}^{0,b,0,d}
+\frac{(1+q)(1-q^{l+2})(1-q^{l+1})}{(1-q)(1-q^2)}q^{-l-1/2}\appGa_{i,j,k-1,l+2}^{0,b,0,d}
,
\end{split}
\label{2para case4}
\end{align}
\begin{enumerate}[(i)]
\item
For the case $(i,k)=(0,0)$, the non-trivial case for (\ref{2para case4}) is $j=b+1,l=d-1$ by (\ref{wc appGa}).
In that case, (\ref{2para case4}) gives
\begin{align}
\appGa_{0,b+1,0,d-1}^{1,b,0,d}
=&-q^{d-1/2}\appGa_{0,b-1,1,d-1}^{0,b,0,d}
+\frac{1-q^{d}}{1-q}\appGa_{0,b,0,d}^{0,b,0,d}
\\
=&-q^{d-1/2}
\left(
-\frac{q^{1/2}(1-(-q)^{b})(1-q^d)}{1-q}
\right)
+\frac{1-q^{d}}{1-q}
\left(
1-(1-(-q)^{b})q^{d}
\right)
\\
=&\frac{1-q^d}{1-q}
,\\
\therefore
\appGa_{0,b+1,0,d}^{1,b,0,d+1}
=&\frac{1-q^{d+1}}{1-q}
.
\label{appGa result11}
\end{align}
\item
For the case $(i,k)=(0,1)$, the non-trivial case for (\ref{2para case4}) is $j=b,l=d-2$ by (\ref{wc appGa}).
In that case, (\ref{2para case4}) gives
\begin{align}
\appGa_{0,b,1,d-2}^{1,b,0,d}
=&-q\frac{1-q^{d-1}}{1-q}\appGa_{0,b-1,1,d-1}^{0,b,0,d}
+\frac{(1+q)(1-q^{d-1})(1-q^{d})}{(1-q)(1-q^2)}q^{-d+3/2}\appGa_{0,b,0,d}^{0,b,0,d}
\\
\begin{split}
=&-q\frac{1-q^{d-1}}{1-q}
\left(
-\frac{q^{1/2}(1-(-q)^{b})(1-q^d)}{1-q}
\right)
\\
&+\frac{(1+q)(1-q^{d-1})(1-q^{d})}{(1-q)(1-q^2)}q^{-d+3/2}
\left(
1-(1-(-q)^{b})q^{d}
\right)
\end{split}
\\
=&
\frac{q^{-d+3/2}(1-q^{d-1})(1-q^{d})}{(1-q)^2}
,\\
\therefore
\appGa_{0,b,1,d}^{1,b,0,d+2}
=&\frac{q^{-d-1/2}(1-q^{d+1})(1-q^{d+2})}{(1-q)^2}
.
\label{appGa result12}
\end{align}
\item
For the case $(i,k)=(1,0)$, the non-trivial case for (\ref{2para case4}) is $j=b,l=d$ by (\ref{wc appGa}).
In that case, (\ref{2para case4}) gives
\begin{align}
\appGa_{1,b,0,d}^{1,b,0,d}
=&(-1)^{b}q^{b+d}\appGa_{0,b,0,d}^{0,b,0,d}
-q^{d+1/2}\appGa_{1,b-2,1,d}^{0,b,0,d}
+\frac{1-q^{d+1}}{1-q}\appGa_{1,b-1,0,d+1}^{0,b,0,d}
\\
\begin{split}
=&
(-q)^{b}q^{d}
\left(
1-(1-(-q)^{b})q^{d}
\right)
-q^{d+1/2}
\left(
-q^{d+1/2}(1-(-q)^{b-1})(1-(-q)^{b})
\right)
\\
&+\frac{1-q^{d+1}}{1-q}
\left(
q^{d}(1-q)(1-(-q)^{b})
\right)
\end{split}
\\
=&q^d
.
\label{appGa result13}
\end{align}
\item
For the case $(i,k)=(1,1)$, the non-trivial case for (\ref{2para case4}) is $j=b-1,l=d-1$ by (\ref{wc appGa}).
In that case, (\ref{2para case4}) gives
\begin{align}
\begin{split}
\appGa_{1,b-1,1,d-1}^{1,b,0,d}
=&(-q)^{b}q^{d}\appGa_{0,b-1,1,d-1}^{0,b,0,d}
-q\frac{1-q^{d}}{1-q}\appGa_{1,b-2,1,d}^{0,b,0,d}
\\
&+\frac{(1+q)(1-q^{d})(1-q^{d+1})}{(1-q)(1-q^2)}q^{-d+1/2}\appGa_{1,b-1,0,d+1}^{0,b,0,d}
\end{split}
\\
\begin{split}
=&(-q)^{b}q^{d}
\left(
-\frac{q^{1/2}(1-(-q)^{b})(1-q^d)}{1-q}
\right)
\\
&-q\frac{1-q^{d}}{1-q}
\left(
-q^{d+1/2}(1-(-q)^{b-1})(1-(-q)^{b})
\right)
\\
&+\frac{(1+q)(1-q^{d})(1-q^{d+1})}{(1-q)(1-q^2)}q^{-d+1/2}
\left(
q^{d}(1-q)(1-(-q)^{b})
\right)
\end{split}
\\
=&\frac{q^{1/2}(1-(-q)^{b})(1-q^{d})}{1-q}
,\\
\therefore
\appGa_{1,b,1,d}^{1,b+1,0,d+1}
=&\frac{q^{1/2}(1-(-q)^{b+1})(1-q^{d+1})}{1-q}
.
\label{appGa result14}
\end{align}
\end{enumerate}
Therefore, (\ref{appGa result1}), (\ref{appGa result2}), (\ref{appGa result3}), (\ref{appGa result4}), (\ref{appGa result5}), (\ref{appGa result6}), (\ref{appGa result7}), (\ref{appGa result8}), (\ref{appGa result9}), (\ref{appGa result10}), (\ref{appGa result11}), (\ref{appGa result12}), (\ref{appGa result13}) and (\ref{appGa result14}) exactly correspond to (\ref{3dX mat el 1 app}) $\sim$ ((\ref{3dX mat el 8 app})).
\section{Recurrence equations for the 3D Z}\label{app 3dZ}
In order to calculate matrix elements of the 3D Z, it is sufficient to calculate $\tmt_{i,j,k,l}^{a,b,c,d}$ by the relation (\ref{tm rel}).
If we obtain a formula for $\appGb$ defined by
\begin{align}
e_{1}^{a}e_{21}^{b}e_{2(21)}^{c}e_{2}^{d}
&=\sum_{i,j,k,l\in\mathbb{Z}_{\geq 0}}
\appGb_{i,j,k,l}^{a,b,c,d}
e_{2}^{l}e_{(12)2}^{k}e_{12}^{j}e_{1}^{i}
,
\label{tB trans mat def2 appB}
\end{align}
we can obtain $\tmt_{i,j,k,l}^{a,b,c,d}$ by
\begin{align}
\tmt_{i,j,k,l}^{a,b,c,d}
=
\frac{[i]_{q}![j]_{q^{1/2,(-1)}}![k]_{q}![l]_{q^{1/2},(-1)}!}{[a]_{q}![b]_{q^{1/2},(-1)}![c]_{q}![d]_{q^{1/2},(-1)}!}
\appGb_{i,j,k,l}^{a,b,c,d}
.
\end{align}
Then, it is sufficient to calculate $\appGb_{i,j,k,l}^{a,b,c,d}$.
In this section, we derive recurrence equations for $\appGb$.
By the same discussion as (\ref{my 3dL result proof2}), we obtain the following weight conservation:
\begin{align}
\appGb_{i,j,k,l}^{a,b,c,d}=0
\quad
(i+j+k\neq a+b+c\quad\mathrm{or}\quad j+2k+l\neq b+2c+d)
.
\label{wc appGb}
\end{align}
\par
For simplicity, we write $b_1=e_{1},{\ }b_2=e_{21},{\ }b_3=e_{2(21)},{\ }b_4=e_{2}$ and $F_1^{i,j,k,l}=b_1^{i}b_2^{j}b_3^{k}b_4^{l}$, $F_2^{l,k,j,i}=\chi(F_1^{i,j,k,l})$, where $\chi$ is the anti-algebra automorphism given by (\ref{anti chi}).
Then, (\ref{tB trans mat def2 appB}) is represented by
\begin{align}
F_1^{a,b,c,d}
&=\sum_{i,j,k,l\in\mathbb{Z}_{\geq 0}}
\appGb_{i,j,k,l}^{a,b,c,d}
F_2^{l,k,j,i}
.
\label{tB trans mat def2 appB simp}
\end{align}
The elements $b_i$ satisfy the following relations:
\begin{align}
\begin{alignedat}{3}
&b_2b_1=q^{-1}b_1b_2
,&\quad
&b_3b_1=b_1b_3+q^{-1/2}b_2^2
,&\quad
&b_4b_1=qb_1b_4+b_2
,\\
&b_3b_2=q^{-1}b_2b_3
,&\quad
&b_4b_2=-b_2b_4+(q^{1/2}+q^{-1/2})b_3
,&\quad
&b_4b_3=q^{-1}b_3b_4
.
\end{alignedat}
\end{align}
We can easily prove the following relations for $n\in\mathbb{N}$ by induction.
\begin{align}
b_2b_1^n
=&q^{-n}b_1^nb_2
,\\
b_3b_1^n
=&b_1^nb_3+\frac{1-q^{2n}}{1-q^2}q^{-2n+3/2}b_1^{n-1}b_2^2
,\\
b_3b_2^n
=&q^{-n}b_2^nb_3
,\\
b_4b_1^n
=&q^{n}b_1^{n}b_4+[n]_{q}b_1^{n-1}b_2
,\\
b_4b_2^n
=&(-1)^{n}b_2^nb_4+q^{-n+1/2}(1-(-1)^{n}q^{n})b_2^{n-1}b_3
,\\
b_4b_3^n
=&q^{-n}b_3^nb_4
,\\
\begin{split}
b_4^nb_1
=&q^{n}b_1b_4^n
-(-1)^{n}\frac{1-(-1)^nq^n}{1+q}b_2b_4^{n-1}
\\
&+q^{-n+3/2}\frac{(1-(-1)^nq^{n})(1-(-1)^{n-1}q^{n-1})}{1-q^2}b_3b_4^{n-2}
,
\end{split}
\\
b_3^nb_1
=&b_1b_3^n
+q^{-2n+3/2}\frac{1-q^{2n}}{1-q^2}b_2^2b_3^{n-1}
,\\
b_2^nb_1
=&q^{-n}b_1b_2^n
,\\
b_3^nb_2
=&q^{-n}b_2b_3^n
.
\end{align}
Then, the left multiplication of $b_2,b_3,b_4$ on $F_1^{a,b,c,d}$ are given by
\begin{align}
b_2F_1^{a,b,c,d}
&=q^{-a}F_1^{a,b+1,c,d}
,\\
b_3F_1^{a,b,c,d}
&=q^{-b}F_1^{a,b,c+1,d}
+\frac{1-q^{2a}}{1-q^2}q^{-2a+3/2}F_1^{a-1,b+2,c,d}
,\\
b_4F_1^{a,b,c,d}
&=(-1)^{b}q^{a-c}F_1^{a,b,c,d+1}
+q^{a-b+1/2}(1-(-1)^{b}q^{b})F_1^{a,b-1,c+1,d}
+[a]_{q}F_1^{a-1,b+1,c,d}
,
\end{align}
and the right multiplication of $b_1,b_4$ on $F_1^{i,j,k,l}$ are given by
\begin{align}
\begin{split}
F_1^{i,j,k,l}b_1
=&q^{l-j}F_1^{i+1,j,k,l}
+q^{l-2k+3/2}\frac{1-q^{2k}}{1-q^2}F_1^{i,j+2,k-1,l}
-(-1)^{l}q^{-k}\frac{1-(-1)^{l}q^{l}}{1+q}F_1^{i,j+1,k,l-1}
\\
&+q^{-l+3/2}\frac{(1-(-1)^{l}q^{l})(1-(-1)^{l-1}q^{l-1})}{1-q^2}F_1^{i,j,k+1,l-2}
,
\end{split}
\\
F_1^{i,j,k,l}b_4
=&F_1^{i,j,k,l+1}
.
\end{align}
Calculating $F_1^{i,j,k,l}(b_1b_4-qb_4b_1)$ and $F_1^{i,j,k,l}(b_1b_4^2+(1-q)b_4b_1b_4-qb_4^2b_1)/(q^{1/2}+q^{-1/2})$, then by using $\chi$, we get
\begin{align}
\begin{split}
b_2F_2^{l,k,j,i}
=&
q^{l-j}(1-q^2)F_2^{l+1,k,j,i+1}
+q^{l-2k+3/2}(1-q^{2k})F_2^{l+1,k-1,j+2,i}
\\
&-(-1)^{l}q^{-k}(1-(-q)^{l}(1-q))F_2^{l,k,j+1,i}
+(-1)^{l}q^{1/2}(1-(-q)^{l})F_2^{l-1,k+1,j,i}
,
\end{split}
\\
\begin{split}
b_3F_2^{l,k,j,i}
=&
q^{l-j+1/2}(1-q^2)F_2^{l+2,k,j,i+1}
+q^{l-2k+2}(1-q^{2k})F_2^{l+2,k-1,j+2,i}
\\
&+q^{l-k+1/2}(1-q)F_2^{l+1,k,j+1,i}
-q^{l+1}F_2^{l,k+1,j,i}
,
\end{split}
\\
b_4F_2^{l,k,j,i}
=&F_2^{l+1,k,j,i}
.
\end{align}
By considering the left multiplication of $b_2,b_3,b_4$ on (\ref{tB trans mat def2 appB simp}), we obtain the following recurrence equations:
\begin{align}
\begin{split}
\appGb_{i,j,k,l}^{a,b,c,d}
=&
q^{a}
\left[
q^{l-j-1}(1-q^2)\appGb_{i-1,j,k,l-1}^{a,b-1,c,d}
+q^{l-2k-3/2}(1-q^{2k+2})\appGb_{i,j-2,k+1,l-1}^{a,b-1,c,d}
\right.
\\
&\left.
-(-1)^{l}q^{-k}(1-(-q)^{l}(1-q))\appGb_{i,j-1,k,l}^{a,b-1,c,d}
+(-1)^{l+1}q^{1/2}(1-(-q)^{l+1})\appGb_{i,j,k-1,l+1}^{a,b-1,c,d}
\right]
\label{rel b2 z B}
,
\end{split}
\\
\begin{split}
\appGb_{i,j,k,l}^{a,b,c,d}
=&
q^{b}\left[
q^{l-j-3/2}(1-q^2)\appGb_{i-1,j,k,l-2}^{a,b,c-1,d}
+q^{l-2k-2}(1-q^{2k+2})\appGb_{i,j-2,k+1,l-2}^{a,b,c-1,d}
\right.
\\
&\left.
+q^{l-k-1/2}(1-q)\appGb_{i,j-1,k,l-1}^{a,b,c-1,d}
-q^{l+1}\appGb_{i,j,k-1,l}^{a,b,c-1,d}
-\frac{1-q^{2a}}{1-q^2}q^{-2a+3/2}\appGb_{i,j,k,l}^{a-1,b+2,c-1,d}
\right]
,
\end{split}
\label{rel b3 z B}
\\
\appGb_{i,j,k,l}^{a,b,c,d}
=&(-1)^{b}
q^{c-a}\left[
\appGb_{i,j,k,l-1}^{a,b,c,d-1}
-[a]_{q}\appGb_{i,j,k,l}^{a-1,b+1,c,d-1}
-q^{a-b+1/2}(1-(-1)^{b}q^{b})\appGb_{i,j,k,l}^{a,b-1,c+1,d-1}
\right]
,
\label{rel b4 z B}
\end{align}
which (\ref{rel b2 z B}) holds for $b\geq 1$, (\ref{rel b3 z B}) holds for $c\geq 1$ and (\ref{rel b4 z B}) holds for $d\geq 1$.
\par
We can calculate $\appGb_{i,j,k,l}^{a,b,c,d}$ by using the above reccurence equations (\ref{rel b2 z B}) $\sim$ (\ref{rel b4 z B}) as follows.
First, we can reduce $\appGb_{i,j,k,l}^{a,b,c,d}$ to the case of $d=0$ by using (\ref{rel b4 z B}).
Second, we can reduce $\appGb_{i,j,k,l}^{a,b,c,0}$ to the case of $c=0$ by using (\ref{rel b3 z B}) keeping $d=0$.
Finally, we can reduce $\appGb_{i,j,k,l}^{a,b,0,0}$ to the case of $b=0$ by using (\ref{rel b2 z B}) keeping $c=d=0$.
Then, by considering the weight conservation (\ref{wc appGb}), we find $\appGb_{i,j,k,l}^{i,0,0,0}\neq 0$ for $j=k=l=0$ and $a=i$.
In that case, we can easily obtain $\appGb_{i,0,0,0}^{i,0,0,0}=1$ by (\ref{tB trans mat def2 appB}).
Therefore, we can obtain any matrix elements of $\appGb$ by the above procedure.
See Example \ref{3dZ example}.{\ }for the cases of $(i,j,k,l)=(0,1,1,2)$.
\addtocontents{toc}{\setcounter{tocdepth}{2}}
\end{document}